\documentclass[onecolumn,12pt,letterpaper,unpublished,nopdfoutputerror]{quantumarticle}
\usepackage{bbm,bm,graphicx,mathtools,color,slashed}
\usepackage{amssymb,amsmath,amsthm,amsfonts}
\usepackage[export]{adjustbox}
\usepackage{float}
\usepackage[colorlinks=true, linkcolor=blue, citecolor=blue, urlcolor=blue]{hyperref}
\usepackage[T1]{fontenc}
\usepackage[caption=false]{subfig}
\usepackage[nottoc]{tocbibind}
\usepackage[toc,page]{appendix}
\usepackage{youngtab}
\usepackage{tabularx}
\usepackage{url}

\numberwithin{equation}{section}

\usepackage{tikz}
\usetikzlibrary{shapes}
\usetikzlibrary{quantikz}
\usetikzlibrary{arrows}
\usepackage{rotating}

\usepackage{mathtools}
\usepackage{mathrsfs} 
\usepackage{cleveref}
\usepackage{placeins}

\usepackage{soul} % strikeout

\allowdisplaybreaks

\newcommand{\tp}{\otimes}
\newcommand{\norm}[1]{\left\lVert#1\right\rVert}
\newcommand{\diag}{\mathop{\mathrm{diag}}}
\newtheorem{theorem}{Theorem}
\newtheorem{lemma}[theorem]{Lemma}
\newtheorem{corollary}[theorem]{Corollary}
\newtheorem*{remark}{Remark}

\def\beq{\begin{equation}}
\def\eeq{\end{equation}}
\def\bea{\begin{eqnarray}}
\def\eea{\end{eqnarray}}
\def\barr{\begin{array}}
\def\earr{\end{array}}

\def\N{\mathcal{N}}

\newcommand{\ketbra}[2]{| #1 \rangle\!\langle #2|}

\DeclareMathOperator\erf{erf}
\DeclareMathOperator\erfc{erfc}
\DeclareMathOperator\Var{Var}
\DeclareMathOperator\Cov{Cov}
\newcommand{\Tr}[1]{\operatorname{Tr}\left(#1\right)}

\DeclareMathOperator{\sgn}{sgn}

\DeclareMathOperator{\openone}{\mathbb{1}}

\newcommand{\ad}{\mathrm{ad}}

\newcommand{\acomm}{\alpha_{\mathrm{comm}}}

\newcommand{\abs}[1]{\left\lvert#1\right\rvert}
\newcommand{\polynom}{P}

\def\UoT{Department of Computer Science, University of Toronto, Toronto, ON M5S 2E4, Canada}
\def\MSU{Facility for Rare Isotope Beams, Michigan State University, East Lansing, MI 48824, USA}
\def\Zapata{Zapata Computing Inc., Boston, MA 02110, USA}
\def\PNNL{Pacific Northwest National Laboratory, Richland, WA 99352, USA}

\begin{document}

\title{Improved Accuracy for Trotter Simulations Using Chebyshev Interpolation}

\author{Gumaro Rendon}
\affiliation{\Zapata}
\email{gumaro.rendon@proton.me}

\author{Jacob Watkins}
\affiliation{\MSU}
\email{watki243@msu.edu}

\author{Nathan Wiebe}
\affiliation{\UoT} 
\affiliation{\PNNL}
\email{nawiebe@cs.toronto.edu}

\maketitle

% \twocolumn[
%   \begin{@twocolumnfalse}
\begin{abstract}
Quantum metrology allows for measuring properties of a quantum system at the optimal Heisenberg limit. However, when the relevant quantum states are prepared using digital Hamiltonian simulation, the accrued algorithmic errors will cause deviations from this fundamental limit. In this work, we show how algorithmic errors due to Trotterized time evolution can be mitigated through the use of standard polynomial interpolation techniques. Our approach is to extrapolate to zero Trotter step size, akin to  zero-noise extrapolation techniques for mitigating hardware errors. We perform a rigorous error analysis of the interpolation approach for estimating eigenvalues and time-evolved expectation values, and show that the Heisenberg limit is achieved up to polylogarithmic factors in the error. Our work suggests that accuracies approaching those of state-of-the-art simulation algorithms may be achieved using Trotter and classical resources alone for a number of relevant algorithmic tasks.
\end{abstract}
%   \end{@twocolumnfalse}
% ]

\section{Introduction}
\label{sec:introduction}
Quantum simulation has become, arguably, the most promising application of quantum computing in the near-term~\cite{Lloyd1073,reiher2017elucidating}, with the potential to provide exponential speedups for a host of problems ranging from the electronic structure~\cite{reiher2017elucidating, whitfield2011simulation,lee2021even,von2021quantum} to simulation of scattering dynamics within quantum field theories~\cite{jordan2012quantum,shaw2020quantum, klco20202}.  The central challenge of digital Hamiltonian simulation is, given a fixed Hamiltonian, simulation time, and error tolerance, provide a minimal-length sequence of quantum gates that approximates the unitary dynamics within that error tolerance.

Major strides have been made in the last several years towards this goal.  Several classes of methods, such as Linear Combinations of Unitaries (LCU) methods~\cite{LCU_2012,low2019well,berry2015simulating,low2018hamiltonian,kieferova2019simulating}, qubitization~\cite{Low2019hamiltonian,babbush2018encoding} and Trotter-based methods~\cite{Berry_2006,wiebe2011simulating,childs2021theory,haah2021quantum,hagan2022composite, low2022complexity} have emerged as leading approaches for simulating quantum dynamics.   Unlike the other aforementioned methods, Trotter methods yield a complexity that scales with the commutators of the Hamiltonian terms~\cite{childs2021theory}, which can lead to substantial performance improvements for simulations of local Hamiltonian~\cite{haah2021quantum}.  However, Trotter methods scale super-polynomially worse with the error tolerance than other existing methods~\cite{Low_2017,berry2015simulating}.  This makes such methods inferior to other strategies in cases where high accuracy is required.

It would be interesting if the gap in simulation accuracy between Trotter and other approaches could be bridged in a minimal way, such as classical post processing. Richardson extrapolation has been proposed as a method  to mitigate algorithmic and hardware errors for simulation and linear systems~\cite{PhysRevA.99.012334,Vazquez2022Richardson}. Recently, the approach been demonstrated on quantum hardware ~\cite{vazquez2023well} and yielded improvements compared to Trotter simulation alone. However, there remains a gap in the theoretical understanding of the errors, though one may be able to adapt existing results from the theory of multiproduct formulas.

The present work takes a different approach to the problem, achieving improved accuracy using polynomial interpolation, hereafter referred to simply as "interpolation." In this scheme, we take Trotter simulation data at various time step sizes, then interpolate this data with a polynomial. Our estimator is then the value of this polynomial at the desired ideal of zero step size. By choosing our nodes so that the interpolation is well-conditioned, we show that a Heisenberg-limited $1/\epsilon$ scaling can be attained for estimating observables up to logarithmic factors. This is in contrast to the $1/\epsilon^{1 + 1/p}$ scaling obtained with a single Trotter estimate alone, where $p$ is the order of the Trotter formula. Interestingly, our result also holds in cases where only the lowest order Trotter formula is used, providing the first poly-logarithmic scaling method with the error that uses a constant number of ancillary qubits.  We study this method in two cases. The first case uses phase estimation to extract eigenvalues of the Hamiltonian; whereas the second method uses dynamical simulation and amplitude estimation to learn expectation values of observables.  We further validate these methods numerically using a newly proposed unbiased phase estimation method that we call Gaussian phase estimation.

Similar to the aforementioned Richardson extrapolation, our approach is best thought of as algorithmic error mitigation. The particular error we are mitigating is imperfect Hamiltonian evolution due to the use of product formulas. Rather than being a new approach to eigenvalue or expectation value estimation, our proposed protocol relies on existing methods for, say, phase or amplitude estimation as a subroutine. We will make use of the known performance of these subroutines as we quantify the effectiveness of our full procedure.

The layout of the paper is as follows.  In \Cref{sec:setup}, we go over the basic setup of our work, including assumptions and notation. We review Suzuki-Trotter (ST) formulas and polynomial interpolation, state an important lemma, and present the well-conditioned interpolation formulas that we use to enable our approach. \Cref{sec:PE} contains our analysis and applications of interpolation to eigenvalue estimation. Specifically, \Cref{subsec:perturbPE} takes a perturbative approach to deriving error bounds and complexity, and applies the results to the problem of estimating Trotter error on a quantum computer. \Cref{subsec:Bernstein} approaches the problem through complex analysis and the Bernstein ellipse, and applies the analysis to the newly proposed Gaussian phase estimation. In \Cref{sec:expvals}, we present our main results involving extrapolation of expectation values of time-evolved observables.  This case is conceptually distinct from the case of phase estimation, because here the time evolutions are generally long. In \Cref{sec:numerics} we validate our claims numerically for small instances of transverse Ising models.  Finally, we conclude in \Cref{sec:conclusion} and discuss future avenues of research.

\section{Setup and Notation}
\label{sec:setup}
The algorithms presented in this paper are simply an application of standard polynomial interpolation to data obtained from Trotter simulations. These two pieces are, in many respects, disjoint, and the interpolation can be thought of as a form of classical post processing. Thus, we will begin this section with a review of the quantum part, Trotterization, then move to interpolation. 

\subsection{Trotter}

Consider a Hamiltonian $H$ on a collection of $n$ qubits, decomposed in a specified way into a sum of $m$ terms.
\begin{align}
    H = \sum^{m}_{j=1} H_j
\end{align}
Notice that both the Hamiltonian \emph{and} decomposition are specified, as the decomposition is not unique and can have implications for algorithmic performance. Questions about optimal decompositions are entirely neglected in this work, as well as questions of how one maps a quantum system of interest (such as a molecule or crystal) onto a set of $n$ qubits appropriately.

In our results, as a proxy for the true simulation cost, we will quote the number of exponentials $U_j = e^{-iH_j t}$ used to approximate $e^{-iHt}$. Of course, this is only indicates the true cost when each $U_j$ is computationally cheap. In a number of relevant situations, the Hamiltonian is $k$-local or sparse, and it is possible to identify decompositions such that each $U_j$ has constant cost  and $m$ grows polynomially with respect to the number of qubits $n$. For such Hamiltonians, there exist efficient quantum algorithms to approximate the time evolution operator $U(t) = e^{-i H t}$ to some desired precision $\epsilon,$ many of which rely on Trotter formulas. Standard examples of Trotter formulas include the first-order formula
\begin{align}
S_1 (t) := \prod^{m}_{j=1} e^{-i H_j t},
\end{align}
the second-order symmetric formula
\begin{align}
S_2 (t) :=  e^{-i  H_1 t/2} \dots  e^{-i  H_m t/2} e^{-i H_m t/2} \dots  e^{-i  H_1 t/2},
\end{align}
and more generally, the order $2k$ symmetric ST formula, defined recursively as
\begin{align} \label{eq:symm_S2k}
S_{2k} (t) := [S_{2k-2}(u_k t)]^2 S_{2k-2}\left((1 - 4 u_k) t\right) [S_{2k-2}(u_k t)]^2
\end{align}
for every $k\in \mathbb{Z}_+\setminus\{1\}$, where $u_k := (4-4^{1/(2k-1)})^{-1}$~\cite{Suzuki1991GeneralTO}.  Though many kinds of Trotter formulas exist, the ST formulas~\eqref{eq:symm_S2k} will be our primary tool in this work for several reasons. First, they are symmetric under naive\footnote{Not to be confused with actual time reversal, which is an antilinear operator.} time reversal
\begin{align}
    S_{2k}(-t) = S_{2k}(t)^\dagger
\end{align}
which will allow us to effectively double the number of interpolation points. Second, the order of the formula $2k$ can be taken arbitrarily large, and therefore our results will apply to quite a general class. Finally, the lowest order $k = 1$ symmetric formula is actually practical, having small constant factors, and therefore our bounds will have something to say about actual simulations.

Trotter formulas approximate $U(t)$ only in a neighborhood around $t = 0$. Thus, the standard trick is to divide the time interval $[0, t]$ into $r$ subintervals, such that each interval is sufficiently small that the Trotter approximation is valid, then string together these simulations. For the simple case of a uniform mesh of $r$ subintervals, this becomes
\begin{align}
    S_{2k} (t/r)^r = U(t) + O(t^{2k+1}/r^{2k})
\end{align}
where big $O$ is understood as taking $r$ large. Clearly, we have that $\lim_{r\to \infty} S_{2k}(t/r)^r = U(t)$. 

Rather than thinking about the number of steps $r$ tending to infinity, it is simpler for our subsequent analysis to consider $s = 1/r$ as a "dimensionless step size," and instead think about $s \rightarrow 0$. In terms of $s,$ we define 
\begin{align} \label{eq:s_param_S2k}
    \tilde{U}_s(t):= S_{2k}(st)^{1/s}.
\end{align}
as the approximate evolution operator for $s\neq0$. The discontinuity at $s = 0$ in~\eqref{eq:s_param_S2k} may be filled by the exact evolution $\tilde{U}_0(t) := U(t)$. Though we defined $s$ as a reciprocal integer, definition~\eqref{eq:s_param_S2k} together with $\tilde{U}_0(t)$ suggests an extension to allow $s$ to be real-valued. In fact, the resulting function $\Tilde{U}_s$ is smooth on a neighborhood of $s = 0$, a fact that will allow us to precisely characterize the interpolation error. For our purposes, we will only consider $\abs{s} \leq 1$. When $1/s$ is not an integer, we may implement $\Tilde{U}_s$ using fractional queries~\cite{gilyen2019quantum} by splitting $1/s$ into integer and fractional parts.
\begin{align}
    1/s = r + f
\end{align}
Here, $r = \mathrm{rnd}(1/s) \in\mathbb{Z}$ is $1/s$ rounded to the nearest integer, and $f\in[-1/2,1/2]$. Finally, we note that $\tilde{U}_s$ is an even function of $s,$ which we will make use of to cut the number of interpolation points in half by reflecting across $s = 0.$

Prior work has demonstrated the value of considering the effective Trotter Hamiltonian in the analysis of Trotter formulas~\cite{Yi2022}. This approach is also helps us calculate high order derivatives of $\tilde{U}_s$ as needed for our error bounds. We thus define an effective Hamiltonian
\begin{align} \label{eq:def_eff_Ham}
    \tilde{H}_s &:= \frac{i}{st} \log S_{2k} \left(st\right)
\end{align}
so that
\begin{align}
    \tilde{U}_s(t) &= e^{-i \tilde{H}_s t}.
\end{align}
Note that $\tilde{H}_s$ depends on $t$ as well, though this dependence will be left implicit. For the purposes of bounding the interpolation error, we require a bound on the norm of $\tilde{H}_s$. This is supplied by the following lemma.
\begin{lemma}\label{lem:Heffderiv}
In the notation introduced above, let $s$ be chosen such that 
\begin{align} \label{eq:lem1_assumpt}
    k(5/3)^{k}m\max_{l\in [1,m]} \|H_l\| \abs{s}t \le \pi/20.
\end{align}
Then the following bound on the derivatives of $\tilde{H}_s$ with respect to $s$ holds. 
$$
\|\partial_s^n \tilde{H}_s\| \le 2t^{-1}n^n(e^2k(5/3)^k m\max_{l\in [1,m]} \|H_l\| t)^{n+1}.
$$
\end{lemma}
The proof of this lemma is technical, so it is relegated to Appendix~\ref{app:setup_proofs}. Note that our bounds are uniformly worse for larger $k$, i.e., higher order ST formulas. Assuming that this is not an artifact of our mathematical treatment, this suggests low order formulas are  unconditionally preferred over high order ones for interpolation. Numerical studies could help determine the true impact of higher order formulas on the interpolation procedure.

\subsection{Polynomial Interpolation}
Having introduced ST formulas and stated a needed lemma, we turn our attention to interpolation. Essentially, our goal is to use interpolation to "extrapolate" to the ideal of $s = 0$.\footnote{We occassionally interchange between the terminology "extra-" and "interpolation." We view our method as an extrapolation beyond the data using a numerical technique commonly known as polynomial interpolation.} There are many quantities that we could be interested in extrapolating, including the eigenvalues $\lambda_i (s)$ of the effective Hamiltonian,
\begin{align} \label{eq:setup_eigvals}
    \tilde{H}_s \ket{\lambda_i (s)} = \lambda_i (s) \ket{\lambda_i (s)}
\end{align}
and expectation values of time-evolved observables.
\begin{align} \label{eq:setup_expvals}
\begin{aligned}
    \langle O_s (t) \rangle &= \Tr{\rho O_s (t)} \\
    O_s (t) &:= \tilde{U}_s(t)^\dagger O \tilde{U}_s(t)
\end{aligned}
\end{align}
These will be the primary focus of the present work. No matter what quantity we are interested in, we'll assume it was obtained from ST on a quantum computer for the purposes of subsequent cost analysis. However, in principle our approach should work for any Trotter scheme, not just ST.

While the interpolation is classical and independent of the method in which the data is generated, we will assume a quantum simulation was used when considering the computational cost. We assume all quantum operations are executed perfectly, including the exponentials $\exp (-i H_j t)$ for simulation. Thus, the primary sources of error are the interpolation error and error in the calculation of the data points (e.g. the Hamiltonian energies or expectation values at various points $s_i$). Error in the data points may arise from hardware noise, but even in its absence, a measurement protocol such as phase estimation induces a systematic error that cannot be removed. In this work, we only account for the latter: algorithmic errors. This is not to say that the interpolation method could not be applied to noisy quantum systems, but rather that our cost analysis does not account for it.

We now describe the interpolation framework. Let $f\in C^\infty([-a,a])$ be a smooth, real-valued function of a single variable $s \in [-a,a]$ and suppose we have calculated $f$ for $n$ distinct points $s_1, s_2 \dots s_n \in [-a,a]$. That is, we have data in the form of a set of tuples $D = \{(s_i, f_i)\}_{i=1}^n$, where $f_i = f(s_i)$.  Let $\polynom_{n-1} f$ be the unique $(n-1)$-degree polynomial interpolating $D$, i.e. $P f_{n-1}(s_i) = f_i$ for each $i = 1, \dots, n$. For any $s \in [-a,a]$, standard results in polynomial interpolation~\cite{quarteroni2010numerical} tell us that the signed error is given by 
\begin{align} \label{eq:error_at_zero}
    E_{n-1}(s) := f(s) - \polynom_{n-1} f(s) = \frac{f^{(n)}(\xi)}{n!} \omega_{n} (s)
\end{align}
for some $\xi \in I_s$, where $I_s \subset [-a,a]$ is the smallest interval containing $s$ and the interpolation points $\{s_i\}$. Throughout this work, superscripts such as in $f^{(n)}$ will refer to $n$th-order derivatives. The $n$th degree nodal polynomial $\omega_n (s)$ is defined as the unique monic polynomial with zeros at the interpolation points.
\begin{align}
    \omega_n (s) := \prod_{i=1}^n (s-s_i)
\end{align}

Our estimate for $f(0)$ is $\polynom_{n-1} f(0)$. Since we are interested in $s = 0$, $\omega_n$ becomes a (signed) product of the interpolation points. We can bound the interpolation error $E_n (0)$ in a way that is independent of  the precise value of $\xi$ (which is unknown and difficult to find) by maximizing over $\xi \in I_s$.
\begin{align} \label{eq:bound_at_zero}
    \abs{E_{n-1} (0)} \leq \max_{s \in I_s} \frac{\abs{f^{(n)}(s)}}{n!} \prod_{i=1}^n \abs{s_i}
\end{align}
Much of the technical work in this paper involves finding suitable bounds on the size of the derivatives $f^{(n)}$ for different forms of $f$. For example, in the eigenvalue setting\footnote{Smoothness assumptions can be violated by eigenvalue "level crossings." Our results will require a minimum gap as $s \rightarrow 0$ to this and other technical problems.} of equation~\eqref{eq:setup_eigvals}, $f(s) = \lambda(s)$ while for the expectation values of equation~\eqref{eq:setup_expvals}, $f(s) = \langle O_s(t)\rangle$.

For reasons which we discuss in the next subsection, we choose the Chebyshev nodes on $[-a,a]$ as our interpolation points. 
\begin{align} \label{eq:cheb_node_def}
    s_i = a \cos\left(\frac{2i-1}{2n} \pi\right)
\end{align}
This allows us to specialize our interpolation error in the manner described in the following lemma.

\begin{lemma} \label{lem:Cheb_error}
    Let $s_i$, $i = 1, 2, \dots, n$ be the collection of Chebyshev interpolation points on the interval $[-a,a]$. In the notation above, we have
    \begin{align}
        \abs{E_{n-1}(0)} \leq \max_{s \in [-a,a]} \abs{f^{(n)}(s)} \left(\frac{a}{2n}\right)^n.
    \end{align}
\end{lemma}
\begin{proof}
    For $n$ odd, $s=0$ is one of the interpolation points, so the error is zero and the bound holds automatically. Hereafter, we only consider $n$ even (which will be the case of practical interest). 
    
    Using the generic bound~\eqref{eq:bound_at_zero} with the Chebyshev nodes,
    \begin{align} \label{eq:cheb_bound_at_zero}
    \begin{aligned}
        \abs{E_{n-1} (0)} &\leq \max_{\xi \in [-a,a]} \abs{f^{(n)}(\xi)} \frac{1}{n!} a^n \prod_{i=1}^n \abs{\cos\left(\frac{2i-1}{2n} \pi\right)}.
    \end{aligned}
    \end{align}
    To obtain the lemma, we just need to appropriately bound the product of cosines. Since $n$ is even, $n = 2m$ for some $m \in \mathbb{Z}_+$. Moreover, we have a reflectional symmetry about $m$, in the sense that
    \begin{align}
        \abs{\cos\left(\frac{2i-1}{2n} \pi \right)} = \abs{\cos\left(\frac{2(n-i+1)-1}{2n}\pi\right)}.
    \end{align}
    Hence, we only need to take the product over $i = 1, \dots, m$ and square it.
    \begin{align} \label{eq:halve_points}
        \prod_{i=1}^n \abs{\cos\left(\frac{2i-1}{2n} \pi\right)} = \left(\prod_{i=1}^m \cos\left(\frac{2i-1}{4m}\pi\right)\right)^2
    \end{align}
    To proceed further, let's reindex the remaining product by $i \to m-i+1$. This gives
    \begin{align}
    \begin{aligned}
        \prod_{i=1}^m \cos\left(\frac{2i-1}{4m}\pi\right) &= \prod_{i=1}^m \cos\left(\frac{\pi}{2}-\frac{2i-1}{4m}\pi\right) \\
        &= \prod_{i=1}^m \sin\left(\frac{2i-1}{4m}\right) \\
        &\leq \prod_{i=1}^m 
        \frac{2i-1}{4m}
    \end{aligned}
    \end{align}
    where we used the fact that $\sin(x) \leq x$ for all $x \geq 0$. Factoring out the denominator from the product, the remaining terms become a double factorial.
    \begin{align}
        \prod_{i=1}^m \frac{2i-1}{4m} = \frac{(2m-1)!!}{(4m)^m}
    \end{align}
    The double factorial can be bounded as follows.
    \begin{align}
        (2m-1)!!^2 \leq (2m-1)!! (2m)!! = 2m!
    \end{align}
    so that $(2m-1)!! \leq \sqrt{(2m)!}$. 
    Returning to the original product of equation \eqref{eq:halve_points}, and reintroducing $n = 2m$, the resulting bound is
    \begin{align}
        \prod_{i=1}^n \abs{\cos\left(\frac{2i-1}{2n}\pi\right)} \leq \left(\frac{\sqrt{n!}}{(2n)^{n/2}}\right)^2 = \frac{n!}{(2n)^n}
    \end{align}
    Reinserting this result into the last line of equation~\eqref{eq:cheb_bound_at_zero} gives the bound stated in the lemma.
\end{proof}

Though Chebyshev interpolation enjoys nice mathematical properties, it presents a challenge for Trotter simulation because of the need for noninteger time steps in equation~\eqref{eq:s_param_S2k}. In the face of this, there are several options one could take: restricting to integer time steps, or perform fractional queries using, say the Quantum Singular Value Transformation (QSVT). 

First, consider restricting to integer time steps, by gathering data at the nearest reciprocal integer $1/\tilde{r}$ to the Chebyshev node $s$. For symmetrical interval $[-a,a]$, this distance goes as $O(a^2)$ as $a \rightarrow 0$. From here, one could either (a) take the estimate for $f(1/\tilde{r})$ as the estimate for $f(s)$, accruing some error in the process, or (b) perform the interpolation at the \emph{approximate} Chebyshev nodes given by the collection of points $1/r_i$. Unfortunately, for our purposes, option (a) leads to unacceptable errors of order $O(a)$ in the data, eliminating accuracy gains. As for option (b), it is possible to use robustness results on Chebyshev interpolation~\cite{Piazzon2018StabilityIF} to argue that almost-Chebyshev nodes should be almost as well-conditioned. Again, however, we find that our scaling of the number of nodes is such that the node displacements must be quite small, leading again to poor scaling.

Because of this, for most of this work we choose to invoke access to fractional queries using the QSVT~\cite{gilyen2019quantum}. A notable exception to this is in application to Gaussian phase estimation in~\Cref{subsubsec:gqpea}, where the error in the fractional part can be mitigated in that context. While fractional queries increase the overhead compared to Trotter alone, this overhead is a constant.

\subsection{Stability Analysis and Interpolation}

Polynomial interpolation is a valuable numerical tool, but some implementations can lead to numerical instability~\cite{DECAMARGO2020112369}. However, the situation is not as bad as often presented in textbooks~\cite{trefethen2011six}. While linear algebraic approaches involving Vandermonde matrices suffer instability for high degree polynomials~\cite{Gautschi1990stable}, methods such as barycentric formulas are provably stable with respect to floating point arithmetic~\cite{higham2004numerical}. 

A particularly important consideration is the choice of interpolation nodes. It is well known that equally spaced nodes can lead to the Runge phenomenon: rapid oscillations near the ends of the interval that grow with polynomial degree~\cite{quarteroni2010numerical}. These oscillations can be overcome with a superior choice of nodes, such as the zeros of the Chebyshev polynomials. Interpolations done with this set of nodes are guaranteed to converge to functions that are Lipschitz continuous as $n\rightarrow \infty$. Moreover, they are well-conditioned in the sense of small errors in the data values. Finally, because they anti-cluster around $s = 0$, they are relatively cheap to compute with Trotter formulas. In this work, we will always interpolate at the $n$th-degree Chebyshev nodes, or approximations thereof, on a symmetric interval $[-a,a]$ about the origin, defined in~\eqref{eq:cheb_node_def}. We choose even $n$ so as to avoid the origin (which has infinite cost to compute), and also utilize the reflectional symmetry of $f(s)$. 

To compute the interpolant $P_{n-1} f$ linear algebraically, we overcome the limitations of the standard Vandemonde approach by expanding in terms of orthonormal Chebyshev polynomials rather than monomials $x^j$. 
\begin{align}\label{eq:O_expans_opt}
    \polynom_{n-1} f(s) = \sum_{j=0}^{n-1} c_{j} p_j (s).
\end{align}
Here, $p_j$ is defined by
\begin{align} \label{eq:cheb_orthonorm_def}
     p_j(s) :=
    \begin{cases}
    \sqrt{\frac1n}T_0(s), &j=0 \\
    \sqrt{\frac2n}T_j(s), &j=1,2,\dots \\
    \end{cases} 
\end{align}
where $T_j$ is the standard $j$th Chebyshev polynomial.
\begin{align}
    T_j(x) := \cos (j \cos^{-1} x)
\end{align}
By orthonormality, we are referring to the condition~\cite{mason2002chebyshev}
\begin{align}
    \sum_{k=1}^n p_i (s_k) p_j (s_k) = \delta_{ij}
\end{align}
for all $0\leq i, j < n$, with $s_k$ being the zeros of $T_n$ given in~\eqref{eq:cheb_node_def}. This immediately implies the matrix 
\begin{align}
 \mathbf{V} :=
\begin{pmatrix}
 p_0(s_1)   & p_1(s_1)   & \dots  & p_{n-1}(s_1) \\
 p_0(s_2)   & p_1(s_2)   & \dots  & p_{n-1}(s_2) \\
 \vdots     & \vdots     & \ddots & \vdots  \\
 p_0(s_{n}) & p_1(s_{n}) & \dots  & p_{n-1}(s_{n})
\end{pmatrix}
\end{align}
is orthogonal, and therefore has condition number $\kappa(\mathbf{V}) := \norm{\mathbf{V}}\norm{\mathbf{V}^{-1}}$ equal to one. This is the source of well-conditioning in our approach. The coefficients $c = (c_0, c_1, \dots , c_{n-1})$ in equation~\eqref{eq:O_expans_opt} satisfy
\begin{align} \label{eq:cheb_sys}
    y = \mathbf{V} c
\end{align}
for the vector of values $y = (f(s_1), f(s_2), \dots, f(s_n)),$ since $\polynom_{n-1} f$ is an interpolant. Hence, $c = \mathbf{V}^T y$ gives the vector of coefficients. 

We now develop our argument for well-conditioning. Unless otherwise subscripted, all logarithms are natural.

\begin{lemma}\label{lem:Trotter_step_bound}
    Let $s_1, s_2, \dots, s_n$ be the standard Chebyshev nodes on $[-a,a]$~\eqref{eq:cheb_node_def} with $n$ even. Then the nodes satisfy
    
    $$
    \sum_{k=1}^n \frac{1}{\abs{s_k}} \leq \frac{4 n}{\pi a}\left(\gamma + \log(2n+2)\right),
    $$
    where $\gamma\approx 0.577$ is the Euler-Mascheroni constant.
\end{lemma}

\begin{proof}
    We focus on the case $a = 1$, since the general result follows by a simple rescaling. Because sine and cosine are phase shifted by $\pi/2$,
    \begin{align}\label{eq:reciprocal_Cheb_sum}
    \sum_{k=1}^n \frac{1}{\abs{s_k}}= \sum^n_{k=1} \frac1{\left\vert \cos \left(\frac{2k-1}{2 n}\pi\right) \right\vert}=\sum^n_{k=1} \frac1{\left\vert\sin \left(\frac{n-2k+1}{2 n}\pi\right)\right\vert}.
    \end{align}
    Taking advantage of the symmetry about $s = 0,$
    \begin{align}
        \sum^n_{k=1} \frac1{\left\lvert\sin \left(\frac{n-2k+1}{2 n}\pi\right)\right\rvert} &= 2\sum^{n/2}_{k=1} \frac1{\sin \left(\frac{2k-1}{2 n}\pi\right)} .
    \end{align}
    Next, we use the lower bound
    \begin{align}
    \sin x \geq x/2 \quad \left( 0 \leq x \leq \pi/2 \right)
    \end{align}
    in order to bound the terms of the sum.
    \begin{align} \label{eq:sin_sum}
    \begin{aligned}
        2\sum^{n/2}_{k=1} \frac1{\sin \left(\frac{2k-1}{2 n}\pi\right)} &\leq \frac{8 n}{\pi} \sum^{n/2}_{k=1} \frac{1}{2k-1} \\
        &= \frac{8 n}{\pi} \left( H_{n} - \frac{1}{2} H_{n/2} \right) \\
    \end{aligned}
    \end{align}
    Here, $H_n$ denotes the $n$th harmonic number. From the relation $H_n = \gamma + \psi(n+1)$, where $\psi$ is the digamma function,
    \begin{align}
        H_{n} - \frac{1}{2} H_{n/2} = \gamma/2 + \psi(n+1) - \frac12 \psi(n/2 + 1).
    \end{align}
    Moreover, since $\psi(x) \in (\log \left(x - \frac12\right), \log x)$ for any $x > 1/2$, this is upper bounded by
    \begin{align}
        H_{n} - \frac{1}{2} H_{n/2} < \gamma/2 + \log(n+1) - \frac12 \log(\frac{n+1}{2}) = \frac{\gamma + \log 2}{2} + \frac{\log (n+1)}{2}.
    \end{align}
    Reinserting this into~\eqref{eq:sin_sum}, one obtains the bound
    \begin{align}
    \begin{aligned}
        \sum_{k=1}^n \frac{1}{\abs{s_k}} &\leq \frac{4 n}{\pi} \left(\gamma + \log(2n+2)\right)
    \end{aligned}
    \end{align}
    The general lemma follows from a rescaling by $1/a$.
\end{proof}

\begin{remark}
    Observe that $1/\abs{s_k}$ is essentially the number of Trotter steps to compute the $k$th interpolation point. Thus, Lemma~\ref{lem:Trotter_step_bound} amounts to a bound on the total number of Trotter steps, and we see this grows as $O(a^{-1} n \log n)$.
\end{remark}

\begin{lemma}\label{lem:c_norm}
Let $p(s) = (p_0(s), p_1(s), \dots, p_{n-1}(s))$ be a vector of (normalized) Chebyshev polynomials on $[-a,a]$. Then,
$$
\| \mathbf{V}p(0) \|_1 < \frac{2}{\pi}\log\left(n+1\right) + 1
$$
where $\norm{\cdot}_1$ denotes the vector 1-norm.
\end{lemma}

\begin{proof}
Let $d(s) = \mathbf{V}p(s)$. For each $k = 1, 2,\dots n$ we have
\begin{align}
\begin{aligned}
    d_k(s) &= \sum_{j=0}^{n-1} \mathbf{V}_{kj} p_j(s) = \sum_{j=0}^{n-1} p_j (s_k) p_j(s) \\
    &= \frac1n+\frac2n \sum_{j=1}^{n-1} \cos\left(j  \left(\frac{2k-1}{2 n}\pi\right) \right) \cos(j \cos^{-1}(s)).
\end{aligned}
\end{align}
At $s  = 0$, $\cos(j \cos^{-1}(0)) = \cos(j \pi/2)$, which is zero for odd $j$. Hence,
\begin{align}
\begin{aligned}
    d_k(0) &= \frac1n + \frac2n \sum_{j=2, \text{even}}^{n-2}\cos\left(j  \left(\frac{2k-1}{2 n}\pi\right) \right) (-1)^{j/2} \\
    &= \frac1n + \frac2n \sum_{j'=1}^{n/2 -1}(-1)^{j'} \cos\left(\pi j' \frac{2k-1}{n}\right).
\end{aligned}
\end{align}
The sum can be evaluated exactly (the authors used Mathematica), yielding
\begin{align}
    d_k(0) &= \frac1n - \frac2n\left(\frac{1-\cos((k+n/2)\pi)\tan(\pi\frac{2k-1}{2n})}{2}\right) \\
    &=\frac{1}{n}-\frac{1}{n}\left(1-(-1)^{k+n/2}\tan(\frac{2k-1}{2n}\pi)\right) \\
    &= \frac{1}{n}(-1)^{k+n/2} \tan\left(\frac{2k-1}{2n}\pi\right).
\end{align}
With coefficients in hand, we now compute the one norm of $d(0)$.
\begin{align}
    \norm{d(0)}_1 = \frac{1}{n} \sum_{k=1}^n \abs{\tan\left(\frac{2k-1}{2n}\right)}
\end{align}
We have a reflectional symmetry about $k \rightarrow n-k + 1$, allowing us to cut the sum in half and remove the absolute value sign.
\begin{align}
\begin{aligned}
    \norm{d(0)}_1 &= \frac{2}{n} \sum_{k=1}^{n/2} \tan\left(\frac{2k-1}{2n}\pi\right) \\
    &= \frac{1}{m} \sum_{k=1}^m \tan\left(\frac{2k-1}{2m}\frac{\pi}{2}\right)
\end{aligned}
\end{align}
Here, $m\equiv n/2$. We observe that the sum increases as $k$ approaches $m$ due to the first order pole at $\pi/2$. We can upper bound $\tan(x)$, and therefore the sum above, as follows.
\begin{align}
\begin{aligned}
    \frac{1}{m} \sum_{k=1}^m \tan\left(\frac{2k-1}{2m}\frac{\pi}{2}\right) \leq \frac{1}{m}\sum_{k=1}^m \frac{1}{\frac{\pi}{2}-\frac{\pi}{2}\left(\frac{2k-1}{2m}\right)} &= \frac{4}{\pi} \sum_{k=1}^m \frac{1}{2(m-k)+1} \\
    &= \frac{4}{\pi}\sum_{j=1}^m \frac{1}{2j-1}
\end{aligned}
\end{align}
In the last line, we reindexed by $j = m-k + 1$. Borrowing the reasoning from the prior lemma, 
\begin{align}
    \frac{4}{\pi} \sum_{j=1}^m \frac{1}{2j-1} < \frac{2}{\pi} (\gamma + \log(2n+2)).
\end{align}
Tracing back, this is an upper bound on $\norm{d(0)}_1$. Hence,
\begin{align}
    \norm{d(0)}_1 < \frac{2}{\pi}\log(n+1) + \frac{2(\gamma + \log(2))}{\pi} < \frac{2}{\pi}\log(n+1) + 1
\end{align} 
\end{proof}
The benefit of well-conditioning comes from relaxing the need to have exquisitely precise data to achieve good interpolations. This property is captured by the following theorem.

\begin{theorem}\label{thm:extrapBd}
Let $y = (f(s_1), f(s_2), \dots, f(s_n))^T$, and let $\tilde{y} \in \mathbb{R}^n$ be an approximation of $y$ in the sense that, for all $1 \leq j \leq n,$ $\vert f(s_j)-\tilde{y}_j \vert\le \epsilon/(\frac{2}{\pi}\log(n+1)+1)$ with probability at least $1-\delta/n$.  Let $p(s) = (p_0(s), \dots, p_{n-1}(s))^T$ be the vector of orthonormal Chebyshev polynomials. Then $\tilde{y}^T \mathbf{V} p(s)$ is an estimate of the interpolant $\polynom_{n-1} f(s)$ at $s = 0$ to precision
$$
\abs{\polynom_{n-1} f(0) -\tilde{y}^T \mathbf{V} p(0)} \le \epsilon
$$
with probability at least $1-\delta$.
\end{theorem}
\begin{proof}
First, observe that $P_{n-1} f (s)= p(s)^T c = p(s)^T \mathbf{V}^T y$ by the discussion surrounding~\eqref{eq:cheb_sys}. Hence,
\begin{align}
    \abs{\polynom_{n-1} f(0) - p(0)^T \mathbf{V}^T \tilde{y}} &= \abs{ (\mathbf{V} p(0))^T (y - \tilde{y})}.
\end{align}
By H{\"o}lder's inequality,
\begin{align}
\begin{aligned}
    \abs{(\mathbf{V} p(0))^T (y - \tilde{y})} &\leq \norm{\mathbf{V} p(0)}_1 \norm{y - \tilde{y}}_\infty.
\end{aligned}
\end{align}
From Lemma~\ref{lem:c_norm}, and from the assumptions on the distance between $y$ and $\tilde{y},$
\begin{align}
    \norm{\mathbf{V} p(s)}_1 \norm{y - \tilde{y}}_\infty \leq \left(\frac{2}{\pi} \log(n + 1) + 1\right) \frac{\epsilon}{\frac{2}{\pi} \log(n + 1) + 1} = \epsilon
\end{align}
with probability $\Pr = (1-\delta/n)^n$. In fact, since the probability of each component $\tilde{y}$ exceeding the specified distance is $\delta/n$, by the union bound the total probability of at least one component exceeding this distance is less than $n \times (\delta/n) = \delta$. Thus, the inequality is satisfied with probability $\Pr \geq 1-\delta$. This completes the proof.
\end{proof}

Theorem~\ref{thm:extrapBd} is what suggests that our interpolation approach may have the potential to achieve accuracy improvements without increasing costs compared to standard Trotter. It tells us that the error in Trotter data can be as large as the error of the final estimate up to a factor which is logarithmically small in the number of interpolation points, and therefore these data $\tilde{y}_i$ can be computed "cheaply enough." Thus, Theorem~\ref{thm:extrapBd} is plays an important role in the proofs of Appendices~\ref{app:PE_proof} and~\ref{app:expvals_proof}, and in the corresponding results presented below.

In the applications to Gaussian phase estimation we consider in~\Cref{subsubsec:gqpea}, we will need a probabilistic statement of well-conditioning that takes into account confidence intervals for the data. This is supplied by the following lemma.

\begin{lemma}\label{lem:extrapGaussBd}
Let $\{Y_j\}$ be normally distributed random variables for the $n$ Chebyshev interpolation points, with central values $\mu_j=f(s_j)$ and variances $\sigma_j$. Then the random variable $R_{n-1}(s)$ associated with the polynomial interpolant at value $s$ is normally distributed with mean $P_{n-1}(s)$ and variance $\sigma_R^2$ bounded as
\begin{align}
    \sigma_{R}^2(s) \leq 2 \max_j\left( \sigma^2_j\right).
\end{align}
\end{lemma}

\begin{proof}
The variable $R_{n-1}(s)$ is given by $p(s)^T \mathbf{V} Y$, with $Y = (Y_1, Y_2, \dots, Y_n)$. By linearity, the expectation value $E(R_{n-1}(s)) = p(s)^T \mathbf{V} E(Y) = p(s)^T \mathbf{V} y = P_{n-1}(s)$ as claimed, where $y = (f(s_1), \dots, f(s_n)).$

To bound the variance, observe that
\begin{align}
    \sigma^2_R(s)=\Var\left(R_{n-1}(s)\right) &\equiv \Cov\left(R_{n-1} (s),R_{n-1}(s)\right) \cr 
                &= \sum_{i,j,k,l} \Cov\left(p_i(s) \left(\mathbf{V}^{-1}\right)_{ij} Y_j,p_l(s) \left(\mathbf{V}^{-1}\right)_{lk} Y_k \right) \cr
                &= p^T(s) \mathbf{V}^{-1} \mathbf{\Sigma}^2 \mathbf{V} p (s)\\
                &= \norm{\mathbf{\Sigma} \mathbf{V}p(s)}_2^2,
\end{align}
where $\mathbf{\Sigma} = \diag\left(\sigma_1,\sigma_2,\dots,\sigma_n\right)$, and the last inner product was identified as a squared two norm. We can bound this using the spectral norm. 
\begin{align}
    \sigma_{R}^2(s) \leq \|\mathbf{\Sigma}\|^2 \| p(s) \|_2^2 \leq 2 \max_j \sigma_j^2
\end{align}
\end{proof}
At first, it may seem that the results of this lemma are quite restrictive, only applying to unbiased, normally distributed estimates of the function values $f(s_i)$. However, we will find occasion to apply this lemma in the context of the Gaussian phase estimation considered in~\Cref{subsubsec:gqpea}.

\section{Interpolation for Eigenvalues} \label{sec:PE}

\subsection{Approach 1: Error Analysis through Perturbation Theory} \label{subsec:perturbPE}
Having laid the groundwork for well-conditioned polynomial interpolation, we apply our framework to the task of phase estimation.  The idea is to perform logarithmically many phase estimation experiments evaluated at the Chebyshev nodes~\ref{eq:cheb_node_def}. We then bound the error using the interpolation theory of the previous section. The following theorem bounds the performance of such an algorithm.

\begin{theorem}\label{thm:PEMain}
Let $H:\mathbb{R}\mapsto \mathbb{C}^{2^n\times 2^n}$ be a Hamiltonian of the form  $H(t) =\sum_{j=1}^m a_j(t) H_j$ for all $t \in \mathbb{R}$, where each $H_j$ is a Hermitian matrix and $a_j(t)$ is real valued and in $C^n$ for positive integer $n$.  For each $t\in \mathbb{R}$, let $\ket{\ell(t)}$ be an eigenstate of $H(t)$ with eigenvalue $\lambda_\ell(t)$, and assume oracular preparation of $\ket{\ell (t)}$. Further, assume that there exists a minimum gap $\gamma(t)>0$ such that for all  $\ell'\ne \ell$ and $t>0$, $|\lambda_\ell(t)-\lambda_{\ell'}(t)|\ge \gamma(t)$.  Finally, that $t$ is sufficiently small such that the assumption~\eqref{eq:lem1_assumpt} of~\Cref{lem:Heffderiv} holds. It is then possible to use an $n$-point polynomial interpolation formula over the $2k$-th order Suzuki Trotter formula to estimate $\lambda_\ell(0)$, within error $\epsilon$ and failure probability at most $1/3$, using a number of operator exponentials and queries that is bounded above by
$$
\widetilde{O}\left(\frac{m^2(25/3)^k \max\|H_i\| (1+\Gamma)}{\epsilon} \right),
$$
where 
$$
\Gamma:= \max_{t,\, p=1,\ldots,n}\left( \frac{p}{n}\right)\left(\frac{ke^2(5/3)^km\max_{l\in [1,m]} \|H_l\|}{\gamma(t)}\right)^{1/p} 
$$
and $n\in \widetilde{O}\left( \log\left(mk(5/3)^k\max_i \|H_i\|(1+\Gamma)\right)  + \log 1/\epsilon\right)$.
\end{theorem}
The proof of this theorem is technical and relegated to~\Cref{app:PE_proof}. As a brief sketch, the proof proceeds by using perturbation theory to evaluate the derivatives of the eigenvalues and eigenvectors of the effective Hamiltonian $\tilde{H}_s$, as well as multiple applications of the triangle inequality and similar linearizing approximations.

The above result applies generically to any piecewise analytic Hamiltonian $H(t)$. To give a better intuition for how it could be used, let us focus our attention on a phase estimation protocol.  There are multiple phase estimation procedures that can be employed, but the shared basic idea is that one provides a quantum state $\ket{\psi} = \sum_j c_j(t) \ket{j(t)}$ and performs a series of evolutions of the form $e^{-i H(t) t}$ to this state to yield an estimate of one of the eigenvalues, $\exp(-i\lambda_l(t)t)$. Each eigenvalue is randomly sampled with probability $|c_j(t)|^2$.  We will usually demand that the variance of this estimate is bounded above by $t^2$ and that the expected error is at most $\epsilon$ for such a procedure, but it is also common for the accuracy guarantees to be given in terms of a probability of failure.

\subsubsection{Application: Estimation of Trotter Error}

\begin{figure}[t!]
\centering
\begin{quantikz}
    \lstick{$\ket{0}$}      & \gate{H} & \ctrl{1} & \octrl{1} & \gate{H} & \qw \\
    \lstick{$\openone/2^n$} & \qw      & \gate{U} & \gate{-V} & \qw      & \qw
\end{quantikz}
    \caption{Quantum circuit for computing a normalized Frobenius distance $d(U,V)$ between unitaries $U, V$, where the probability of measuring zero on the first qubit is $d(U,V)^2$. By taking $U$ as a single Trotter step and $V$ as the $s$-dependent product formula of equation~\eqref{eq:s_param_S2k}, polynomial extrapolation can be performed to estimate the true Trotter error.}
    \label{fig:errorObs}
\end{figure}
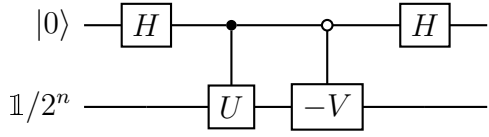

One important capability of our method is the estimation of error in ST formulas.  This is important because existing error bounds are typically not tight, and leading order expansions for the error are prohibitively expensive even for short evolutions~\cite{childs2021theory}.  We address this via a method for computing a distance metric $d$ between two unitaries $U,V$. This metric is related to the Frobenius distance $\|U-V \|_F$, where
\begin{align}
    \|A\|_F := \sqrt{\mathrm{Tr}{A^\dagger A}}
\end{align}
is the Frobenius norm. Specifically, we fix an ST order $k$ and let $\tilde{U}_s$ be as in equation~\eqref{eq:s_param_S2k}. Our method applies to any product formula, but our theoretical analysis restricts us to consider ST formulas. To estimate the Trotter error, we compute
\begin{align} \label{eq:Frobdistdef}
    d(\tilde{U}_1(t),\tilde{U}_s(t)):=\frac{1}{\sqrt{2^n}} \left\|\frac{\tilde{U}_1(t) - \tilde{U}_s(t)}{2}\right\|_F
\end{align}
using the circuit of Figure~\ref{fig:errorObs}. Observe that $\tilde{U}_1 = S_{2k}$ is a single Trotter step, while 
\begin{align}
    \lim_{s\rightarrow 0} \tilde{U}_s(t) = U(t)
\end{align}
is the exact time evolution operator. Therefore, by estimating
\begin{equation} \label{eq:Frob_dist}
d(S_{2k}(t), e^{-i H t}) = \lim_{s\rightarrow 0} d(\tilde{U}_s(t),\tilde{U}_1(t)),    
\end{equation}
we obtain a metric of the Trotter error. This metric $d\equiv d(\tilde{U}_1(t),\tilde{U}_s(t))$ corresponds to a root mean square of the singular values of $S_{2k}(t) - \tilde{U}_s(t)$. That is,
\begin{align}
    d(U,V) = \sqrt{\frac{1}{2^n} \sum_{j=1}^{2^n} (\sigma_j/2)^2}
\end{align}
for singular values $\sigma_j \in [0,2]$ of $U-V$. It has the following relation with the spectral distance.
\begin{align}
    \frac{1}{\sqrt{2^n}}\norm{\frac{U-V}{2}} \leq d(U,V) \leq \norm{\frac{U-V}{2}}
\end{align}
Both the lower and upper bounds are tight generically. The upper bound is tight in the case that all singular values are equal. The normalized Frobenius norm is best interpreted as an average error over all "principal directions." Thus, though not a worst-case error, $d$ still provides a useful characterization of the Trotter error.

Per usual, we consider an extrapolation of the normalized Frobenius distance $d$ to $s=0$.  
\begin{corollary} \label{cor:Frob_est}
Under the assumptions of Theorem~\ref{thm:PEMain}, there exists a quantum algorithm that can compute the quantity $d(S_{2k}(t),e^{-i H t})$ for $H=\sum_{j=1}^m H_j$ with $\epsilon$ error and failure probability (also) $\epsilon$ using a number of operator exponentials $N_{\exp}$ of the $H_j$ satisfying
$$
N_{\exp} \in \tilde{O}\left(\frac{m^2(25/3)^k\max \|H_i\| }{\epsilon} \right).
$$
\end{corollary}
\begin{proof}
Consider the  circuit of Figure~\ref{fig:errorObs}, with $U = \tilde{U}_s$ and $V = \tilde{U}_1$. Let $\mathrm{SELECT}$ refer to the control operations in this LCU-type circuit. We first demonstrate that the probability of measuring the first qubit to be zero is $ \|S_{2k}(t) - \tilde{U}_s(t)\|^2_F/(2^{n+2})$.  This fact follows from the analysis of the LCU lemma~\cite{LCU_2012}; however, that analysis is typically performed for pure states, so we generalize it here. Let $\ketbra{0}{0} \otimes \rho$ be the input state of the circuit.  Then the gate operations proceed as follows.
\begin{align}
    \ketbra{0}{0}\otimes \rho &\mapsto_H \sum_{ij}\frac{\ket{i}{\bra{j}}}{2} \otimes \rho\nonumber\\
    &\mapsto_{\mathrm{SELECT}}  \sum_{ij}\frac{\ket{i}{\bra{j}}}{2} \otimes (\tilde{U}_s(t)^i (-\tilde{U}_1(t))^{1-i})\rho (\tilde{U}_s(t)^j (-\tilde{U}_1(t))^{1-j})^\dagger\nonumber\\
    &\mapsto_{H}  \sum_{ijk\ell} (-1)^{ik+j\ell}\frac{\ket{k}{\bra{\ell}}}{4} \otimes (\tilde{U}_s(t)^i (-\tilde{U}_1(t))^{1-i})\rho (\tilde{U}_s(t)^j (-\tilde{U}_1(t))^{1-j})^\dagger:=\sigma
\end{align}
We then have that the probability of measuring the first qubit to be zero is
\begin{equation}
    \mathrm{Tr}((\ketbra{0}{0} \otimes \rho) \sigma) = \frac{1}{4} \mathrm{Tr}\left((\tilde{U}_s(t) - \tilde{U}_1(t)) \rho (\tilde{U}_s(t) - \tilde{U}_1(t))^\dagger\right).
\end{equation}
Taking $\rho = \openone /2^n$ in particular yields
\begin{align}
     \frac{1}{4} \mathrm{Tr}\left((\tilde{U}_s(t) - \tilde{U}_1(t)) \rho (\tilde{U}_s(t) - \tilde{U}_1(t))^\dagger\right) 
     &= \frac{\mathrm{Tr}\left((\tilde{U}_s(t) - \tilde{U}_1(t)) (\tilde{U}_s(t) - \tilde{U}_1(t))^\dagger\right)}{2^{n+2}} \nonumber\\
     & = \frac{\|\tilde{U}_s(t) - \tilde{U}_1(t)\|_F^2}{2^{n+2}}
\end{align}
Thus, the probability of measuring $0$ on the first qubit gives a normalized Frobenius distance between the two operators.  The cost of doing this is $O(1)$ queries to the underlying ST formulas, each of which boils down to $O(m5^{k})$ operator exponentials.

Using Amplitude Estimation on the target, we can construct an operator $W$ such that the eigenvalues of $W$ within the subspace supporting the initial state are of the form 
\begin{equation} \label{eq:Trotter_error_walk}
\lambda(W) = 
    \exp\left\{\pm i \sin^{-1}\sqrt{\frac{\|\tilde{U}_s(t) - \tilde{U}_1(t)\|_F^2}{2^{n+2}}} \right\}.
\end{equation} We then can invoke Theorem~\ref{thm:PEMain} to show that the number of exponentials needed to learn the extrapolated phase, within error $\epsilon'$ with probability greater than $2/3$, satisfies
\begin{equation}
    N_\mathrm{exp} \in \tilde{O}\left(\frac{m^2(25/3)^k\max \|H_i\| \log^{5/2}(\max_i \|H_i\| / \epsilon')}{\epsilon'} \right).
\end{equation}
The remaining question is how small $\epsilon'$ must be to guarantee that the error in the normalized Frobenius distance is at most $\epsilon$.  Let $\hat{\phi}$ denote the estimate of the phase returned by our protocol, which has error at most $\epsilon'$.  Our estimate $\hat{d}$ of the distance metric is then related to the phase in equation~\eqref{eq:Trotter_error_walk} as
\begin{equation}
    \hat{d} = \sin(\hat{\phi}).
\end{equation}
Thus we have that
\begin{align}
    \abs{\hat{d} - d} = \abs{\sin\hat{\theta} - \sin\theta} \leq \abs{\hat{\theta}-\theta}  \leq \epsilon'.
\end{align}
Choosing $\epsilon' = \epsilon$, then, is sufficient to ensure the desired tolerance as given in the lemma.
\end{proof}

As mentioned, $d$ is not a worst-case simulation error, but it may be a good approximation in typical situations. To get a true upper bound, we could employ the relation $\norm{A} \leq \norm{A}_F$, but estimating $\norm{U-V}_F$ directly from $d$ requires precision exponential in $n$. Moreover, the upper bound would not be very tight unless $d$ were already quite close to $\norm{U-V}$. 

In summary, our estimates should be better when the spread of the eigenvalues in the error operator is not too large. We formalize this in the following theorem.

\begin{corollary}\label{cor:Markov_Frob}
Let $\delta U^2 = |\tilde{U}_s(t) - \tilde{U}_1(t)|^2$, and suppose that 
$$
\sigma^2 (\delta U^2) = \mathrm{Tr}(\delta U^4)/2^n - \mathrm{Tr}(\delta U^2/2^n)^2 \le \xi
$$ 
for some $\xi > 0$.
Then each eigenvalue $\lambda^2$ of $\delta U^2$ satisfies
\begin{equation}
\mathrm{Pr}\left(\abs{\lambda^2 - \frac{\|\tilde{U}_s(t) - \tilde{U}_1(t)\|_F^2}{2^n}} \geq k\sqrt{\xi}\right) \leq \frac{1}{k^2}. \label{eq:PE_Cheb}
\end{equation}
for all $k>0$. Further, even with no such promise about the variance, the following weaker bound holds.
\begin{equation}
\mathrm{Pr}\left(\abs{\lambda^2 - \frac{\|\tilde{U}_s(t) - \tilde{U}_1(t)\|_F^2}{2^n}} \geq k\frac{\|\tilde{U}_s(t) - \tilde{U}_1(t)\|_F^2}{2^n}\right)\leq 1/k
\end{equation}
\end{corollary}
\begin{proof}
The result follows from the Chebyshev and Markov inequalities.
\end{proof}
\Cref{cor:Frob_est} and~\Cref{cor:Markov_Frob} show that we can use our interpolation procedure to the estimate largest eigenvalue of the error operator.  In particular, let the probability of an eigenvalue being greater than the estimate be $O(1/2^n)$.  Then the Chebyshev bound~\eqref{eq:PE_Cheb} implies that it suffices to take $k \in O(\sqrt{2^n})$.  Thus, with high probability, all of the eigenvalues for the square of the error operator will be at most $4 \hat{d}^2 + O(\sqrt{\xi 2^n})$.  Thus if ${\xi} \in o(\hat{d}^4/2^{n})$ then the estimate yielded by this procedure will also estimate the spectral norm.

\subsection{Approach 2: Error Analysis through Bernstein Ellipses and Analyticity} 
\label{subsec:Bernstein}
In this section, we propose an alternative analysis for eigenvalue estimation, as well as a specific approach to phase estimation. We first prepare each effective eigenstate through the procedure in~\cite{rendon2022effects}, except using one single qubit the semi-classical QFT, summarized in \Cref{thm:gauss_state_prep}. Then we perform a new Gaussian phase estimation (\Cref{thm:spectral_gaussian_error}) on said eigenstate and then interpolate the estimates. The main result is summarized in~\Cref{thm:gauss_phase_est}.

To analyze the error and complexity of this approach, we use a different formalism that relies on complex analyticity~\cite{trefethen2019approximation}, which allows us to estimate the convergence rate of the interpolation in terms of range of analyticity rather than and not derivatives of the effective Hamiltonian. This is beneficial since the bounds using derivatives might become unmanageable near level crossings (See \Cref{app:PE_proof}). 

\subsubsection{Ancillary Lemmas}

To begin, we need to introduce some notions from complex analysis. For each $\rho > 1$, let $B_\rho \subset \mathbb{C}$ be the Bernstein ellipse, which is an ellipse with foci at $\pm 1$ and semimajor axis $(\rho + \rho^{-1})/2$. The following lemma bounds the Chebyshev interpolation error for analytic functions on $B_\rho.$
\begin{lemma}\label{lem:Berns}
    Let $f(z)\in \mathbb{C}$ be an analytic function on $B_\rho$, and suppose $C\in\mathbb{R}_+$ is an upper bound such that $|f(z)| \le C$, for all $z \in B_\rho$. Then the Chebyshev interpolation error on $[-1,1]$ satisfies
    $$
    \norm{f - P_n f}_{\infty} \leq \frac{4 C \rho^{-n}}{\rho - 1}
    $$
    for each degree $n>0$ of the interpolant through the $n+1$ Chebyshev nodes.
\end{lemma}
\begin{proof}
    Theorem 8.2 of Ref.~\cite{trefethen2019approximation}.
\end{proof}
\Cref{lem:Berns} shows that the interpolation error shrinks exponentially in $n$. We would like to apply this lemma to analyze the eigenvalues $\lambda(z)$ of the effective Hamiltonian $\tilde{H}_z$, which is a continuation of $\tilde{H}_s$ to the complex plane. However, we need to understand the domain under which $\lambda(z)$ is analytic, i.e. free of level-crossings.

To characterize this domain, we first utilize a result from Bauer and Fike~\cite{bauer_fike} which bounds the shift in the eigenvalues under a shift in the operator.

\begin{lemma}\label{lem:BauerFike}
Let $A$ be a normal matrix with eigenvalues $\{\lambda_i\}$. Then, if $\lambda$ is an eigenvalue of a matrix $B$, there exists an eigenvalue $\lambda_k$ of $A$ such that
\begin{align}
|\lambda - \lambda_k| \leq  \norm{B-A} 
\end{align}
for at least one eigenvalue $\lambda_k$ of $A$, where $\norm{\cdot}$ is the spectral norm.
\end{lemma}

In order to determine the radius of analyticity, we assume a specific form for the upper bound on $\|H-\tilde{H}_z\|$, where $\tilde{H}_z$ is the analytic continuation of $\tilde{H}_s$ to the complex plane. Here is the resulting theorem

\begin{lemma}\label{lem:radius_of_analiticity}
Let $\lambda(\tau)$ be an eigenvalue of $\tilde{H}_z$, where $\tau = z t$ for $z\in\mathbb{C}$ and $t\in\mathbb{R}^+$, and $\gamma_0$ be a lower bound on the spectral gap of $\tilde{H}_0 = H$. Let $\alpha, \beta\in\mathbb{R}^+$ be constants such that
\begin{align} \label{eq:Hdiff_bd}
    \| H - \tilde{H}_z \| \leq \frac{\alpha |\tau|^p}{(p+1)!} e^{\beta |\tau|}.
\end{align}
Then $\lambda(\tau)$ is analytic on any origin-centered disc of radius $r$ provided that $r \leq r_{\mathrm{max}},$ with

\begin{align}
    r_\mathrm{max} := \frac{p}{\beta}W_0\left(\frac{\beta}{p} \left(\frac{\gamma_0 (p+1)!}{2 \alpha }\right)^{1/p}\right).
\end{align}
Here $W_0$ is principal branch of the Lambert W function.
\end{lemma}
\begin{proof}
    See \Cref{app:radius_of_analiticity}
\end{proof}

The above lemma gives an upper bound $r_\mathrm{max}$ on the radius of analyticity. From this, assuming $r_\mathrm{max} > 1$, we can calculate the largest possible $B_\rho$ by equating the semimajor axis with the maximum radius
\begin{align}
    (\rho_\mathrm{max} + \rho_\mathrm{max}^{-1})/2 = r_\mathrm{max}
\end{align}
which has solution
\begin{align}
    \rho_\mathrm{max} = r_\mathrm{max} + \sqrt{r_\mathrm{max}^2 -1}.
\end{align}

The following lemmas give a clue as to what $\alpha$ and $\beta$ should look like in terms of commutators between the terms of $H$ and also with respect to its overall spectral norm. However, first, we introduce a bound on the error on the effective total Hamiltonian evolution in terms of the "instantaneous" Hamiltonian error, $\mathscr{E}$, which is defined through the complex-time Schr\"odinger equation. Let $S_p(\tau)$ be a $p$th order product formula approximating $\exp(i H \tau)$, such that
\begin{align}
    \frac{\mathrm{d}}{\mathrm{d}\tau}S_p(\tau) = i(H+\mathscr{E}(\tau))S_p(\tau).
\end{align}
For this purpose, we also introduce the "accumulated" error
\begin{align}
    E(\tau) = \log S_p(\tau) - i \tau H,
\end{align}
where $\log S_p (\tau)$ can be defined through the complex-time Magnus expansion (See~\cite{blanes2009magnus} for compilation of proofs, complex-time extension is straight-forward)
\begin{align}
    \log S_p (\tau) &= \Omega = i\int_{\mathcal{P}} \sum_{n=0}^\infty \frac{B_n}{n!} \mathrm{ad}^n_{\Omega}(H+\mathscr{E}(\tau_1)) \mathrm{d}\tau_1 \cr 
\end{align}
where $\mathcal{P}$ is a path going from $\tau_1=0$ to $\tau_1=\tau$. This, like the real-time Magnus expansion, converges provided $\|\Omega\|\leq \pi$. With this, we obtain an upper bound on $\|E(\tau)\|$ stated through the following lemma:

\begin{lemma}\label{lem:Magnus_error}
Given the error definition
\begin{align*}
    E (\tau) = \log S_p(\tau) - i\tau H 
\end{align*}
for complex $\tau$, we can bound its norm through
\begin{align*}
    \left\|E(\tau)\right\| &\leq \frac{5}{2}\int^{|\tau|}_0  \|\mathscr{E}(s \tau/\abs{\tau})\| \mathrm{d} s,
\end{align*}
provided that $\max_{|\tilde{\tau}|\leq |\tau|}\| E(\tilde{\tau}) \|\leq 15/8$, $|\tau|\leq 1/8$, and $\int^{|\tau|}_0 \|\mathscr{E}(s \tau/\abs{\tau})\|\mathrm{d}s \leq 1/8$.
\end{lemma}
\begin{proof}
    See \Cref{app:Magnus_error}
\end{proof}

Now, we bound the norm for $\mathscr{E}(\tau)$ in \Cref{app:H_error_bound} and make use of $\norm{H - \tilde{H}_z} \leq  \left\|E(\tau) \right\|/|\tau|$.
\begin{lemma}\label{lem:H_error_bound}
Let $\tilde{H}_z$ be the effective Hamiltonian associated with a complex-time $p$th order product formula $S_p(\tau)$ , where $\tau = z t$ with $z\in\mathbb{C}$ and $t\in\mathbb{R}_+$, which is analytically continued to an open neighborhood containing $[-1,1]$. The norm of operator error can upper bounded as
\begin{align*}
\norm{H - \tilde{H}_z} \leq \frac{5}{2}\sum_{(\upsilon,m)} \acomm\big(a_{(\Upsilon,m)}H_{\pi_{\Upsilon}(m)},\ldots,a_{(\upsilon,m+1)}H_{\pi_{\upsilon}(m+1)},a_{(\upsilon,m)}H_{\pi_{\upsilon}(m)}\big) \cr
\cdot\frac{\abs{\tau}^p}{(p+1)!} e^{2|\tau| \, \sum^{m}_{j} \left\| H_{j} \right\|}
\end{align*}
provided that $|\tau|\norm{H - \tilde{H}_z}\leq15/8$, $|\tau|\leq 1/8$, and
\begin{align*}
\sum_{(\upsilon,m)} \acomm\big(a_{(\Upsilon,m)}H_{\pi_{\Upsilon}(m)},\ldots,a_{(\upsilon,m+1)}H_{\pi_{\upsilon}(m+1)},a_{(\upsilon,m)}H_{\pi_{\upsilon}(m)}\big)\frac{\abs{\tau}^p}{(p+1)!} e^{2|\tau| \, \sum^{m}_{j} \left\| H_{j} \right\|} \leq \frac{1}{20}.
\end{align*}
Here, $\acomm\big(A_s,\ldots,A_1,B\big):=\sum_{q_1+\cdots+q_s=p}\binom{p}{q_1\ \cdots\ q_s}\norm{\ad_{A_s}^{q_s}\cdots\ad_{A_1}^{q_1}(B)}$.
\end{lemma}

While the above theorem is a useful theoretical bound, a problem arises because bounding the commutators is often impractical computationally.  In such cases, $\acomm$ can be upper bounded through the triangle inequality and submultiplicativity of the norm.
\begin{align}
    \acomm\big(A_s,\ldots,A_1,B\big) \leq \| B \| 2^p (\sum^s_{j=1} \| A_j \|)^p.
\end{align}

With these set of lemmas at our disposal, we can now provide a cost estimate for state preparation for all the ground states of $\tilde{H}_{s_k}$.

\subsubsection{State Preparation Cost}

\begin{theorem}[State preparation cost]\label{thm:gauss_state_prep}
Let $H=\sum_{j=1}^m H_j$ satisfy $\sum^{m}_{j=1}\| H_j \| \leq 1$, and let $\gamma \in \mathbb{R}_+$ be a lower bound on the spectral gap of $H$. Let $s_1, \dots ,s_n$ be the $n$th degree Chebyshev nodes on $[-1,1]$ for even $n$. Then there exists an algorithm to prepare the ground state of $\tilde{H}_{s_k}$ to precision $\epsilon$ using controlled queries of
$
    \tilde{U}_{s_k}'=\left(S_{p} \left(t s_k \right) \right)^{e'_k}, 
$
where $t\in\Theta(\gamma^{1/p})$ and $e'_k = \sgn{s_k}\left\lceil {s_1}/{\left| s_k\right|} \right\rceil$, with a total gate count
\begin{align*}
    C_{\rm prep} \in O\left(\frac{\log(1/\epsilon)\log\log(1/\epsilon)}{\gamma^{\;\;1 + 1/p}} \right)
\end{align*}
and a single auxiliary qubit. 
\end{theorem}
\begin{proof}
    See \Cref{app:proof_state_prep}.
\end{proof}

Now that the cost of state preparation is covered, we assume an exact state preparation for our next gate count upper bound for eigenvalue estimation. However, first, we will introduce a new phase estimation algorithm as one of our methods for eigenvalue estimation. 

\subsubsection{Interlude: Gaussian Quantum Phase Estimation}\label{subsubsec:gqpea}

\begin{figure*}[t!]
\scalebox{0.8}{
\begin{quantikz}
 \lstick{$\ket{0^{\tp (q-m)}}$} &\qwbundle[alternate]{} &\qwbundle[alternate]{} & \qwbundle[alternate]{} & \qwbundle[alternate]{}& \qwbundle[alternate]{} & \qwbundle[alternate]{} & \qwbundle[alternate]{} & \qwbundle[alternate]{} &  \gate[wires=6,nwires=4]{Z(a,b)}\qwbundle[alternate]{} & \gate[wires=6,nwires=4]{QFT}\qwbundle[alternate]{}\slice{$\sim\ket{p^{(q,\mu=\theta_0,\sigma_f)}_{G}} $} & \meter{}\qwbundle[alternate]{} & \rstick[wires=6]{$z$}
 \\
 \lstick[wires=5]{$\ket{p^{(m,\mu=0,\sigma)}_{G}} $} & \qw & \qw &  \qw & \qw & \qw & \qw\ldots &\qw & \ctrl{5}  & \qw & \qw & \meter{} & 
 \\
  & \qw &\qw &\qw& \qw&\qw & \qw\ldots & \ctrl{4}   & \qw & \qw & \qw &  \meter{} &
 \\
  & \vdots & \vdots &\vdots& \vdots & \vdots & \vdots &\vdots&\vdots&\vdots&\vdots& \vdots
 \\
  & \qw &\qw &\qw& \qw&\ctrl{2}  & \qw\ldots & \qw   & \qw & \qw & \qw &  \meter{} &
 \\
  & \qw & \qw&\qw& \ctrl{1} & \qw  & \qw\ldots & \qw  & \qw & \qw &  \qw & \meter{} &
 \\
 \lstick{$\ket{\psi}$} & \qwbundle[alternate]{} & \qwbundle[alternate]{} & \qwbundle[alternate]{} & \gate{\tilde{U}_{s_k}^{'2^0}}\qwbundle[alternate]{}& \gate{\tilde{U}_{s_k}^{'2^1}}\qwbundle[alternate]{} & \qwbundle[alternate]{}\ldots & \gate{\tilde{U}_{s_k}^{'2^{m-2}}}\qwbundle[alternate]{} & \gate{\tilde{U}_{s_k}^{'2^{m-1}}}\qwbundle[alternate]{} & \qwbundle[alternate]{} & \qwbundle[alternate]{} & \qwbundle[alternate]{} &
\end{quantikz}
}
\caption{Circuit to implement a Gaussian $m$-qubit phase estimation algorithm with $(m-q)$-qubits for spectral interpolation. The $\tilde{U}'_{s_k}$ operator is evolved by an effective Hamiltonian $\tilde{H}_{s_k}$, with a Trotter step size $ts_{k}$ (See \Cref{thm:gauss_state_prep}) and corresponding evolution time $T'_k= e'_k\times dt  = \Theta (t s_1)$.}\label{fig:gaussian_qpea}
\end{figure*}
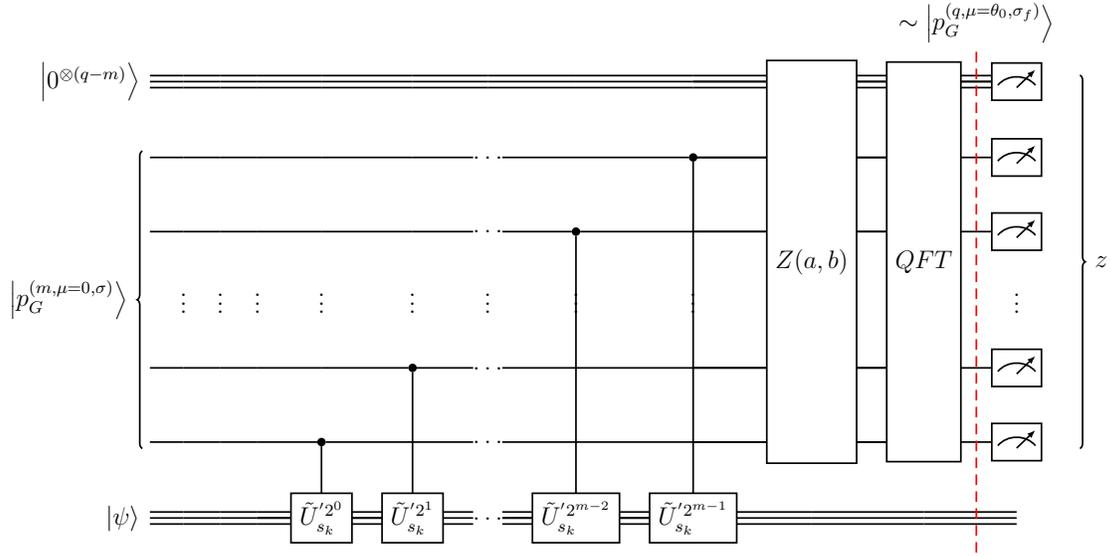

The first step in the phase estimation protocol involves preparing a Gaussian distribution in the ancillary register of $m$ qubits. Let $p(w; \sigma)$ be the probability density function for a Gaussian distribution with zero  mean ($\mu=0$) and variance $\sigma^2$. 
\begin{align}
    p (w;\sigma) := \frac{1}{\sigma\sqrt{2\pi}} e^{-w^2/(2\sigma^2)}
\end{align}
Define
\begin{align}
\ket{p_G(\sigma, T, m)} := \sum^{2^{m-1}-1}_{x=-(2^{m-1}-1)} \sqrt{\frac{p(x\,T;\sigma)}{\N(\sigma, T,m)}}\ket{x}
\end{align}
where $T>0$ plays the role of a sampling rate, and $\N(\sigma, T,m) > 0$ is the normalization constant for the discrete, truncated Gaussian.
\begin{align}\label{eq:gqpea_normalization_def}
    \mathcal{N}(\sigma,T,m) := \sum^{2^{m-1}-1}_{x=-(2^{m-1}-1)} p(nT;\sigma).
\end{align}

An exact preparation of $\ket{p_G}$ can be achieved through the methods proposed in Ref.~\cite{Klco_2020}. However, within a target error, we can prepare a coarser Gaussian distribution and then perform an upsampling through QFT, zero-padding, and then QFT$^{-1}$. This is an established method for upsampling/interpolation of discrete signals in classical discrete signal processing (DSP), and was introduced for quantum distribution preparation in Ref~\cite{garcia2021quantum}.
In \Cref{fig:gaussian_qpea} we illustrate the circuit used for this method. The operator $Z$ plays the role of zero-padding ubiquitous to classical DSP for interpolation in the conjugate space. The errors introduced by truncation and finite sample rate in the time domain are estimated in \Cref{app:gaussian_sample_error}, and can be summarized with the following theorem.

\begin{theorem}\label{thm:spectral_gaussian_error}
Let $q, m$ be positive integers such that $q\ge m$. Let $T, \sigma >0$ be chosen such that $\sigma/T = \Theta(\sqrt{2^m})$, and define the Fourier conjugates $F:=1/(2^q T)$ and $\sigma_f:=1/(4\pi\sigma)$. Finally, let
\begin{align}
    X(f) := \mathcal{F}\{\sqrt{p(\cdot \; ;\sigma)}\}(f) := \int_\mathbb{R} \sqrt{p(t;\sigma)}  e^{-2\pi i  t f} dt = \sqrt{\sqrt{2\pi}\sigma_f} p (f; \sigma_f)
\end{align}
denote the Fourier transform of the (root of the) Gaussian $p$. Then the spectral norm error of the prepared wave function in Fourier domain is
\begin{align}
     &\left\Vert \sum^{2^q}_k \frac{X(kF)}{\sqrt{\mathcal{N}(\sigma_f,F,q)}}  \ket{k} - \mathrm{QFT}\ket{p_G(\sigma, T, m)} \right\Vert \in  O\left(\frac{2^{m/2}}{e^{\Omega(2^m)}}\right).
\end{align}
\end{theorem}
Using these errors in the approximate state constructed in Fourier domain, we can perform a sequence of unitaries controlled on the Gaussian window.  This approach is similar in spirit to the Kaiser-window approach taken in~\cite{gorecki2020pi}. However, here we use a Gaussian window, which by inverting the conjugate spaces labels in \Cref{thm:spectral_gaussian_error}, we know we can prepare efficiently with a cost $O(\mathrm{polylog}(1/\varepsilon)/\sigma)$, where $\varepsilon$ is the vector error on the window state.

\subsubsection{Phase Estimation Cost}

Via the above theorem, we can now bound the cost for extrapolating eigenvalues obtained using Gaussian phase estimation. Through \Cref{thm:extrapBd} and \Cref{lem:extrapGaussBd}, we estimate the effects of uncertainty propagation that interpolation has on the final interpolant for Gaussian and one-qubit phase estimation. The result goes as follows:

\begin{theorem}[Phase estimation cost] \label{thm:gauss_phase_est}
The gate count of estimating the eigenvalues of $\tilde{H}_{s_k}$ using Gaussian phase estimation is 
\begin{align*}
    C_\mathrm{est, gpe} \in O\left(\frac{\log(1/\epsilon)}{\sigma_{P}}\right),
\end{align*}
where $\sigma_P$ is the standard deviation on the interpolated observable at $s=0$. The number of auxiliary qubits required is $O\left(\log (1/\sigma_P)\right)$. Alternatively, using a single-ancillary-qubit approach,
\begin{align*}
    C_\mathrm{est,1-qubit} \in O\left(\frac{\log(1/\epsilon)\log\log(1/\epsilon)}{\epsilon}\right),
\end{align*}
where $\epsilon>0$ is such that $\max_s|\lambda(s)-\polynom_{n-1}\tilde{\lambda}(s)|\leq \epsilon$ and $\polynom_{n-1} \tilde{\lambda}$ is the (unique) $(n-1)$-degree polynomial interpolating the data $\tilde{\lambda}_i$ at the $n$ interpolation points $s_i$. 
\end{theorem}
The bound for $C_\mathrm{est, gpe}$ can also be cast in terms of a confidence interval, $\epsilon= w \sigma_P $, around the mean, which introduces an error rate that decreases super-exponentially with $w$.

\section{Interpolation for Expectation Values}
\label{sec:expvals}
We now consider the application of Chebyshev interpolation to estimate expectation values, a fundamental task in quantum computation. The setting is as follows: given a quantum state $\rho$ and observable $O$, the expectation value is given by $\langle O \rangle = \Tr{\rho O}$. We evolve our system according to a $2k$-th order ST formula $\tilde{U}_s$ given by~\eqref{eq:s_param_S2k}. The time evolved expectation values of interest is captured by the function
\begin{align} \label{eq:evals}
    f(s) := \frac{\Tr{\rho O_s(t)}}{\norm{O}}
\end{align}
where $O_s(t)$ is given by equation~\eqref{eq:setup_expvals}. We've normalized the expectation values by $\norm{O}$ because the relative error is a useful and natural metric, and also the normalized operators may be block encoded for amplitude estimation. Alternatively, we simply restrict our attention to normalized observables with $\norm{O} = 1$. The interpolation algorithm we propose can be summarized as follows.
\begin{enumerate}
    \item Given Hamiltonian $H$, simulation time $t$, and tolerance $\epsilon$ for the estimate of $\langle O(t)\rangle/\norm{O}$, choose the appropriate interpolation interval $[-a,a]$ and an even number $n$ of Chebyshev nodes. We neglect the cost of this step. The error analysis we will perform subsequently will inform the choices of $a$ and $n$.
    \item Compute estimates $\tilde{y}_i$ of the expectation values $\langle O_{s_i}(t) \rangle$ for each $s_i$ with $i = 1, \dots, n/2$, to an accuracy depending on $\epsilon$ and $n$. We will assume this step is done with Iterative Quantum Amplitude Estimation (IQAE)~\cite{IQAE}, a recent approach to amplitude estimation that exhibits low quantum overhead. Our metric of cost is the number of $H_j$ exponentials executed on a quantum circuit, where $H = \sum_{j} H_j$. Note that by symmetry, we need not compute $\tilde{y}_i$ for $i > n/2$. We have $f(s_i) = f(s_{n-i+1})$ for all $i \in \{1,\dots,n\}$.
    \item Perform the polynomial fit $\tilde{P}_{n-1} f$ through the points $(s_i, \tilde{y}_i)$ using a Chebyshev expansion~\eqref{eq:O_expans_opt}. Note that $\tilde{P}_{n-1} f$ will automatically be even. This fit is well-conditioned, and we neglect the cost of this step.
    \item Evaluate the $\tilde{P}_{n-1} f(0)$ to be our estimate of $\langle O(t)\rangle$.
\end{enumerate}
To summarize, one performs amplitude estimation to acquire the time evolved expectation value at each Chebyshev node, then performs a polynomial interpolation of the data. The estimate is the value at $s = 0$.

Given an even set of Chebyshev nodes $\{s_1, \dots, s_n\}$, and making use of Lemma~\ref{lem:Cheb_error}, the interpolation error $E_{n-1}$ assuming perfect data points is given by
\begin{align}
    \abs{E_{n-1} (0)} \leq \frac{\abs{\Tr{\rho \,\partial_s^n O_s (t)}}}{\norm{O} n!} \prod_{i=1}^{n} \abs{s_i} \leq \max_{s \in [-a,a]} \frac{\norm{\partial_s^n O_s (t)}}{\norm{O}} \left(\frac{a}{2n}\right)^n.
\end{align}
With a suitable bound on $\partial_s^n O(t)$, we can provide an upper bound on the interpolation error at $s=0$. This bound is provided by the following lemma. In what follows, it will be helpful to define the parameter
\begin{align}
    c:=k (5/3)^k m \max_{l\in[1,m]} \norm{H_l} t
\end{align}
for ease of notation. This parameter is proportional to the the Hamiltonian size and the "total Trotter time," meaning the sum of all the forward and backward time steps, in absolute value, for a $2k$th ST formula.
\begin{lemma}[Error extrapolation for time-evolved observables] \label{lem:interp_error_evals}
    Under the conditions of~\Cref{lem:Heffderiv} ($ c a \leq \pi/20$), the following bounds holds on the Trotterized evolution $O_s (t)$ with step parameter $s \in (0,a]$:
    \begin{enumerate}
        \item for $c > n$ we have that
        $$\frac{\norm{\partial_s^n O_s (t)}}{\norm{O}} < \left(c\sqrt{e^3(1+\sqrt{8/\pi}e^2)}\right)^{2n}$$
    which gives an interpolation error 
    $$
        \abs{E_{n-1}(0)} < \left(129\frac{c^2 a}{n}\right)^n.
    $$
    \item For $c \le n$, we have
    $$
        \frac{\norm{\partial_s^n O_s(t)}}{\norm{O}} \leq \sqrt{\frac{2n}{\pi}} \left(\frac{e^4 c}{2}\right)^n n! e^{4ce^2\sqrt{2/\pi}}
    $$
    giving an interpolation error
    $$
        \abs{E_{n-1}(0)} \leq 2 \sqrt2 n \left(6ca\right)^n e^{24c}.
    $$
    \end{enumerate}
\end{lemma}
The proof of this lemma is a tedious exercise in repeated in the combinatorics of large derivatives and the triangle inequality, and is left to~\Cref{app:expvals_proof}. Note that once the derivative bound holds, the interpolation error bound follows immediately from~\Cref{lem:Cheb_error}. 

One motivation for these bounds is deriving asymptotic expressions for the algorithmic complexity. The following theorem gives an asymptotic query complexity for the number $N_\mathrm{exp}$ of Trotter exponentials $\exp(-iH_j \tau)$.

\begin{theorem} \label{thm:expvals_cost}
    Let $O(t) = U^\dagger (t) O U (t)$ be a time-evolved observable under a Hamiltonian $H = \sum_{l = 1}^m H_l$ on $n$ qubits, so that $U(t) = e^{-i H t}$. Suppose there exists a $\gamma \in \mathbb{R}_+$ such that $O/\gamma$ can be block encoded via a unitary $U_\mathrm{enc}$ by a state $\ket{G}$ on a set of $L$ auxiliary qubits. Let $\rho$ be a quantum state on $n$ qubits, and suppose $\gamma/\norm{O} \in O(1)$. Then, the number of exponentials $N_\mathrm{exp}$ required to estimate $\Tr{\rho O (t)}/\norm{O}$ to precision $\epsilon$ with confidence $1-\delta$ using a $2k$ order Suzuki Trotter formula satisfies 
    $$
        N_\mathrm{exp} \in \tilde{O}\left(c \max\{c, \log(1/\epsilon)\} \epsilon^{-1} \log(1/\delta) \right).
    $$
    Here, $\tilde{O}$ is big-$O$ with multiplicative terms suppressed which are logarithmically smaller in $1/\epsilon$ and $c$. Moreover, the number of auxiliary qubits needed is $O(L)$.
\end{theorem}
We give a sketch of the proof. Given a choice of interval $[-a,a]$ and (even) number of interpolation points $n$, we have from~\Cref{lem:Trotter_step_bound} that the number of exponentials to perform evolutions for all Chebyshev nodes goes as
\begin{align} \label{eq:Nexp_na}
    O\left(\frac{n \log n}{a}\right).
\end{align}
However, this is not the total cost since these circuits need to be repeated to perform the appropriate measurement protocols. Since $O$ can be block encoded, the expectation value can be obtained via an amplitude estimation protocol. By the well-conditioning of Lemma~\ref{thm:extrapBd}, each data point needs to be within $\epsilon$ of the exact Trotter value, up to a logarithmic factor in $n$. This robustness is why our result maintains a $\widetilde{O}(\epsilon^{-1})$ scaling. 

In our proof, we assume the IQAE protocol is used, requiring only a single qubit overhead. The fractional queries for the noninteger timestep also require $O(1)$ overhead, meaning the total overhead is $O(L)$ due to the block encoding. To relate $n$ and $a$ to the required precision $\epsilon$, simulation time $t$ and Hamiltonian $H$,~\Cref{lem:interp_error_evals} can be used. Thus, we can relate $N_{\mathrm{exp}}$ to these basic parameters. We carry out the formal proof in Appendix~\ref{app:expvals_proof}. 

As advertised, we see there is a "near-Heisenberg" scaling of $1/\epsilon$, up to logarithmic factors. However, there is an unsavory quadratic scaling in the simulation time in cases without high accuracy demands. We believe this can be improved, because our approach us forces us to have $n$ scale linearly in $T$. We believe this is overly pessimistic, and a $\log(\norm{H} T)$ scaling is more likely. Finally, our results suggest the best performance for using low order formulas, since our bounds are strictly worse for increasing ST order $k$.

\section{Numerical Experiments}
\label{sec:numerics}
In the previous sections, we presented theorems, with proofs in the appendices, that interpolation of Trotter data can lead to higher accuracies for eigenvalues and expectation values than Trotter alone.  This section provides numerical evidence backing up our analytic claims, not only showing that improved scaling is possible, but also that high-order Trotter formulas need not always provide better error scaling when used in conjunction with interpolation. Specifically, we demonstrate this improved scaling for phase estimation using interpolation for second-order and fourth-order ST formulas. For our demonstration we consider the transverse Ising model of two spins.
\begin{align}
H=-J\left( Z\otimes Z +g (X\otimes I  + I\otimes X)\right)
\end{align}
This is a minimal example with just two non-commuting terms, yielding a nonzero Trotter error. Our aim will be to estimate the ground state energy of the effective Hamiltonian $\tilde{H}_s$ in the limit $s \rightarrow 0$, and seek numerical evidence of the improved performance relative to low-order Trotter formulas that Theorems~\ref{thm:gauss_state_prep} and~\ref{thm:gauss_phase_est} suggest.

In order to avoid aliasing with our Fourier spectral methods, we must satisfy the bound
\begin{align}
\norm{\tilde{H}_{s_k}}\frac{t}{2\pi}  \abs{e_k' s_k} + \Delta_\text{pad} \leq 1,
\end{align}
where
\begin{align}
e'_k = \sgn{s_k}\left\lceil \frac{s_1}{\left| s_k\right|} \right\rceil.
\end{align}
The padding $\Delta_\text{pad}$ is there to suppress the probability leakage wrapping around the boundary. For simplicity, define a dimensionless effective Hamiltonian
\begin{align}
\tilde{h}_s := \tilde{H}_s \frac{t e_s' s}{2\pi},
\end{align}
with $e_s' = \sgn{s}\left\lceil s_1/\left| s\right|\right\rceil$, such that
\begin{align}\label{eq:new_spec_bound_Fourier}
\norm{\tilde{h}_{s_k}}+\Delta_{\text{pad}} \leq 1.
\end{align}
This way, the spectrum of $\tilde{h}_{s_k}$ lies within the domain $[-1/2+\Delta_{\text{pad}}/2,1/2-\Delta_{\text{pad}}/2]$ and now is in line with the conventions of Fourier spectral methods in~\cite{rendon2022effects}.   
With this, we now define the normalized Hamiltonian
\begin{align}
    h = \tilde{h}_0.
\end{align}
Here $t$ is sufficiently small as to fulfill \Cref{eq:new_spec_bound_Fourier}. We must also make sure that $s$ is chosen according to the constraints of~\Cref{subsec:Bernstein} such that there are no level crossings throughout the interpolation interval.

We now describe the results.~\Cref{fig:extrapol_second_order} displays the results of Chebyshev interpolation with second order Trotter formulas using various (even) numbers of points.  For each $\tilde{H}'_{s_k}$, we numerically simulate the ground state preparation protocol to a state error, $\epsilon_{\text{state}}$. This is efficient because the algorithm scales $O(\log 1/\epsilon_\text{state})$. We then simulate the Gaussian phase estimation on these states. Moreover, we have used zero-padding as described in the \Cref{subsubsec:gqpea} in order to upsample the spectrum and be able to ignore digitization errors. The spectral upsampling through zero-padding is efficient since the cost scales like $O(q^2)$, where $q$ is the number of total qubits after padding. The blue bands in these plots show an uncertainty equivalent to $0.01\times\sigma_P(s)$. We have down-scaled the statistical uncertainties by a factor of a 100 such that the interpolation error changes are noticeable when we change the number of nodes.

Next, in~\Cref{fig:truncation_error}, we show the exact systematic error from our simulations as well as the upper bounds of the Bernstein ellipse analysis of Lemmas~\ref{lem:H_error_bound},~\ref{lem:radius_of_analiticity}, and~\ref{lem:Berns}. We see from this data that interpolation indeed improves the quality of the estimate of the energy as anticipated, even when relatively low-precision estimates are used at $s = 1$. As discussed previously, only positive values of $s_k$ are computed due to reflection symmetry.

Finally, in~\Cref{fig:Trotter_vs_extrapol} we performed ground state energy estimation using second and fourth order formulas without interpolation, in order to compare the costs when plotted against the error. From the plots, we observe that beyond $n=2$ the extrapolation methods already outperform using second-order product formulas alone.  Specifically, we see clear indications from this data that the bias in the interpolated error (for large $1/\epsilon$) scales logarithmically for these Gaussian phase estimation experiments. This agrees with our expectations from Theorem~\ref{thm:gauss_phase_est} wherein the systematic errors from interpolating phase estimation experiments are expected to scale as $\mathrm {polylog}(1/\epsilon)$. In contrast, the bias from a fixed experiment is expected to be $O(1/\epsilon^{3/2})$ and $O(1/\epsilon^{5/4})$ respectively.  

It is worth noting that this does not violate the Heisenberg limit as here we are only showing the interpolation error, with data computed to machine precision. The exponential reduction in the uncertainty with the number of interpolation points considered is independently shown in \Cref{fig:truncation_error}.  Note that for this data we do not see an improvement from transitioning to higher-order formulas, which is consistent with our prior expectations.
These results show that extrapolation tends to outperform transitioning to a higher-order product formula does not necessarily lead to an improved scaling for the problem of phase estimation.

\begin{figure*}
\includegraphics[width=0.5\textwidth]{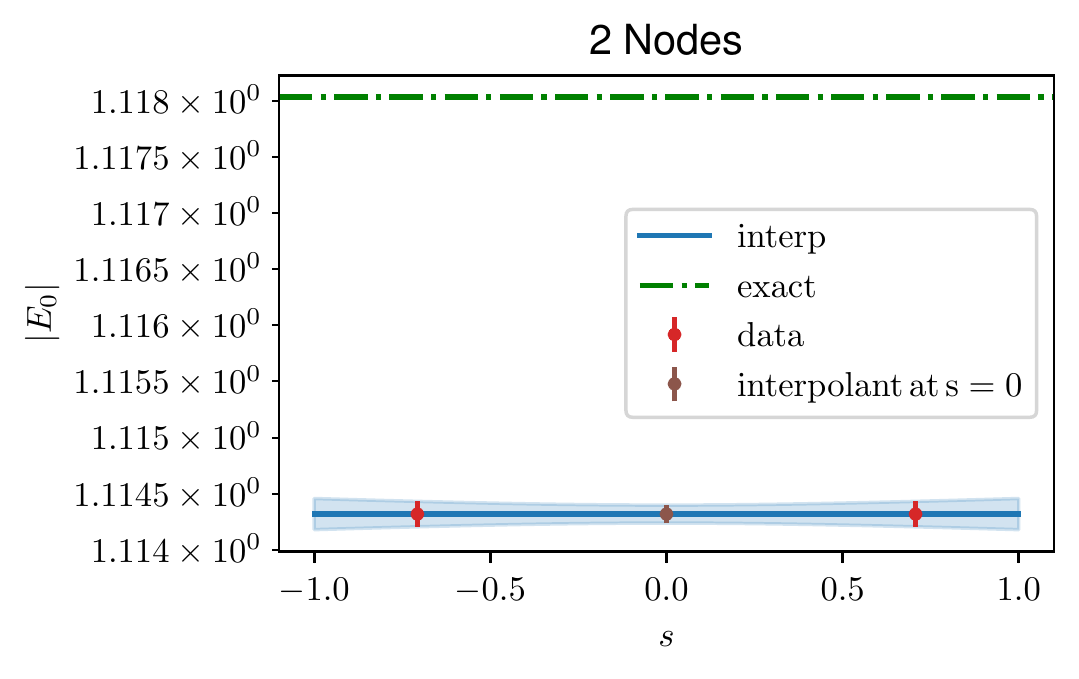}   
\includegraphics[width=0.5\textwidth]{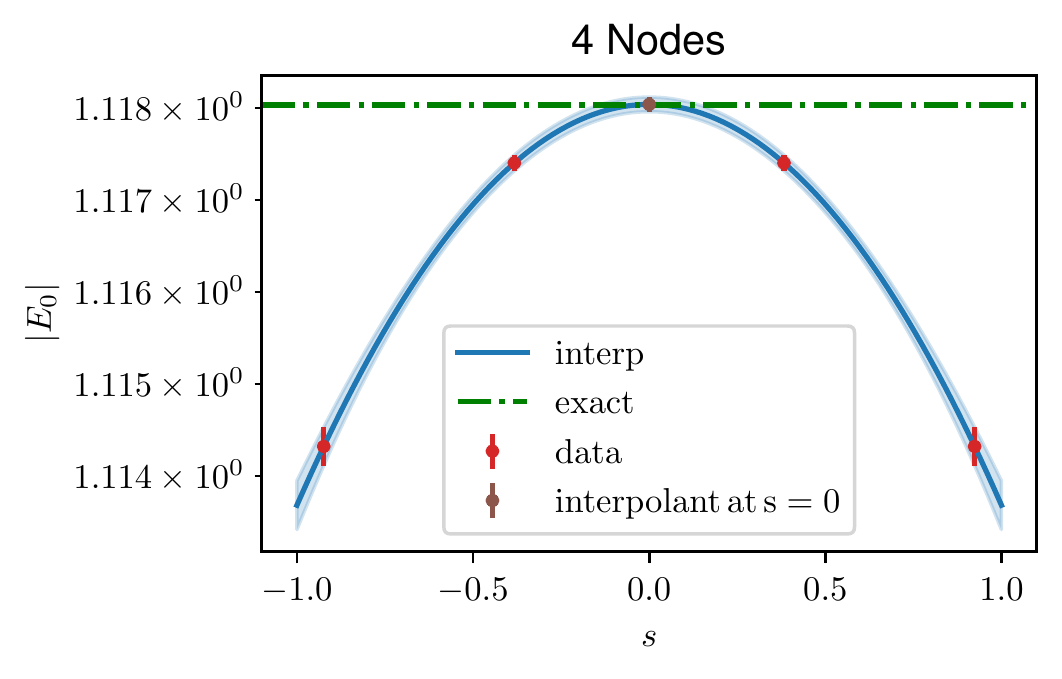}
\includegraphics[width=0.5\textwidth]{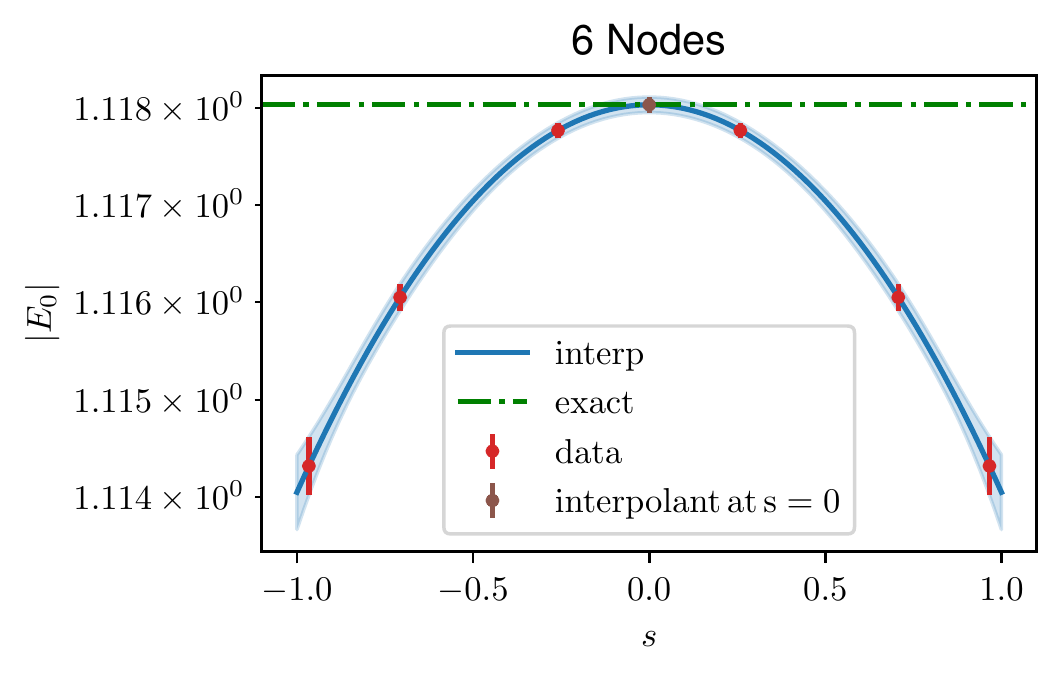}  
\includegraphics[width=0.5\textwidth]{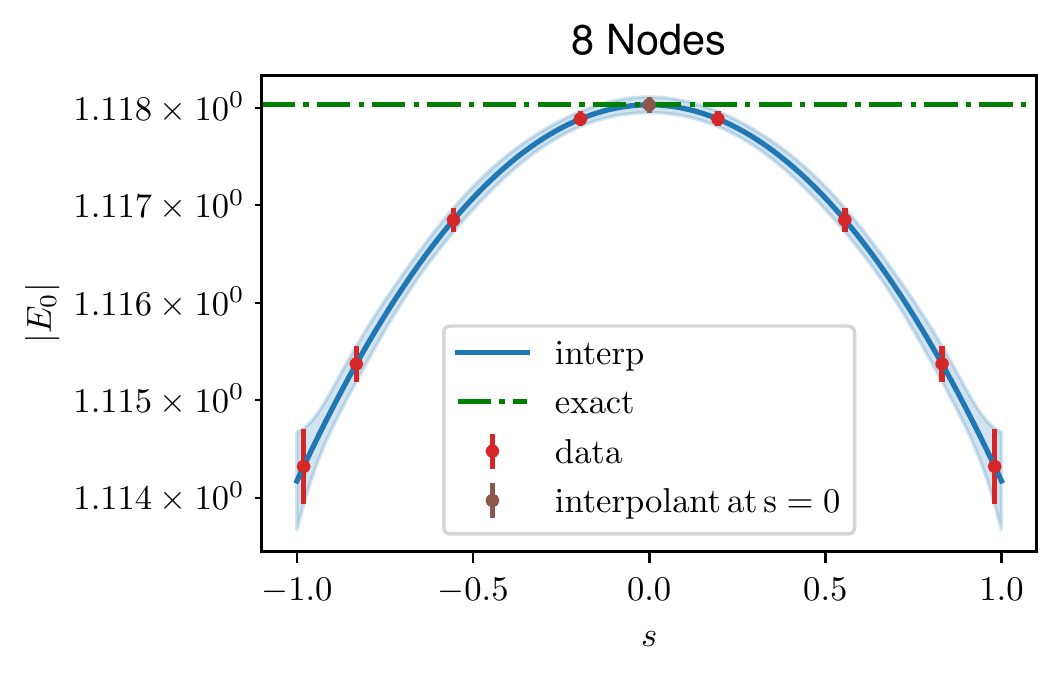}
\caption{Estimation of the ground state energy of the two-spin transverse Ising model using interpolation with Gaussian phase estimation. The four plots, left to right and top to bottom, have $n = 2, 4, 6, 8$ interpolation points. In all, the second order Trotter formula was used for simulation. The statistical uncertainties on the data and the interpolants shown here are equivalent to $0.01$ times their standard deviation. They are scaled this way so that interpolation error changes are noticeable. \label{fig:extrapol_second_order}}
\end{figure*}

\begin{figure*}
\centering
\includegraphics[width=0.45\textwidth]{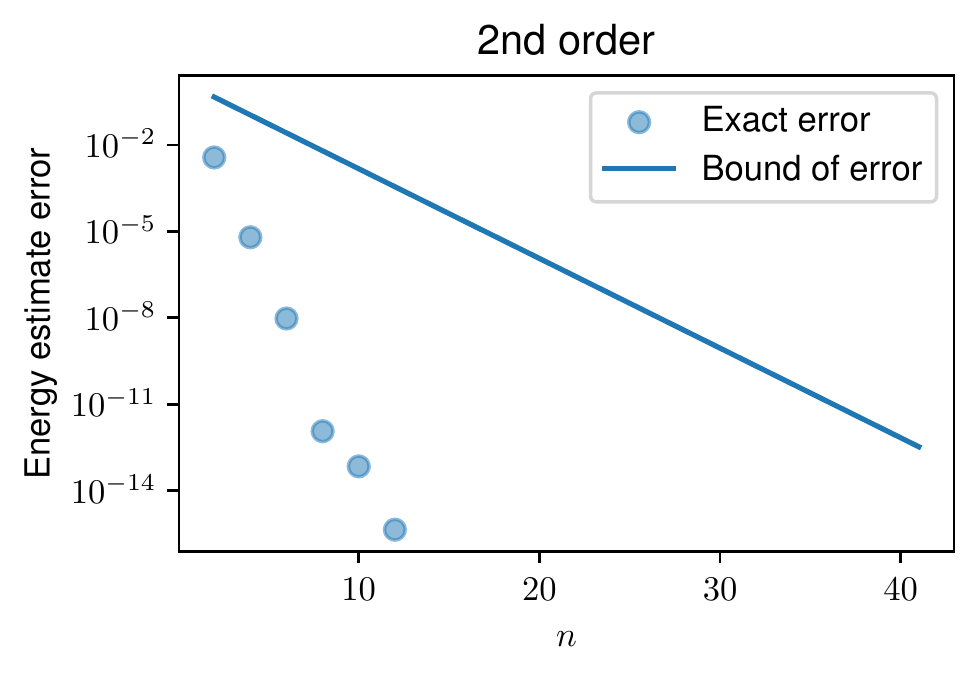}
\includegraphics[width=0.45\textwidth]{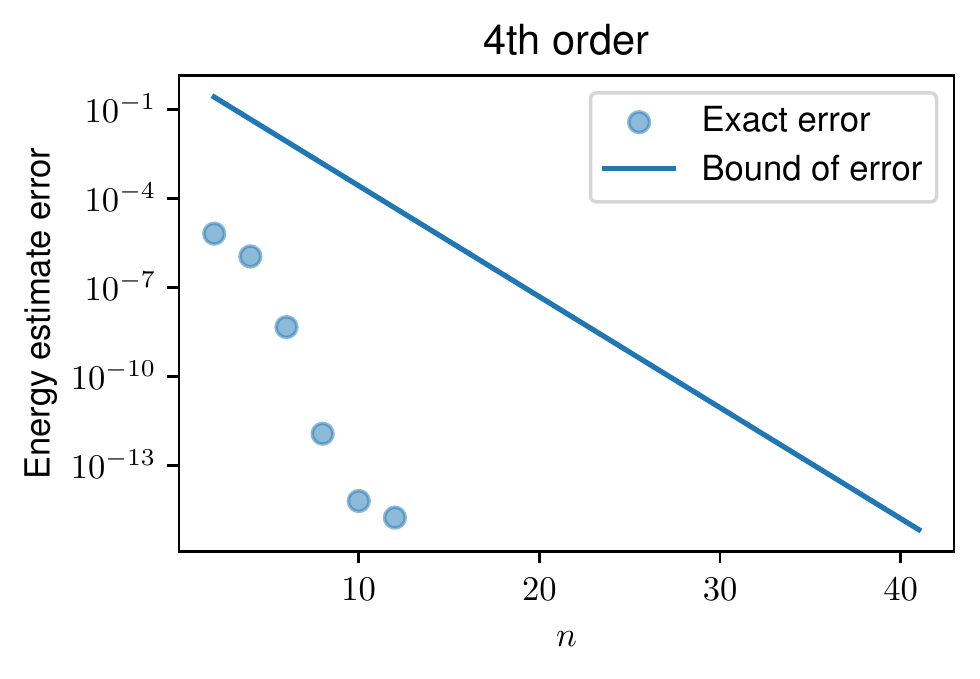}
\caption{The systematic error coming from truncating the polynomial expansion. In blue we have the exact result for the error and in orange are the upper bounds for the error. On the left, we used a second-order formula and on the right a fourth-order one. As we can see from this log-log plot, the convergence is exponential as expected.\label{fig:truncation_error}}
\end{figure*}
\begin{figure*}
\centering
\includegraphics[width=0.45\textwidth]{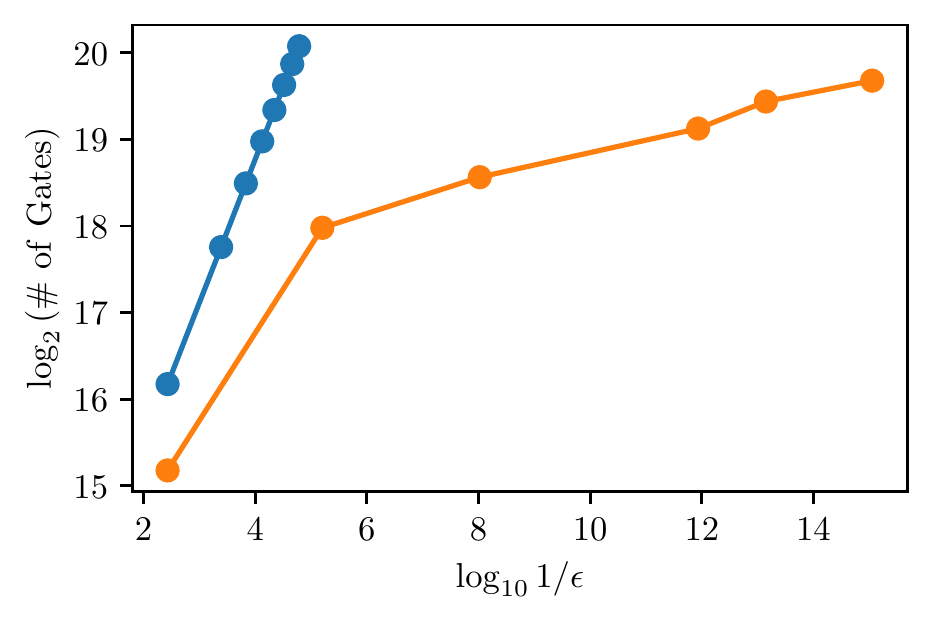}
\includegraphics[width=0.45\textwidth]{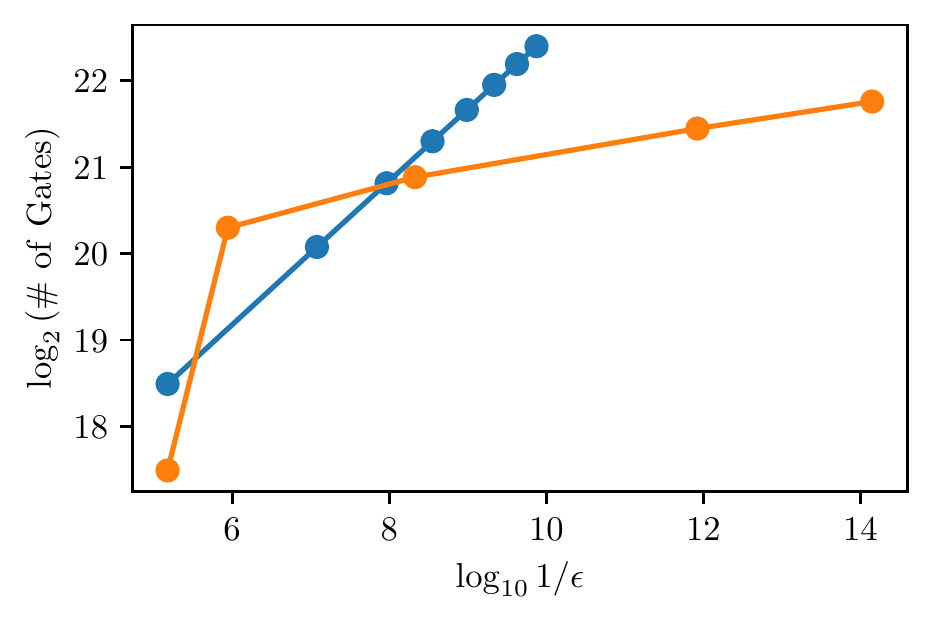}
\caption{The gate cost from using second- (left) and fourth-order (right) formulas for ground state estimation plotted against the systematic error on the eigenvalue estimation. In blue, we have the cost of using single long-depth circuits using product formulas alone. In orange, we have the gate cost of using multiple shallower circuits and then using the extrapolation methods presented in this work.\label{fig:Trotter_vs_extrapol}}
\end{figure*}

\section{Conclusion}
\label{sec:conclusion}
In this work, we present a framework for achieving improved accuracy of Trotter simulations using polynomial interpolation. We apply our framework to the problems of eigenvalue and expectation value estimation, and derive query complexities assuming the Trotter data is computed quantumly. We find that our methods achieve a Heisenberg scaling of $\tilde{O}(1/\epsilon)$ up to polylogarithmic factors. This is an improvement relative to a single Trotter calculation, which would yield $\tilde{O}(1/\epsilon^{1 + 1/p})$ for order $p$ formulas. For the low-order formulas used in practice, where $p$ is not so large, this improvement may be significant.

An interesting feature of polynomial interpolation in this context is that it provides a uniform approximation of the function $f(s)$ on the interpolation interval, thereby giving a characterization of the Trotter error for various step sizes. In contrast, similar techniques such as Richardson extrapolation only attempt to estimate $f(0)$. Whether this global characterization of the Trotter error can be usefully employed is a question we leave open.

There are several areas in which improvements in our analysis can likely be made. First, in most of our claims, we resort to fractional queries in order to produce step sizes at the values of the Chebyshev nodes. We suspect this machinery is not needed in practice, and that improved analysis, coupled with mathematical results on the robustness of the Lesbesgue constant to node perturbations, will show that interpolating at nearby points will be sufficient. Second, the quadratic dependence in simulation time that arises for expectation value interpolation in certain regimes is unsavory. The root cause of this in the analysis appears to be a linear relation between the simulation time and number of interpolation nodes, which seems overly cautious. Numerical studies beyond our own, as well as implementations on quantum hardware, may reveal a rosier picture than our naive bounds suggest. Third, our work does not consider the effects of imperfect operations. We can reasonably expect quantum operational errors to be more severe for the foreseeable future. In this case, the most salient effect is additional error on the data values that, following any mitigation procedures, cannot be arbitrarily diminished. We should expect that this will limit the interpolation accuracy to a certain threshold given roughly by the product $\epsilon_\mathrm{phys} \times \Lambda$, where $\epsilon_\mathrm{phys}$ is the typical error from imperfect operations and $\Lambda$ is the Lesbesgue constant. Finally, our results do not show commutator scaling, though the method itself almost certainly does. For Hamiltonians with almost commuting terms, this would manifest in the data being nearly flat, resulting in an approximately straight line fit to the correct answer. Our difficulty in showing this expected relationship arises from the messy expressions for high order derivatives of the effective Hamiltonian, but this should not be a fundamental issue.

Our work shows that polynomial interpolation can be a powerful supplement to quantum subroutines using low-order Trotter simulations. Given that our results ameliorate the worst aspects of low order formulas, we believe that investigation of highly efficient low-order formulas for simulation is likely to be an increasingly important topic.  Improvements that arise from the superior constant factors that low order formulas can provide may lead to a new generation of more efficient algorithms for simulating chemistry and physics using quantum computers.

\section{Acknowledgements}
\label{sec:acknowledgements}
We thank Christoph Lehner, Taku Izubuchi and Yuta Kikuchi for the insightful discussions. We also thank the reviewers for improving this article with specific feedback. During the inception of this work G. R. was supported through Brookhaven National Laboratory (BNL), the Laboratory Directed Research and Development (LDRD) program No. 21-043, by Program Development Fund No. NPP PD 19-025, and by the US DOE Office of Nuclear Physics, under the Contract No. DE-SC0012704 (BNL). During later additions to the manuscript, G.R. and N.W. were supported by the US DOE National Quantum Information Science Research Centers, Co-design Center for Quantum Advantage (C2QA) under contract number DE-SC0012704. G.R. was also supported and employed by Zapata Computing, Inc. through part of the writing of the manuscript. J.W. is supported by the NSF Graduate Research Fellowship and U.S. Department of Energy grant No. DE-SC0023658.

\FloatBarrier

\bibliographystyle{JHEP}
\bibliography{bibliography.bib}

\FloatBarrier

\appendix

\section{Proof of \Cref{lem:Heffderiv}: derivatives of effective Hamiltonian}
\label{app:setup_proofs}
\begin{proof}
Recall the definition of the effective Hamiltonian
\begin{align}
    \tilde{H}_s := \frac{i}{st}\log S_{2k} (st)
\end{align}
for $s \in \mathbb{R}\setminus\{0\}$, with $\tilde{H}_0 := \lim_{s\rightarrow 0} \tilde{H}_s = H$.
We will understand $\log S_{2k}(st)$ through a power series expansion about the identity.
\begin{align}\label{eq:logS_expansion}
    \log S_{2k} (st) = \sum_{j=0}^{\infty} \frac{(-1)^j}{j+1} (S_{2k}(st) - \openone)^{j+1}
\end{align}
This series converges precisely when
\begin{align}
    \norm{S_{2k}(st) - \openone} \leq 1.
\end{align}
Using the fundamental theorem of calculus, we can derive a suitable condition for convergence as a neighborhood about $s = 0$. The condition above implies
\begin{align}
    \norm{\int_0^{st} S_{2k}(x) dx} \leq 1
\end{align}
which is satisfied provided
\begin{align}
    \abs{st} \max_{x\in[0,st]}\norm{\frac{d}{dx} S_{2k}(x)} \leq 1.
\end{align}
Writing out
$S_{2k}(x) = \prod_{l = 1}^{N_k} \exp(-i H_{j_l} \tau_l x)$
where $H_{j_l}$ is some Hamiltonian piece $H_j$ indexed by $l$, the derivative can be upper bounded as
\begin{align}
\begin{aligned}
    \max_{x\in[0,st]}\norm{\frac{d}{dx} S_{2k}(x)} &\leq \sum_{l=1}^{N_k} \norm{H_{j_l}} \abs{\tau_l} \\
    &\leq \max_{j} \norm{H_j} \norm{\tau}_1
\end{aligned}
\end{align}
where $\tau = (\tau_l)_{l=1}^{N_k}$ is the vector of ST coefficients, and in going to the second line we used a H\"{o}lder inequality. We have $\norm{\tau_l}_1 \leq N_k \max_l \abs{\tau_l}$, and from Appendix A of~\cite{wiebe2010higher} we have
\begin{align}
    \max_l \abs{\tau_l} \leq 2k/3^{k}.
\end{align}
Thus, the requirement for convergence of the logarithm becomes
\begin{align} \label{eq:log_converge}
    \frac{4}{3} k (5/3)^{k-1} m \abs{st} \max_j \norm{H_j} \leq 1
\end{align}
where we used the expression $N_k = (2m)5^{k-1}$ for the number of ST exponentials. 

We now assume $s$ is within the symmetric interval defined by~\eqref{eq:log_converge}, such that~\eqref{eq:logS_expansion} is convergent. Since $\log S_{2k} (0) = 0$, $s=0$ is a zero of order at least one. We want to absorb the diverging $1/s$ term and better understanding the leading dependence in $s$. To facilitate this, we write
\begin{align} \label{eq:expand_H}
    \tilde{H}_s = -\frac{1}{it}\sum_{j=0}^\infty \frac{(-1)^j}{j+1} s^j \Delta S_{2k}(st)^{j+1}
\end{align}
where we defined 
\begin{align}
    \Delta S_{2k}(st) := \frac{S_{2k}(st) - \openone}{s}.
\end{align}
Note that $\Delta S_{2k}$ is analytic in $s$, and is a finite difference around $s = 0,$ such that 
\begin{align}
    \lim_{s\rightarrow 0} \Delta S_{2k} (st)= -i H t.
\end{align}
Through the series expansion~\eqref{eq:expand_H} we may bound derivatives of $\tilde{H}_s$ via bounds on derivatives of $\Delta S_{2k}$.
We first obtain a power series of $\Delta S_{2k}$ by Taylor expanding every term in the product formula $S_{2k}$. Regrouping in powers of $st$, the result is
\begin{align}
    \Delta S_{2k} (st) = \sum_{j = 1}^\infty \frac{s^{j-1} (-it)^j}{j!} \sum_J \binom{j}{j_1 \dots j_{N_k}} \prod_{l=1}^{N_k} (H_l \tau_l)^{j_l}
\end{align}
where the parenthetical symbol is the multinomial coefficient, and the sum $\sum_J$ is over all values of $J = (j_1, \dots, j_{N_k})$ such that $\sum_k j_k = j$. The derivatives with respect to $s$ are now easy to compute. Using the fact that 
\begin{align}
    \partial_s^n s^{j-1} = \frac{(j-1)!}{(j-1-n)!} s^{j-n-1}
\end{align}
for $j > n$ (and zero otherwise), we have
\begin{align} 
\begin{aligned}
    \partial_s^{n} \Delta S_{2k} (st) &= \sum_{j=n+1}^\infty \frac{s^{j-n-1}(-it)^j }{j!} \frac{(j-1)!}{(j-1-n)!}\sum_J \binom{j}{j_1 \dots j_{N_k}} \prod_{l=1}^{N_k} (H_l \tau_l)^{j_l} \\
    \norm{\partial_s^{n} \Delta S_{2k} (st) } &\leq \sum_{j=n+1}^\infty\frac{t^j}{(j-n-1)!} s^{j-n-1} (\tau_\mathrm{max} N_k \Lambda )^j \\
\end{aligned}
\end{align}
where $\Lambda := \max_j \norm{H_j}$ and $\tau_\mathrm{max} = \max_l \abs{\tau_l}$. Factoring out powers of $n+1$ and reindexing, we are left with the following bound on derivatives of $\Delta S_{2k}$.
\begin{align} \label{eq:Delta_S_deriv}
    \norm{\partial_s^n \Delta S_{2k} (st)} \leq (\tau_\mathrm{max} N_k \Lambda t)^{n+1} e^{s \tau_\mathrm{max} N_k \Lambda t}
\end{align}
This expression is quite elegant; it is as if we were taking $n+1$ derivatives of the exponential $e^{cs}$ with 
\begin{align} \label{eq:cdef}
\begin{aligned}
    c &:= \tau_\mathrm{max} N_k \Lambda t \\
    &\leq k(5/3)^k m \Lambda t
\end{aligned}
\end{align}
Factors of $c$ will occur frequently in what follows, so we find it convenient to adopt this symbol as shorthand.

We return to bounding the derivatives of powers of $\Delta S_{2k} (st)$ as in equation~\eqref{eq:expand_H}.
\begin{align}
    \partial_s^n \left[\Delta S_{2k}(st)^{j+1}\right]
\end{align}
We reduce this to the previous case by performing a multinomial expansion.
\begin{align}
    \partial_s^{n} \Delta S_{2k} (st)^{j+1} = \sum_N \binom{n}{n_0\dots n_j} \prod_{l=0}^{j} \partial_s^{n_l} \Delta S_{2k}(st)
\end{align}
As usual, the capital letter $N$ denotes the set of all nonnegative indices $n_0, \dots, n_j$ summing to $n$. Applying the triangle inequality and submultiplicativity, and employing the bound~\eqref{eq:Delta_S_deriv},
\begin{align}
\begin{aligned}
    \norm{\partial_s^{n} \Delta S_{2k} (st)^{j+1}} &\leq  \sum_N \binom{n}{n_0\dots n_j} \prod_{l=0}^j \norm{\partial_s^{n_l} \Delta S_{2k}(st)} \\
    &\leq \sum_N \binom{n}{n_0\dots n_j} \prod_{l=0}^j  c^{n_l + 1} e^{cs} \\
    &=  e^{(j+1)cs} c^{n + j + 1} \sum_N \binom{n}{n_0\dots n_j},
\end{aligned}
\end{align}
where we've used the sum property of the $n_l$ where appropriate. The remaining sum over the multinomial coefficient is given by $(j+1)^n$. Hence, 
\begin{align} \label{eq:deriv_S_powers}
    \norm{\partial_s^{n} \Delta S_{2k} (st)^{j+1}} &\leq  ((j+1)c)^n (c e^{cs})^{j+1}
\end{align}
Notice that, when $j = 0$, this is consistent with equation~\eqref{eq:Delta_S_deriv}. 

With result~\eqref{eq:deriv_S_powers} in hand, we return to the power series~\eqref{eq:expand_H}. Differentiating term by term
\begin{align}
    \partial_s^n \tilde{H}_s = -\frac{1}{it}\sum_{j=0}^{\infty} \frac{(-1)^j}{j+1} \partial_s^n\left(s^j \Delta S_{2k}(st)^{j+1}\right)
\end{align}
and performing a binomial expansion for each term
\begin{align} \label{eq:binom_sjS2k}
    \partial_s^n\left(s^j \Delta S_{2k}(st)^{j+1}\right) = \sum_{q=0}^n \binom{n}{q} \left(\partial_s^{q} s^j\right) \left(\partial_s^{n-q} \Delta S_{2k} (st)^{j+1}\right)
\end{align}
will allow us to apply our previous results. It will be helpful to consider two cases separately: $j \leq n$ and $j > n$. These regimes are somewhat qualitatively different, since the derivatives of $s^j$ may or may not vanish depending on the number of derivatives. Focusing on the case $j \leq n$, we have
\begin{align} \label{eq:binom}
    \partial_s^n\left(s^j \Delta S_{2k}(st)^{j+1}\right) = \sum_{q=0}^j \binom{n}{q} \frac{j!}{(j-q)!} s^{j-q} \left(\partial_s^{n-q} \Delta S_{2k} (st)^{j+1}\right).
\end{align}
Note that the sum runs only to $j$, not $n$. Taking a triangle inequality upper bound using~\eqref{eq:deriv_S_powers}, we may upper bound~\eqref{eq:binom} as
\begin{align}
\begin{aligned}
    &\sum_{q=0}^j \binom{n}{q} \frac{j!}{(j-q)!} s^{j-q}  ((j+1) c)^{n-q} (c e^{cs})^{j+1} \\
    &= (ce^{cs})^{j+1} \sum_{q=0}^j \binom{j}{q} \frac{n!}{(n-q)!} s^{j-q}((j+1)c)^{n-q}
\end{aligned}
\end{align}
where we have factored out terms not involving $q$ from the sum, and manipulated the factorials for reasons which will be seen presently. Taking the upper bound $n!/(n-q)! < n^q$, and factoring out $n-j$ powers of $(j+1) c$, we may upper bound the above expression by
\begin{align}
\begin{aligned}
    &(c e^{cs})^{j+1} ((j+1) c)^{n-j} \sum_{q=0}^j \binom{j}{q} n^q ((j+1) cs)^{j-q} \\
    &= (c e^{cs})^{j+1} ((j+1) c)^{n-j} (n + (j+1)cs)^j
\end{aligned}
\end{align}
Thus, with some minor polishing, we may express the bound on~\eqref{eq:binom_sjS2k} for $j \leq n$ as
\begin{align}\label{eq:j_smaller_n_result}
    \norm{\partial_s^n\left(s^j \Delta S_{2k}(st)^{j+1}\right)} \leq e^{(j+1)cs} c^{n+1} (j+1)^n \left(\frac{n}{j+1} + cs\right)^j.
\end{align}

Now let's move on to the $j > n$ case. Here, there are not enough derivatives to kill off the $s^j$ term, so the binomial sum in~\eqref{eq:binom} will run from $q = 0$ to $n$. 
\begin{align} 
    \partial_s^n\left(s^j \Delta S_{2k}(st)^{j+1}\right) = \sum_{q=0}^n \binom{n}{q} \frac{j!}{(j-q)!} s^{j-q} \left(\partial_s^{n-q} \Delta S_{2k} (st)^{j+1}\right)
\end{align}
Similar to before, we use the bound~\eqref{eq:deriv_S_powers}, to obtain
\begin{align}
\begin{aligned}
    \norm{\partial_s^n\left(s^j \Delta S_{2k}(st)^{j+1}\right)} &\leq \sum_{q=0}^n \binom{n}{q} \frac{j!}{(j-q)!} s^{j-q} ((j+1)c)^{n-q} (ce^{cs})^{j+1}\\
    &= (c e^{cs})^{j+1} s^{j-n} \sum_{q=0}^n \binom{n}{q} \frac{j!}{(j-q)!} ((j+1) cs)^{n-q}. \\
\end{aligned}
\end{align}
Taking $j!/(j-q)! < j^q$, a simpler upper bound is given by
\begin{align}
    (c e^{cs})^{j+1} s^{j-n} \sum_{q=0}^n \binom{n}{q} j^q ((j+1) cs)^{n-q}
    &= (c e^{cs})^{j+1} s^{j-n} (j + (j+1) cs)^n.
\end{align}
With some minor rearrangements, this gives the following upper bound for $j > n$.
\begin{align} \label{eq:j_larger_n_result}
    \norm{\partial_s^n\left(s^j \Delta S_{2k}(st)^{j+1}\right)} \leq e^{(j+1) cs} c^{n+1} (j+1)^n (cs)^{j-n} \left(\frac{j}{j+1} + cs\right)^n
\end{align}

With the bounds~\eqref{eq:j_smaller_n_result} and~\eqref{eq:j_larger_n_result}, we can return to bounding $\partial_s^n \tilde{H}_s$. Still separating the two cases $j \leq n$ and $j > n$, we can write
\begin{align}
\begin{aligned}
    \norm{\partial_s^n \tilde{H}_s} t &\leq \sum_{j=0}^n \frac{1}{j+1} \norm{\partial_s^n \left(s^j \Delta S_{2k}(st)^{j+1}\right)} + \sum_{j=n+1}^\infty \frac{1}{j+1} \norm{\partial_s^n \left(s^j \Delta S_{2k}(st)^{j+1}\right)} \\
    &= B_l + B_h
\end{aligned}
\end{align}
where $B_l$ and $B_h$ refer to bounds on the ``low'' and ``high'' parts of the series. Employing the bounds from equations~\eqref{eq:j_smaller_n_result} and~\eqref{eq:j_larger_n_result}, we have
\begin{align} \label{eq:B_low}
\begin{aligned}
    B_l &\leq \sum_{j=0}^{n} \frac{1}{j+1} e^{(j+1) cs}c^{n+1}(j+1)^n \left(\frac{n}{j+1} + cs\right)^j \\
    &= c^{n+1} \sum_{j=0}^n e^{(j+1) cs} (j+1)^{n-1} \left(cs + \frac{n}{j+1}\right)^j
\end{aligned}
\end{align}
and
\begin{align} \label{eq:B_high}
\begin{aligned}
    B_h &\leq \sum_{j=n+1}^\infty \frac{1}{j+1} e^{(j+1) cs} c^{n+1} (j+1)^n (cs)^{j-n}\left(\frac{j}{j+1}+ cs\right)^n \\
    &\leq c^{n+1} \sum_{j=n+1}^\infty e^{(j+1) cs} (j+1)^{n-1} (cs)^{j-n} \left(1 + cs\right)^n \\
    &= c^{n+1} \left(1 + cs\right)^n \sum_{j=n+1}^\infty e^{(j+1) cs} (j+1)^{n-1} (cs)^{j-n}.
\end{aligned}
\end{align}
Let's begin by simplifying the bound on $B_l$. We will at this point make the assumption that $s$ is sufficiently small such that $cs<1$. This will necessarily factor into the cost later. This simplification yields
\begin{align}
\begin{aligned}
    B_l &\leq c^{n+1} \sum_{j=0}^n e^{j+1} (j+1)^{n-1} \left(1 + \frac{n}{j+1}\right)^j \\
    &\leq c^{n+1} \sum_{j=0}^n e^{j+1} (j+1)^{n-1} e^n \\
    &\leq \frac{1}{e}(e^2 c)^{n+1} \sum_{j=1}^{n+1} j^{n-1} 
\end{aligned}
\end{align}
The remaining sum can be bounded by $(n+1)^n$, hence,
\begin{align}
    B_l \leq \frac{(e^2 (n+1) c)^{n+1}}{e(n+1)} \leq (e^2 c)^{n+1} n^n,
\end{align}
where the definition that $0^0 = 1$ handles the edge case. Let's turn our attention to $B_h$. We will start by reindexing so that the series begins at $j=0$ in~\eqref{eq:B_high}.
\begin{align} \label{eq:B_high_in_S}
    B_h &\leq c^{n+1} (1+cs)^n \sum_{j=0}^\infty e^{(j+n+2) cs} (j+n+2)^{n-1} (cs)^{j+1} \\
    &=(ce^{cs})^{n+1} (1+cs)^n \sum_{j=0}^\infty (cs e^{cs})^{j+1} (j+n+2)^{n-1}
\end{align}
The series converges if and only if
\begin{align}\label{eq:theorem_condition}
    cse^{cs} < 1.
\end{align}
This condition is slightly stronger than the condition~\eqref{eq:log_converge} that we need for convergence of the logarithm~\eqref{eq:logS_expansion}, and is equivalent to $cs < W(1) \approx 0.567$, where $W$ is the principal brach of the Lambert W function. Returning to~\eqref{eq:B_high_in_S}, we have the bound
\begin{align}
    (j+n+2)^{n-1} = (n+2)^{n-1} \left(1 + \frac{j}{n+2}\right)^{n-1} \leq (n+2)^{n-1} e^j.
\end{align}
Thus, we have
\begin{align}
\begin{aligned}
    B_h &\leq (c e^{cs})^{n+1} (1+ cs)^n (n+2)^{n-1} \sum_{j=0}^\infty (e cs e^{cs})^j \\
    &= (c e^{cs})^{n+1} (1+ cs)^n (n+2)^{n-1} \frac{1}{1- e cs e^{cs}}.
\end{aligned}
\end{align}
To be concrete, let's take $ecse^{cs}< 1/2$, which is implied by $cs<\pi/20$. Coupled with the inequality in~\eqref{eq:cdef}, this condition can be met provided that
\begin{align}
    k (5/3)^k m \Lambda st \leq \pi/20,
\end{align}
which is exactly the assumption of~\Cref{lem:Heffderiv}. This allows us to upper bound $B_h$ further as
\begin{align}
    B_h \leq 2e^{\pi/20} (cs)^{n+1} (3e^{\pi/20}/2)^n (n+2)^{n-1}\le 4(cs)^{n+1} (9/5)^{n} (n+2)^{n-1}.
\end{align}
Since $(n+2)^{n-1} \leq e^2 n^n /2$ (using $0^0 := 1$ for the edge case $n= 0$), we have
\begin{align}
    B_h \leq 2e^2 (cs)^{n+1}\left(9/5\right)^n n^n.
\end{align}
Altogether, using $s\le 1$
\begin{align} \label{eq:H_deriv_finalbound}
\begin{aligned}
    \norm{\partial_s^n \tilde{H}_s} t &\leq n^n (e^2 c)^{n+1} \left(1 + 2 (9/5e^2)^{n} \right) \\
    &\leq 2n^n (e^2 c)^{n+1} 
\end{aligned}
\end{align}
The final result then follows from substituting for $c$ and noting that the duration of each time step is at most $2k/3^{k-1}$ using the results of~\cite{wiebe2011simulating}. 
\end{proof}

\section{Proof of~\Cref{thm:PEMain}}
\label{app:PE_proof}
In order to use the result of Lemma~\ref{lem:Heffderiv} in the proof of Theorem~\ref{thm:PEMain}, we need to re-express the derivatives of the eigenvalues in terms of the effective Hamiltonian $\tilde{H}_t$ (we use $t$ rather than $s$ due to the extrapolation parameter taking the role of time).  The following lemma provides such a bound on the derivatives under the assumption that the Hamiltonian is gapped.
\begin{lemma}\label{lem:Lambda}
Let $\tilde{H}_t:\mathbb{R}\mapsto \mathbb{C}^{2^n\times 2^n}$ be a Hamiltonian on $n$ qubits with $p$ continuous derivatives. Suppose  $\tilde{H}_t$ has a minimum eigenvalue gap
$$
\gamma(t):= \min_{k\ne j}|\lambda_k(t) - \lambda_j(t)|
$$ 
where $\lambda_j(t)$ and $\lambda_k(t)$ are eigenvalues of $\tilde{H}_t$. Define
$$
\Lambda_\gamma:= \max_p\sup_t\left(\|\partial_t^p \tilde{H}_t\|^{1/(p+1)},\left( \frac{\|\partial_t^p \tilde{H}_t\|}{\gamma(t)} \right)^{1/p} \right),
$$
which has dimensions of $1/t$. Then for all eigenvalues $\lambda_k(t)$
$$
|\partial_t^p \lambda_k(t)| \le 8^{p-1}(p-1)!\Lambda_\gamma^{p+1}.
$$
\end{lemma}
\begin{proof}
The essence of this proof is to examine the derivatives of $\lambda_k(t)$ by expressing them in terms of derivatives of $\tilde{H}_t$ via perturbation theory. Using the bound~\Cref{lem:Heffderiv} on $\tilde{H}_t$ lets us bound the eigenvalues in turn. Let $\ket{k(t)}$ be the eigenvector for eigenvalue $\lambda_k(t)$ and assume that the eigenvectors' phases are chosen such that $\braket{k(t)}{\dot{k}(t)}=0$ for all $t$. Such a choice is always possible.
Standard results from perturbation theory~\cite{horn2012matrix} show  that
\begin{align} \label{eq:eigenderiv}
\begin{aligned}
    \partial_t \lambda_k(t) &= \bra{k(t)} (\partial_t \tilde{H}_t) \ket{k(t)} \\
    \partial_t \ket{k(t)} &= \sum_{j\ne k} \ket{j(t)} \frac{ \bra{j(t)} \partial_t \tilde{H}_t \ket{k(t)}}{\lambda_k(t)-\lambda_j(t)}.
\end{aligned}
\end{align}
Further, taking the Euclidean norm of the 2nd equation, 
\begin{equation} \label{eq:jeffbd}
    \|\partial_t \ket{k(t)}\|^2 \le \sum_{j\ne k} \left|\frac{ \bra{j(t)} \partial_t \tilde{H}_t \ket{k(t)}}{\lambda_k(t)-\lambda_j(t)}\right|^2\le \frac{\|\partial_t \tilde{H}_t\|^2}{\gamma^2(t)}.
\end{equation}
where $\gamma(t)$ is the minimum eigenvalue gap of the effective Hamiltonian at time $t$.

We now consider second- and higher-order derivatives of $\lambda_k(t)$. By repeatedly taking derivatives of $\lambda_k(t)$ in the first line of~\eqref{eq:eigenderiv}, we obtain an expression of the general form
\begin{equation}
    \partial_t^q \lambda_k(t) = \sum_{\ell=1}^q\sum_{x,y} c_{\vec{x},\vec{y},\ell}  \bra{x_{1}(t)} \partial_t^{y_\mu} \tilde{H}_t \ket{x_{2}(t)}\prod_{\mu=2}^\ell\frac{\bra{x_{2\mu-1}(t)} \partial_t^{y_\mu} \tilde{H}_t \ket{x_{2\mu}(t)}}{(\lambda_{x_{2\mu}}(t)-\lambda_{x_{2\mu-1}}(t))}.
\end{equation}
Note that there is a pattern to this series.  If there are $m$ $H^{(p)}$ (for $p\ge 0$) present in the expansion there must be $m-1$ powers of the gaps present in the terms.  Before going into more detail let us define the ``degree'' of a term to be the number of differentiable factors (eigenvectors, Hamiltonians or inverse gaps) present in a product.  For example $\bra{x(t)}\partial_t \tilde{H}_t \ket{x(t)}$ is degree $3$.

We have that if $A$ is a term of degree ${\mathrm{deg}}(A)$.  Then it follows from~\eqref{eq:eigenderiv} that 
\begin{equation}
    {\mathrm{deg}}(\partial_t A) \le {\mathrm{deg}}(A)+4,
\end{equation}
since from the product rule the derivative of the product of the factors is the sum of the distributed sum of the derivatives of the factors and that if $\ket{\lambda_k}$ or $(\lambda_j -\lambda_k)^{-1}$ are differentiated then the degree increases by at most $4$.  While the derivative of the effective Hamiltonian does not increase the degree, the increase in the degree is still at most $4$ in this case thus in the worst case scenario the degree is increased by $4$.  As the degree ${\mathrm{deg}}(\partial_t \lambda_k(t)) = 3\le 4$, it is straight forward to see that 
\begin{equation}
    {\mathrm{deg}}(\partial_t^q \lambda_k(t)) = 4q\label{eq:degeq}
\end{equation}

Next, let us assume that $B$ is a sum of $N_{\mathrm{terms}}(B)$ products of the above factors.  We then have from~\eqref{eq:eigenderiv} that
\begin{equation}
    N_{\mathrm{terms}}(\partial_t B) \le 2N_{\mathrm{terms}}(B) {\mathrm{deg}}(B).\label{eq:ntermseq}
\end{equation}
Next we will show that the number of terms present in $\partial_t^q \lambda_k(t)$ obeys 
\begin{equation}
    N_{\mathrm{terms}}(\partial_t^q \lambda_k(t)) \le 8^{q-1}(q-1)!.
\end{equation}
We prove this by induction.  The number of terms present when $q=1$ is $1$.  This demonstrates the base case.  Assume that for some value  $q'$ that the induction hypothesis holds.  Then from~\eqref{eq:degeq} and~\eqref{eq:ntermseq} that
\begin{equation}
    N_{\mathrm{terms}}(\partial_t^{q'+1} \lambda_k(t))\le 2(8^{q'-1}(q'-1)!)(4q') = 8^{q'} q'!
\end{equation}
which proves the relation by induction.  

We then see that if any factor is differentiated $\nu$ times then its norm is bounded above by $\Lambda_\gamma^{\nu+1}$ where
\begin{equation}\label{eq:lambdaDef}
    \Lambda_\gamma:= \max_p\left(\max\left\{\|\partial_t^p \tilde{H}_t\|^{1/(p+1)},\left( \frac{\|\partial_t^p \tilde{H}_t\|}{\gamma(t)} \right)^{1/p}\right\} \right).
\end{equation}
We therefore have that since the maximum value of the derivative is the number of terms multiplied by the maximum norm of the product of all the factors that
\begin{equation}
    |\partial_t^q \lambda_k(t)| \le 8^{q-1} (q-1)! \Lambda_\gamma^{q+1}.
\end{equation}
\end{proof}

\begin{proof}[Proof of Theorem~\ref{thm:PEMain}]
First, from the Baker Campbell Hausdorff formula, 
$$
S_{2k}(t) = \exp\left(-iHt+ O(t^{2k+1})\right)
$$ 
and thus $\log(S_{2k}(t))/t \in O(1)$ as $t\rightarrow 0$. Thus, an eigenvalue $E_{\mathrm{eff}}(t)$ exists for $\tilde{H}_t$ for all $t>0$, demonstrating the first claim.

Next, since $\tilde{H}_t$ has $p$ continuous derivatives, $\partial_t^p \lambda_k(t)$ exists and is bounded above by $8^{p-1} (p-1)! \Lambda_\gamma ^{p+1}$ because
\begin{equation}
 \partial_t \tilde{H}_t = \frac{1}{t}\partial_s \tilde{H}_{(st)}.
\end{equation}
The value of $\Lambda_\gamma$ can then be bounded, using Lemma~\ref{lem:Heffderiv}, by
\begin{align}
    \Lambda_\gamma &\le \left(2t^{-n-1}n^n(e^2k(5/3)^{k-1}m\max_{l\in [1,m]} \|H_l\| t)^{n+1}\right)^{1/(n+1)}   \nonumber\\
    &\qquad+\max_p\left(\frac{2t^{-p-1}p^p(e^2k(5/3)^{k-1}m\max_{l\in [1,m]} \|H_l\| t)^{p+1}}{\gamma} \right)^{1/p}\nonumber\\
    &\le 4nke^2(5/3)^{k-1}  m\max_{l\in [1,m]} \|H_l\|)\left(1+ \max_p\left( \frac{p}{n}\right)\left(\frac{e^2k(5/3)^{k-1}m\max_{l\in [1,m]} \|H_l\|}{\gamma}\right)^{1/p} \right)\nonumber\\
    &:=4nke^2(5/3)^{k-1}  m\max_{l\in [1,m]} \|H_l\|)(1+\Gamma).
\end{align}
This implies that
\begin{equation}
    |\partial_t^q \lambda_k(t)| \le 8^{q-1} (q-1)! \Lambda_\gamma^{q+1} \le  \frac{q!q^{q}}{64}\left(32ke^2(5/3)^{k-1}  m\max_{l\in [1,m]} \|H_l\|)(1+\Gamma)\right)^{q+1}.
\end{equation}
For all $t\in [-a,a]$, we have from Lemma~\ref{lem:Cheb_error} that the interpolation error satisfies the following upper bound.
\begin{equation}
    \lvert\lambda_k(0) - P(0)\rvert \le \frac{n!}{64}\left(32ke^2(5/3)^{k-1}  m\max_{l\in [1,m]} \|H_l\|)(1+\Gamma) \right)^{n+1} \left(\frac{a}{2} \right)^n
\end{equation}
We wish for this error to be at most $\epsilon_\mathrm{int}$. Using the inequality above, we can solve for a value of $a$ such that the error is within $\epsilon_\mathrm{int}$ for all $t\in [-a,a]$. Doing so, we find the following value of $a$. 
\begin{equation} \label{eq:a_val}
    a= \frac{2(64\epsilon_\mathrm{int})^{1/n}}{(n!)^{1/n}}\left(\frac{1}{32ke^2(5/3)^{k-1}  m\max_{l\in [1,m]} \|H_l\|)(1+\Gamma)} \right)^{1+1/n}
\end{equation}
We now count the number of exponentials needed for the degree $n$ polynomial interpolation (where $n$ is even) using~\eqref{eq:a_val}. Using the fact that the error for time $-t$ is equal to the error for time $t$, that the $j^\mathrm{th}$ Chebyshev node $s_j$ scales like $a/j$, and that the cost of phase estimation within error $\epsilon_\mathrm{PE}$ with failure probability $\delta_\mathrm{PE}$ scales as $O(\log(1/\delta_\mathrm{PE})/\epsilon_\mathrm{PE})$, the number of exponentials scales as
\begin{align} \label{eq:Nexp_inter}
    N_\mathrm{exp}&\in 
    O\left(\sum_{j=1}^{n/2}\frac{m5^{k-1}}{ja}\right) \nonumber\\
    &= O\left(\frac{n \log(n)m5^{k-1}\log(1/\delta_\mathrm{PE})}{a\epsilon_\mathrm{PE}}\right)\nonumber\\
    &=O\left(\frac{n\log(n) \log(1/\delta_\mathrm{PE}) (n!)^{1/n}m^{2+1/n} (5/3)^{2k+k/n}(\max \|H_i\| (1+\Gamma) k )^{1+1/n}}{\epsilon_\mathrm{int}^{1/n}\epsilon_\mathrm{PE}} \right)\nonumber\\
    &=\widetilde{O}\left(\frac{n^2 m^{2+1/n} \log(1/\delta_\mathrm{PE}) 5^{k}(5/3)^{k+k/n}(\max \|H_i\| (1+\Gamma) k )^{1+1/n}}{\epsilon_\mathrm{int}^{1/n}\epsilon_\mathrm{PE}} \right).
\end{align}
There are two competing tendencies in the cost. The number of operator exponentials increases polynomially with $n$, but the scaling with the error tolerance improves exponentially with $n$.  Setting these two equal to each other to estimate the optimal scaling yields
\begin{align} \label{eq:match_n}
    n^2 = \left(\frac{mk(5/3)^k\max_i \|H_i\|(1+\Gamma)}{\epsilon_{\rm int}}\right)^{1/n}.
\end{align}
Under the assumption that $\epsilon_{\rm int} \le 5/3$, the solution is of the form
\begin{align}
\begin{aligned}
    n &= \frac{\ln\left(mk(5/3)^k\max_i \|H_i\|(1+\Gamma)\epsilon_\mathrm{int}^{-1} \right)}{2{\rm LambertW}\left(\ln\left( mk(5/3)^k\max_i \|H_i\|(1+\Gamma)\epsilon_\mathrm{int}^{-1}\right)/2 \right)} \\
    &\in O\left( \log\left(mk(5/3)^k\max_i \|H_i\|(1+\Gamma)\epsilon_\mathrm{int}^{-1} \right)\right).
\end{aligned}
\end{align}
Finally, we have from Theorem~\ref{thm:extrapBd} that the error in the interpolated eigenvalue is within $\epsilon$ provided that the error in the data satisfies
\begin{equation}
    \left(\frac{2}{\pi}\log(n+1) + 1\right) \max_i|\lambda_k(s_i) - \tilde{\lambda}_k(s_i)| \le \left(\frac{2}{\pi}\log(n+1) + 1\right)\epsilon_\mathrm{PE}\le \epsilon. 
\end{equation}
Taking $\epsilon_\mathrm{PE}$ to satisfy this, the number of exponentials scales via equation~\eqref{eq:Nexp_inter} as
\begin{align}
    N_{\exp}\in \widetilde{O}\left(\frac{m^2 k (25/3)^k \max\|H_i\| (1+\Gamma)}{\epsilon} \right),
\end{align}
where we've also implemented~\eqref{eq:match_n}. By the union bound, the total probability of failure is at most $1/3$.

\end{proof}

\section{Proof of \Cref{thm:spectral_gaussian_error} \label{app:gaussian_sample_error}}
Per the statement of the theorem, we would like to estimate the error between the Discrete Fourier Transform (DFT) and corresponding samples of the Fourier transform of the continuous Gaussian distribution. We split this analysis in three parts: calculating time-domain truncation errors, estimating frequency-domain truncation errors (non-zero sampling rate), and finally calculate the error coming from renormalization.~\Cref{fig:error_schematic} gives an schematic representation of the analysis, including the resulting big-$O$ bounds.

\begin{figure}[H]
\centering
\resizebox{\columnwidth}{!}{
    \begin{tikzpicture}[->,>=stealth',shorten >=1pt,auto, node distance=6cm,thick, main node/.style={circle,draw,font=\Large\bfseries, minimum size=7em}]
    
      \node[main node] (1) {$\frac{X_{1/T}\left(\frac{k}{2^qT}\right)}{\sqrt{2^{q}\N(\sigma,T,m)}}$};
      \node[main node] (4)  [right=4.5cm of 1] {$\frac{X\left(\frac{k}{2^qT}\right)}{T\sqrt{2^q\N(\sigma,T,m)}}$};  
      \node[main node] (2) [below right=1.5cm and 0.3cm of 4] {$\frac{X\left(\frac{k}{2^qT}\right)}{\sqrt{\N(\sigma_f,F,q)}}$};
      \node[main node] (3) [below left=1.5cm and 0.3cm of 1] {$X[k]$};

      \path
        (1) edge node[above,font=\large] {$\epsilon_{\mathrm{alias}}\in O\left(\frac{2^{m/4}}{2^{q/2}e^{\Omega(2^m)}}\right)$} (4)
        (3) edge node[above left,font=\large] {$\epsilon_{\mathrm{trunc}}\in O\left(\frac{2^{m/4}}{2^{q/2}e^{\Omega(2^m)}}\right)$} (1)
        (4) edge node[above right,font=\large] {$\epsilon_{\rm renorm}\in O\left(\frac{2^{(m-q)/2}}{\exp\left(\Omega(2^m)\right)}\right)$} (2)
        (3) edge node[below,font=\large] {$\epsilon_{\rm total}\leq \epsilon_{\rm trunc} + \epsilon_{\rm alias} + \epsilon_{\rm renorm}\in O\left(\frac{\sqrt{2^{m}}}{2^{q/2} e^{\Omega(2^m)}}\right)$} (2);      
    \end{tikzpicture}
}
\caption{Schematic representation of the error estimation between the DFT of $x[n]$, $X[k]$, and the re-normalized samples of the Fourier transform $X\left(\frac{k}{2^q T}\right)$. At the edges of the diagram, we have noted the additive error between the expressions at the nodes.}
\label{fig:error_schematic}
\end{figure}
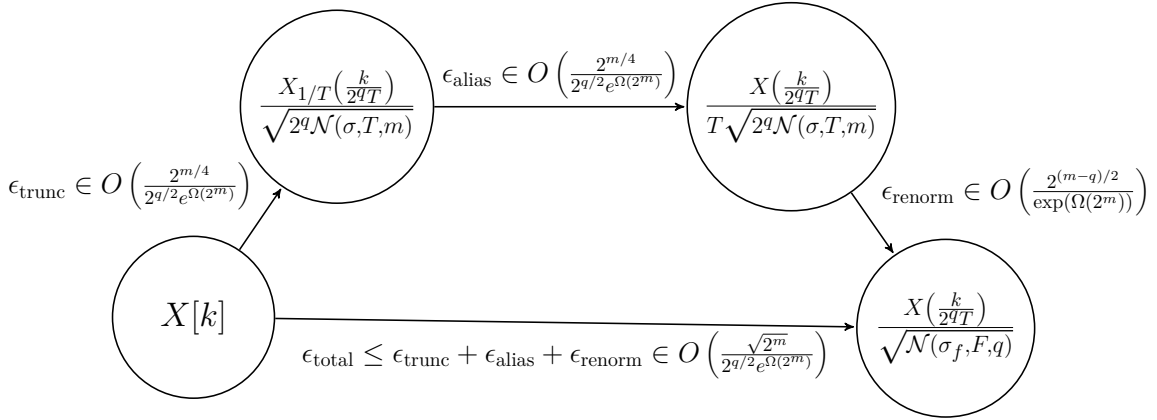

For reference, here are some useful definitions which we will use in the proof to follow.
\begin{align}
    p(w;\sigma) &:= \frac{1}{\sigma \sqrt{2\pi}}e^{-w^2/(2\sigma^2)}\\
    x(w) &:= \sqrt{p(w;\sigma)}\\
    X(f) &:= \mathcal{F}\{x\}(f) = \int_{\mathbb{R}} x(t)  e^{-2\pi i t f} dt \\
    X_{1/T}\left(\frac{k}{2^qT}\right) &:= \sum_{n=-\infty}^\infty x(nT)\cdot e^{-i 2\pi k n/2^q} \quad \quad k = -2^{q-1}, \dots, 2^{q-1}-1 \\
    \N(\sigma,T,m) &:= \sum^{2^{m-1}-1}_{n=-(2^{m-1}-1)} \lvert x(nT) \rvert^2 \\
    x[n] &:= \begin{cases}
              \frac{1}{\N(\sigma,T,m)^{1/2}}x(n T) & n = -2^{m-1}-1, \dots, 2^{m-1}-1 \\
              0  & n = -2^{m-1}
            \end{cases} \label{eq:normalized_xn}\\
    X[k] &:= \frac1{2^{q/2}}\sum^{2^{m-1}-1}_{n=-(2^{m-1}-1)} x[n]\cdot e^{-i 2\pi kn/2^q} \label{eq:DFT}
\end{align}

It will also be convenient to define the Fourier dual sampling rates and widths, respectively.
\begin{align}
    F := (2^q T)^{-1} \label{eq:F_def}\\
    \sigma_f := (4\pi\sigma)^{-1} \label{eq:sigmaf_def}
\end{align}

\subsection{Truncation Error}

The error between the scaled Discrete-time Fourier Transform (DTFT)
\begin{align}
    \frac{X_{1/T}\left(\frac{k}{2^qT}\right)}{\sqrt{2^{q}\N(\sigma,T,m)}} \label{eq:truncb}
\end{align}
and the untruncated DFT $X[k]$ of \Cref{eq:DFT} can be estimated to be

\begin{align} 
    \epsilon_\mathrm{trunc} &=\frac{1}{2^{q/2}\N(\sigma,T,m)^{1/2}}\left(\sum^{\infty}_{n=2^{m-1}} \sqrt{p(nT;\sigma)} +\sum_{n=-\infty}^{-2^{m-1}}\sqrt{p(nT;\sigma)} \right)\cr
    &= \frac{2\tilde{\epsilon}_\mathrm{trunc}}{2^{q/2}\N(\sigma,T,m)^{1/2}}, \label{eq:first_eps_trunc}
\end{align}
where
\begin{align}
    \tilde{\epsilon}_\mathrm{trunc} := \sum_{n = 2^{m-1}}^\infty \sqrt{p(nT;\sigma)}.
\end{align}
We can upper bound $\tilde{\epsilon}_\mathrm{trunc}$ using a right-Riemann sum approximation of the error function,
\begin{align*}
    \sum^{\infty}_{n=2^{m}/2} \sqrt{p(nT;\sigma)} &\leq \sqrt[4]{\frac{\pi }{2}} \frac{\sqrt{\sigma }}{T} \erfc\left(\frac{T(2^{m}/2-1)}{2 \sigma }\right)
\end{align*}
and together with the simple bound~\cite{chiani2003new}
\begin{align}
    \erfc{x} \leq e^{-x^2} \qquad (x > 0)
\end{align}
we have
\begin{align} \label{eq:tilde_trunc_final}
    \tilde{\epsilon}_\mathrm{trunc}\leq \sqrt[4]{\frac{\pi }{2}} \frac{\sqrt{\sigma }}{T} \exp\left\{-\left(\frac{T(2^{m-1}-1)}{2 \sigma }\right)^2\right\}.
\end{align}
Finally, we can express the truncation error $\epsilon_\mathrm{trunc}$, using~\eqref{eq:first_eps_trunc} and the fact that $\N(\sigma,T,m) \in O(1/T)$.
\begin{align}
    \epsilon_{\mathrm{trunc}}\le \frac{2^{3/4} \pi^{1/4} \frac{\sqrt{\sigma }}{T} e^{-\left(\frac{T(2^{m-1}-1)}{2 \sigma }\right)^2}}{2^{q/2}\N(\sigma,T,m)^{1/2}} \in O\left( \sqrt{\frac{\sigma}{2^q T}} e^{-\Omega(2^m/(\sigma/T))} \right)
\end{align}
along

\subsection{Aliasing error}

Now that we have estimated the error on the DTFT by time-domain truncation, we would also like to estimate the aliasing error, that is, the difference between the DTFT and the Fourier transform over a period $1/T$. This is coming from having a finite sampling rate of a function that is not bounded in the frequency domain. In order to estimate the aliasing error, we must look at the Fourier transform of a Gaussian distribution, which is another Gaussian.
\begin{align}
X(f) 
%&=\left(\frac{1}{\sigma\sqrt{2\pi}}\right)^{1/2}2\sigma\sqrt{\pi}\cdot e^{-(2 \sigma \pi f)^2} \cr
%&= \pi^{1/4}2^{3/4}\sqrt{\sigma}\cdot e^{-(2 \sigma \pi f)^2} \cr
%&= \left( 4 \sigma \pi /\sqrt{2\pi} \right)^{1/2}\cdot e^{-\frac{f^2}{4/(4 \sigma \pi)^2}} \cr
    &=\sqrt{p(f;\sigma_f)}
\end{align}
Now, the DTFT can be expressed in terms of a Fourier transform in the following way
\begin{align}
X_{1/T}(f)= \frac{1}{T} \sum^{\infty}_{k=-\infty} X ( f - k/T ).
\end{align}
Therefore, we define the aliasing error as
\begin{align}
\tilde{\epsilon}_\mathrm{alias} := X_{1/T}(f) - \frac{1}{T}X(f)= \frac{1}{T}\sum^{-1}_{k=-\infty} X ( f - k/T )+\frac{1}{T}\sum^{\infty}_{k=1} X ( f - k/T ).
\end{align}
The aliasing error has a critical point at $f=0$ and its second derivative is strictly positive throughout $f=[-1/2T,1/2T]$ and thus we know that the error is largest at the boundaries of the DTFT. That is, at $f=-1/2T,1/2T$. The error at these two points is expected to be the same so we just choose $f=1/2T$ in order to bound the error from above. That means

\begin{align}
\tilde{\epsilon}_{\mathrm{alias}}&\leq\frac{1}{T}\sum^{-1}_{k=-\infty} X ( 1/2T - k/T )+\frac{1}{T}\sum^{\infty}_{k=1} X (1/2T - k/T ) \cr
&\leq \frac{1}{T}\sum^{\infty}_{k=-\infty} X ( 1/2T - k/T ) \cr
&= \pi^{1/4}2^{3/4}\sqrt{\sigma}\frac{1}{T}  \theta_{2}(e^{-(2 \sigma \pi/T)^2}) \cr
&= \pi^{1/4}2^{3/4}\sqrt{\sigma}\frac{2}{T}  e^{-( \pi \sigma/ T)^2} \sum_{n=0}^\infty (e^{-4(2\pi \sigma/ T)^2})^{n(n+1)}, 
\end{align}
where $\theta_2(q) \equiv \theta_2(0,q)$ is the 2nd Jacobi theta function~\cite{borwein1987pi}. For $e^{-( 2\pi \sigma/ T)^2}\leq 1$,
\begin{align}
    \tilde{\epsilon}_{\mathrm{alias}}&\leq \pi^{1/4}2^{3/4}\sqrt{\sigma}\frac{1}{T}\cdot \left( 2  e^{-\pi^2( \sigma/ T)^2} \sum_{n=0}^{\infty} (e^{-4\pi^2( \sigma/ T)^2})^{n} \right ) \cr
    &= \pi^{1/4}2^{3/4}\sqrt{\sigma}\frac{1}{T}\cdot \left( 2  e^{-\pi^2( \sigma/ T)^2} \frac{1}{1-e^{-4\pi^2( \sigma/ T)^2}} \right )
\end{align}
For $e^{-4\pi^2( \sigma/ T)^2} \leq 1/2$
\begin{align}
    \tilde{\epsilon}_{\mathrm{alias}}&\leq \pi^{1/4}2^{3/4}\sqrt{\sigma}\frac{1}{T}\cdot \left( 2  e^{-\pi^2( \sigma/ T)^2} \frac{1}{1-e^{-4\pi^2( \sigma/ T)^2}} \right ) \cr
    & \leq 4\pi^{1/4}2^{3/4}\sqrt{\sigma}\frac{1}{T}\cdot e^{-\pi^2( \sigma/ T)^2}\label{eq:aliasbd}
\end{align}
Therefore, the error $\epsilon_\mathrm{alias}$ between 
\begin{align}
\frac{X_{1/T}\left(k /(2^q T)\right)}{2^{q/2}\N(\sigma,T,m)^{1/2}}
\end{align}
and 
\begin{align}\label{eq:unormalized_ft}
&\frac{X\left(k/(2^q T)\right)}{T2^{q/2}\N(\sigma,T,m)^{1/2}}
%&=\frac{F^{1/2}}{T^{1/2}\N(\sigma,T,m)^{1/2}}X\left(\frac{k}{2^qT}\right)
\end{align}
can be estimated using~\eqref{eq:aliasbd} to be
\begin{align}
\epsilon_{\mathrm{alias}}=\frac{\tilde{\epsilon}_{\mathrm{alias}}}{2^{q/2}\N(\sigma,T,m)^{1/2}} \le \left( \frac{4\pi^{1/4}2^{3/4}\sqrt{\sigma}\frac{1}{T}\cdot e^{-\pi^2( \sigma/ T)^2}}{2^{q/2}\N(\sigma,T,m)^{1/2}}\right).
\end{align}
Thus, with a choice of 
\begin{align}
\frac{\sigma}{T}\sim \sqrt{2^m},
\end{align}
%
% \textbf{Alternate Bound}
% \begin{align}
%     \tilde{\epsilon}_{\mathrm{alias}}&=\tilde{\epsilon}_{\mathrm{trunc}} \cr
%     4\pi^{1/4}2^{3/4}\sqrt{\sigma}\frac{1}{T}\cdot e^{-\pi^2( \sigma/ T)^2} &= 2^{3/4} \pi^{1/4} \frac{\sqrt{\sigma }}{T} e^{-\left(\frac{T(2^{m-1}-1)}{2 \sigma }\right)^2} \cr
%     -\pi^2 (\sigma/T)^2 + \ln{4} &= -(2^{m-1}-1)^2/(\sigma/T)^2/4
% \end{align}
% %
% Solving for $\sigma/T$ we find that the only solution that makes sense as $2^m \to \infty$ is
% %
% \begin{align}
%     \sigma/T=\frac{\sqrt{\sqrt{\pi ^2 \left(2^m-2\right)^2+4 \ln ^2(4)}+\ln (16)}}{2 \pi }
% \end{align}
% \textbf{End of Alternate Bound}\\
%
both sources of error can be bounded with
\begin{align}
\epsilon_{\mathrm{trunc}}=O\left(\frac{2^{m/4}}{2^{q/2}e^{\Omega(2^m)}}\right)\cr
\epsilon_{\mathrm{alias}}=O\left(\frac{2^{m/4}}{2^{q/2}e^{\Omega(2^m)}}\right).
\end{align}
This means that after the DFT, we know that we will have in the register the amplitude
\begin{align}
\sqrt{\frac{F}{T\N(\sigma,T,m)}}X\left(kF\right)+\epsilon_\mathrm{DFT},
\end{align}
where $\epsilon_\mathrm{DFT}$ is the total error from both truncation as well as aliasing.  We then bound the error by summing~\eqref{eq:aliasbd} and~\eqref{eq:truncb} to find
\begin{align}\label{eq:DFTbd}
    \epsilon_\mathrm{DFT} &\leq \frac{(8\pi )^{1/4}\sqrt{\sigma}}{T\sqrt{2^q \mathcal{N}(\sigma,T,m)}} e^{-\left(\frac{T(2^{m-1}-1)}{2 \sigma }\right)^2}   + 4 \frac{(8\pi )^{1/4}\sqrt{\sigma}}{T\sqrt{2^q \mathcal{N}(\sigma,T,m)}} e^{-\pi^2( \sigma/ T)^2} \nonumber\\ 
    &\le  2 \max\left(\epsilon_{\mathrm{alias}}, \epsilon_{\mathrm{trunc}}\right) 
    \in O\left(\frac{2^{m/4}}{2^{q/2}e^{\Omega(2^m)}}\right).
\end{align}

\subsection{Renormalization error and ultimate additive error}

We know that the signal after DFT has to be normalized, we also know that it is locally $\epsilon_\mathrm{DFT}$ close to ~\Cref{eq:unormalized_ft}, which is not necessarily normalized. Therefore, for the norm of samples of ~\Cref{eq:unormalized_ft} we obtain the following upper bound on normalization ratios

\begin{align}\label{eq:renorm_error}
\frac{F}{T\N(\sigma,T,m)}\sum_k \left|X\left(kF\right)\right|^2 &= \frac{F\N(\sigma_f,F,q)}{T\N(\sigma,T,m)}\cr
&=  1 - 2\left(\frac{F}{T \N(\sigma,T,m)}\right)^{1/2}\Re\left(\vec{\epsilon}_{DFT}\cdot \vec{X}\right) +  \| \vec{\epsilon}_\mathrm{DFT} \|^2 \cr 
& \leq 1 +  \| \vec{\epsilon}_\mathrm{DFT} \|^2 \cr 
& \leq 1 +  \| \vec{\epsilon}_\mathrm{DFT} \| \cr 
& \leq 1  + \epsilon_\mathrm{DFT} 2^{q/2},
\end{align}
where we have used the assumption $\|\vec{\epsilon}_{DFT}\|\leq1$. Similarly, the lower bound is
\begin{align}
\begin{aligned}
    \frac{F}{T\N(\sigma,T,m)}\sum_k \left|X\left(\frac{k}{2^qT}\right)\right|^2 &=\frac{F\N(\sigma_f,F,q)}{T\N(\sigma,T,m)}\\
    &\geq  1 - 2\left(\frac{F}{T \N(\sigma,T,m)}\right)^{1/2}\Re\left(\|\vec{\epsilon}_\mathrm{DFT}\| \|\vec{X}\|\right) \\
    &\geq 1 - 2 \epsilon_\mathrm{DFT} 2^{q/2} \left(\frac{F\N(\sigma_f,F,q)}{T \N(\sigma,T,m)}\right)^{1/2} \\
    &\geq 1 - 2 \epsilon_\mathrm{DFT} 2^{q/2} \left(1 + \|\vec{\epsilon}_\mathrm{DFT}\|^2 \right) \\
    &\geq 1 - 4 \epsilon_\mathrm{DFT} 2^{q/2}.
\end{aligned}
\end{align}
Thus, if we take the full inequality
\begin{align}
    &1 - 4 \epsilon_\mathrm{DFT} 2^{q/2} \leq \left(\frac{F\N(\sigma_f,F,q)}{T\N(\sigma,T,m)}\right)\leq 1 + \epsilon_\mathrm{DFT}2^{q/2}
\end{align}
and take the square root, we obtain,
\begin{align}
    &1 - 2\epsilon_\mathrm{DFT} 2^{q/2} \leq \left(\frac{F\N(\sigma_f,F,q)}{T\N(\sigma,T,m)}\right)^{1/2} \leq 1 + \epsilon_\mathrm{DFT}2^{q/2}, \cr
\end{align} 
where we have used $1+\sqrt{x}\leq 1 + x$ and $1-\sqrt{x}\geq 1-x/2$ for $x\leq 1$. We now divide every side by $\N(\sigma_f,F,q)$ to obtain
\begin{align}
    &1/\N(\sigma_f,F,q)^{1/2} - 2\epsilon_\mathrm{DFT} 2^{q/2} / \N(\sigma_f,F,q) \leq \left(\frac{F}{T\N(\sigma,T,m)}\right)^{1/2} \leq\cr &1/\N(\sigma_f,F,q)^{1/2} + \epsilon_\mathrm{DFT}\sqrt{2^{q}}/\N(\sigma_f,F,q)^{1/2}. \cr
\end{align}
Subtract $1/\N(\sigma_f,F,q)$, which gives
\begin{align}
\begin{aligned}
    &- 2\epsilon_\mathrm{DFT} \sqrt{\frac{2^q }{ \N(\sigma_f,F,q) }} \leq \sqrt{\frac{F}{T\N(\sigma,T,m)}} -\sqrt{\frac{1}{\N(\sigma_f,F,q)}} \cr 
    &\leq  \epsilon_\mathrm{DFT}\sqrt{\frac{2^{q}}{\N(\sigma_f,F,q)}}.
\end{aligned}
\end{align}
We multiply these inequalities by $\max |X(k)|$
\begin{align}
    &- 2\epsilon_\mathrm{DFT} \sqrt{\frac{2^q }{ \N(\sigma_f,F,q) }}\left(\frac{1}{2\pi\sigma_f^{2} } \right)^{1/4} \cr 
    &\leq \left(\sqrt{\frac{F}{T\N(\sigma,T,m)}} -\sqrt{\frac{1}{\N(\sigma_f,F,q)}}\right) \max\left| X(k)\right| \cr
    &\leq  \epsilon_\mathrm{DFT}\sqrt{\frac{2^q }{ \N(\sigma_f,F,q) }} \left(\frac{1}{2\pi\sigma_f^{2} } \right)^{1/4}.
\end{align}
We take the absolute value to have a two-sided inequality, and use the definitions~\eqref{eq:F_def} and~\eqref{eq:sigmaf_def}.
\begin{align}
    &\left|\sqrt{\frac{F}{T\N(\sigma,T,m)}} -\sqrt{\frac{1}{\N(\sigma_f,F,q)}}\right| \max\left| X(k)\right|  \cr 
    &\leq 2\epsilon_\mathrm{DFT} \sqrt{\frac{1}{ T F\N(\sigma_f,F,q) }} \left(\frac{ (4\pi\sigma)^2 }{2\pi } \right)^{1/4} \cr
    &\leq 2 \frac{(8\pi )^{1/4}\sqrt{\sigma}}{T\sqrt{2^q \mathcal{N}(\sigma,T,m)}} \sqrt{\frac{1}{ T F\N(\sigma_f,F,q) }} \left(\frac{ (4\pi\sigma)^2 }{2\pi } \right)^{1/4} \cr 
    &\qquad\qquad\max\left( 4 \exp\left\{-\pi^2 (\sigma/T)^2\right\}, \exp\left\{-\left(\frac{2^{m-1}-1}{2 \sigma/T}\right)^2\right\} \right) \cr
    &\leq 4 \frac{\sigma}{T} \sqrt{\frac{2\pi}{ 2^q T  \mathcal{N}(\sigma,T,m) F\N(\sigma_f,F,q) }}  \cr 
    &\qquad\qquad\max\left( 4 \exp\left\{-\pi^2 (\sigma/T)^2\right\}, \exp\left\{-\left(\frac{2^{m-1}-1}{2 \sigma/T}\right)^2\right\} \right)  
\end{align}

It is worth noting that, for the purpose of estimating $\left(T\N(\sigma,T,m)\right)^{1/2}$, we can increase $2^q$ arbitrarily (or decrease $F$ arbitrarily). Through this, we see that:
\begin{align}
\lim_{F\to 0}F\N(\sigma_f,F,q)&=\lim_{F\to 0}\erf\left(\frac{ (2^{q-1}-1/2) F}{\sqrt{2}\sigma_f}\right) = \erf\left(\sqrt{2} \pi \sigma /T\right)\cr
&\geq 1-\frac{1}{\exp{\left(\sqrt{2}\pi\sigma/T\right)^2}}
\end{align}
For our choice of $\frac{\sigma}{T}$ this would mean that
\begin{align}
\left(T\N(\sigma,T,m)\right)^{1/2} = 1 - O\left(\frac{2^{m/4}}{\exp\left(\Omega(2^m)\right)}\right),
\end{align}
and 
\begin{align}
T^{1/2}\N(\sigma,T,m)^{1/2} &= F^{1/2}\N(\sigma_f,F,q)^{1/2}  \cr
&- T^{1/2}\N(\sigma,T,m)^{1/2}O\left(\frac{2^{m/4}}{\exp\left(\Omega(2^m)\right)}\right) \cr
&= F^{1/2}\N(\sigma_f,F,q)^{1/2}  - O\left(\frac{2^{m/4}}{\exp\left(\Omega(2^m)\right)}\right). 
\end{align}
Therefore, the amplitudes that we get on the ancillary register are
\begin{align}
\frac{1}{\sqrt{\N(\sigma_f,F,q)}}X\left(\frac{k}{2^q T}\right)
+ O\left(\frac{2^{(m-q)/2}}{\exp\left(\Omega(2^m)\right)}\right)+\epsilon_\mathrm{DFT}.
\end{align}

\section{Proofs of State Preparation and Ground State Energy Extrapolation Costs}
\label{app:Bernstein_proof}
Here we prove results related to the costs of state preparation and phase estimation for extrapolation of the ground state energy, from \Cref{thm:gauss_state_prep} and \Cref{thm:gauss_phase_est} respectively.

\subsection{Proof of \Cref{lem:radius_of_analiticity}\label{app:radius_of_analiticity}}
\begin{proof}

Let $\lambda_{m-1}$, $\lambda_{m}$, and $\lambda_{m+1}$ be consecutive simple eigenvalues of $H$ 
\begin{align}
    \lambda_{m-1} < \lambda_{m} < \lambda_{m+1}
\end{align} 
such that $\left|\lambda_m-\lambda_{m-1}\right| \geq \gamma_0$, and $\left|\lambda_m-\lambda_{m+1}\right| \geq \gamma_0$. With $\lambda(\tau)$ as in the lemma statement, we have $\lambda_m(0) = \lambda_m$ for all indices $m$. By \Cref{lem:BauerFike}, there exist eigenvalues $\lambda_j(\tau)$, $\lambda_k(\tau)$, and $\lambda_l(\tau)$ of $\tilde{H}_z$ such that
\begin{align}
    \max_{\tau \in \{\tau:|\tau|\leq r\}}|\lambda_{m-1}(0)-\lambda_j(\tau) | \leq \max_{\tau \in \{\tau:|\tau|\leq r\}}\| H- \tilde{H}_z\| \cr
    \max_{\tau \in \{\tau:|\tau|\leq r\}}|\lambda_m(0)- \lambda_k(\tau)| \leq \max_{\tau \in \{\tau:|\tau|\leq r\}}\| H- \tilde{H}_z\| \cr
    \max_{\tau \in \{\tau:|\tau|\leq r\}}|\lambda_{m+1}(0)- \lambda_l(\tau)| \leq \max_{\tau \in \{\tau:|\tau|\leq r\}}\| H- \tilde{H}_z\|  
\end{align}
for some $j$, $k$, and $l$. We are guaranteed that $j\neq m$, $l\neq m$ and $k=m$ if
\begin{align}
 \max_{\tau \in \{\tau:|\tau|\leq r\}}\| H- \tilde{H}_z\| \leq \gamma_0/2.
\end{align}
Assuming we have the upper bound
\begin{align}
 \| H - \tilde{H}_z \| \leq \alpha \frac{|\tau|^p}{(p+1)!} e^{\beta |\tau|}, 
\end{align}
this guarantee can be satisfied by the constraint
\begin{align}
    \gamma_0 /2 \geq \max_{\tau\in \{\tau:|\tau|\leq r\}}\left(\alpha \frac{|\tau|^p}{(p+1)!}e^{\beta |\tau|}\right) = \alpha \frac{r^p}{(p+1)!}e^{\beta r}.
\end{align}
Solving for $r$ on the last inequality gives us
\begin{align}
    r \leq \frac{p}{\beta}\mathrm{LambertW}\left(\frac{\beta}{p} \left(\frac{\gamma_0 (p+1)!}{2 \alpha }\right)^{1/p}\right).
\end{align}
If $\lambda_{m}$ is the ground state or highest state of $H$, we just repeat a similar proof except there is only one state above or below respectively.
\end{proof}

\subsection{Proof of \Cref{lem:Magnus_error}\label{app:Magnus_error}}

\begin{proof}

We take a look at the complex-time Schr\"odinger equation for $S_{p}(\tau)$.
\begin{align}
    S_p'(\tau) = \overbrace{i(H+\mathscr{E}(\tau))}^{A}S_p(\tau)
\end{align}
We now recall that we can express
\begin{align}
    \Omega(\tau) = \log S_p(\tau)
\end{align}
through the Magnus expansion for the complex-time derivative,
\begin{align}
    \Omega'(\tau) = \sum_{n=0}^\infty \frac{B_n}{n!} \mathrm{ad}^n_{\Omega}(A).
\end{align}
We then re-express $E(\tau)$ and expand $\Omega'$
\begin{align}
   E (\tau) &= \int_{\mathcal{P}}\Omega'(\tau_1)\mathrm{d} \tau_1  - \int_{\mathcal{P}} H \mathrm{d} \tau_1 \cr 
   &=\int_{\mathcal{P}}\left(\Omega'(\tau_1) - H \right)\mathrm{d} \tau_1 \cr 
   &=\int_{\mathcal{P}}\left(\mathscr{E}(\tau_1) + \sum_{n=1}^\infty \frac{B_n}{n!} \mathrm{ad}^n_{\Omega(\tau_1)}(A(\tau_1)) \right)\mathrm{d} \tau_1. 
\end{align}
Here, $\mathcal{P}$ is any path in the complex plane going from $\tau_1=0$ to $\tau_1=\tau$.  We now expand out the first application of ${\rm ad}_{\Omega}$ on $A$ and separate the $n=1$ term from the sum obtaining:
\begin{align}
   E (\tau) &=\int_{\mathcal{P}}\left(\mathscr{E}(\tau_1) + B_1 [\Omega(\tau_1),A(\tau_1)]  + \sum_{n=2}^\infty \frac{B_n}{n!} \mathrm{ad}^{n-1}_{\Omega}[\Omega(\tau_1),A(\tau_1)] \right)\mathrm{d} \tau_1.
\end{align}
Following this, we make use of the estimation lemma (`ML` inequality) to obtain
\begin{align}
    \| E(\tau) \|  &\leq \int^{S}_0 \left|\gamma_{\mathcal{P}}\right| \left(\|\mathscr{E}\| + \left|B_1\right| \left\| [\Omega,A]\right\|  + \sum_{n=2}^\infty \frac{|B_n|}{n!} \left\| \mathrm{ad}^{n-1}_{\Omega}\right\| \left\|[\Omega,A] \right\|\right)\mathrm{d} s \cr 
\end{align}
where $s$ parametrizes the path $\mathcal{P}$ and $\left| \gamma_{\mathcal{P}} \right|ds=\left| \frac{\mathrm{d}\tau_1}{\mathrm{d}s}\right|ds$ is the arc-length differential. Now, we make use of the bound $\|{\rm ad}_{\Omega}\| \leq 2 \|\Omega\|$ to obtain:
\begin{align}
    &\leq \int^{S}_0 \left|\gamma_{\mathcal{P}}\right| \left(\|\mathscr{E}\| + \left|B_1\right| \left\| [\Omega,A]\right\|  + \sum_{n=2}^\infty \frac{|B_n|}{n!} \left(2\left\|{\Omega}\right\|\right)^{n-1} \left\|[\Omega,A] \right\|\right)\mathrm{d} s \cr 
    &\leq \int^{S}_0 \left|\gamma_{\mathcal{P}}\right| \left(\|\mathscr{E}\| + \left|B_1\right| \left\| [\Omega,A]\right\|  + \frac{\left\|[\Omega,A] \right\|}{\left(2\left\| {\Omega}\right\|\right)}\sum_{n=2}^\infty \frac{|B_n|}{n!} \left(2\left\| {\Omega}\right\|\right)^{n} \right)\mathrm{d} s
\end{align}
We now make use of

\begin{align}  
 \sum_{n=2}^\infty \left|B_n\right| \frac{|t|^n}{n!} = \sum_{n=1}^\infty B^{*}_n \frac{|t|^{2n}}{(2n)!} =1 - \frac{|t|}{2}\cot{\frac{|t|}{2}} \qquad\text{for } |t| \leq 2 \pi
\end{align}
to obtain:
\begin{align}
    \left\|E(\tau)\right\|  &\leq \int^{S}_0 \left|\gamma_{\mathcal{P}}\right| \left(\|\mathscr{E}\| + \frac{1}{2}\left\| [\Omega,A]\right\|  + \frac{\left\|[\Omega,A] \right\|}{\left(2\left\| {\Omega}\right\|\right)}\left(1-\|\Omega\|\cot{\|\Omega\|} \right)  \right)\mathrm{d} s \cr 
    &\leq \int^{S}_0 \left|\gamma_{\mathcal{P}}\right| \left(\|\mathscr{E}\| + \frac{1}{2}\left\| [\Omega,A]\right\|  + \frac{1}{2}\left\|[\Omega,A] \right\|  \left(\frac{1}{\left\| {\Omega}\right\|}-\cot{\|\Omega\|} \right)  \right)\mathrm{d} s \cr 
\end{align}
We now assume that
\begin{align}
    \left|\gamma_{\mathcal{P}}\right|=\left|\frac{\mathrm{d}\tau}{\mathrm{d}s}\right| = 1,
\end{align}
and
\begin{align}
    s = |\tau_1|,
\end{align}
thus
\begin{align}
    \left\|E(\tau)\right\| &\leq \int^{|\tau|}_0  \left(\|\mathscr{E}\| + \frac{1}{2}\left\| [\Omega,A]\right\|  + \frac{1}{2}\left\|[\Omega,A] \right\|  \left(\frac{1}{\left\| {\Omega}\right\|}-\cot{\|\Omega\|} \right)  \right)\mathrm{d} s.
\end{align}
We assume 
\begin{align}
    \left(\frac{1}{\left\| {\Omega}\right\|}-\cot{\|\Omega\|} \right) \leq 1,
\end{align}
which is fulfilled provided 
\begin{align}
\| E(\tau) \|\leq 1.
\end{align}
Now,
\begin{align}\label{eq:E_calE_commOA}
    \left\|E(\tau)\right\| &\leq \int^{|\tau|}_0  \left(\|\mathscr{E}\| +\left\| [\Omega,A]\right\|   \right)\mathrm{d} s. 
\end{align}
Recalling that $\Omega(\tau) = E(\tau) + \tau H$,
\begin{align}
    \left\| [\Omega,A]\right\| = \left\| [E + \tau H,\mathscr{E}+H]\right\| \leq 2\|E\|\left(\|\mathscr{E}\|+\|H\|\right) + 2 |\tau| \|H\|\|\mathscr{E}\|.
\end{align}
Plugging back into \Cref{eq:E_calE_commOA} and using estimation lemma for $\|E\|$ inside the integral we obtain:
\begin{align}
    \left\|E(\tau)\right\| &\leq \int^{|\tau|}_0 \left(\|\mathscr{E}\| + 2 s \|H\|\|\mathscr{E}\|   \right)\mathrm{d} s + 2\max_{|\tilde\tau|\leq |\tau|}\|E(\tilde\tau)\|\int^{|\tau|}_0 \left(\|\mathscr{E}\|+\|H\|\right)\mathrm{d} s.
\end{align}
We now do the relabelling $\tau \to \tau_2$ and then take $\max_{|\tau_2| \leq |\tau|}$ of both sides
\begin{align}
    \max_{|\tau_2| \leq |\tau|}\left\|E(\tau_2)\right\| &\leq \max_{|\tau_2| \leq |\tau|}\int^{|\tau_2|}_0 \left(\|\mathscr{E}\| + 2 s \|H\|\|\mathscr{E}\|   \right)\mathrm{d} s + 2\max_{|\tau_2|\leq |\tau|}\|E(\tau_2)\|\max_{|\tau_2| \leq |\tau|}\int^{|\tau_2|}_0 \left(\|\mathscr{E}\|+\|H\|\right)\mathrm{d} s,
\end{align}
which becomes
\begin{align}
    \max_{|\tau_2| \leq |\tau|}\left\|E(\tau_2)\right\| &\leq \int^{|\tau|}_0 \left(\|\mathscr{E}\| + 2 s \|H\|\|\mathscr{E}\|   \right)\mathrm{d} s + 2\max_{|\tau_2|\leq |\tau|}\|E(\tau_2)\|\int^{|\tau|}_0 \left(\|\mathscr{E}\|+\|H\|\right)\mathrm{d} s,
\end{align}
Solving for $\max_{|\tau_2| \leq |\tau|}\left\|E(\tau_2)\right\|$
\begin{align}
   \max_{|\tau_2|\leq |\tau|}\|E(\tau_2)\| &\leq \frac{\int^{|\tau|}_0  \left(\|\mathscr{E}\| + 2 s \|H\|\|\mathscr{E}\|   \right)\mathrm{d} s }{1- 2|\tau|\|H\| - 2\int^{|\tau|}_0 \|\mathscr{E}\|\mathrm{d} s}  \cr 
    &\leq \frac{\left(1+2|\tau|\|H\|\right)\int^{|\tau|}_0  \left(\|\mathscr{E}\| \right)\mathrm{d} s }{1- 2|\tau|\|H\| - 2\int^{|\tau|}_0 \|\mathscr{E}\|\mathrm{d} s}  \cr 
\end{align}
We now assume
\begin{align}
    \int^{|\tau|}_0 \|\mathscr{E}\|\mathrm{d} s \leq 1/8
\end{align}
and
\begin{align}
    |\tau| \leq 1/8,
\end{align}
thus,
\begin{align}
    \left\|E(\tau)\right\|\leq \max_{|\tau_2|\leq |\tau|}\left\|E(\tau_2)\right\| &\leq \frac{5}{2}\int^{|\tau|}_0  \|\mathscr{E}(\tau_1)\| \mathrm{d} s.
\end{align}
\end{proof}

\subsection{Proof of \Cref{lem:H_error_bound}\label{app:H_error_bound}}

Our task here is to show that a bound of the form shown~\eqref{eq:Hdiff_bd} actually exists, and for this we follow closely the work done in~\cite{childs2021theory}. In order to make assertions over the general complex plane we will extend their results to complex time. The extension will be straightforward because
$$
e^{i\tau A}
$$
has no singularities except for non-removable singularities at $\operatorname{Im}(\tau)\in \{-\infty,\infty\}$ for a Hermitian $A$ and at $\operatorname{Im}(\tau) = -\infty$ when it is positive semi-definite. Thus, we can replace all the real-axis integrals
$$
\int^{t}_0 f(t_1)\mathrm{d}t_1
$$
for the contour integrals
$$
\int_{C} f(\tau_1)\mathrm{d}\tau_1,
$$
where $\mathcal{P}$ is any path from $0$ to a generally complex $\tau$ as long as that path does not include $\operatorname{Im}(\tau)\in\{-\infty,\infty\}$. Our results are summarized by the following

\begin{proof}
First, we define a general product formula the following way
\begin{equation}
\mathscr{S}(\tau):=\prod_{\upsilon=1}^{\Upsilon}\prod_{m=1}^{m}e^{i \tau a_{(\upsilon,m)}H_{\pi_{\upsilon}(m)}},\label{eq:pf}
\end{equation}
where the coefficients $a_{(\upsilon,m)}$ are real numbers. The parameter $\Upsilon$ denotes the number of \emph{stages} of the product formula. For example, for the ST formula $S_{2k}(t)$, we have $\Upsilon=2\cdot 5^{k-1}$. The permutation $\pi_{\upsilon}$ controls the ordering of operator summands within stage $\upsilon$ of the product formula. We then consider the error in the effective Hamiltonian with the help of the definition of the path-ordered exponential:
\begin{equation}
\mathscr{S}(\tau)=\exp_{\mathcal{P}}\bigg(i\int\mathrm{d}\tau \big(H+\mathscr{E}(\tau)\big)\bigg).
\end{equation}
Using this definition and \Cref{lem:Magnus_error}, we can write the operator error on the Hamiltonian the following way:
\begin{equation}
\left\|\tilde{H}_z(t)-H\right\|\leq\frac{5}{2}\bigg(\int^{|\tau|}_{0} \left\|\mathscr{E}(\tau_1)\right\| \mathrm{d}s \bigg) /|\tau|.
\end{equation}
where $s$ in this context is $|\tau_1|$.
Let $\succeq$ be a lexicographic, or dictionary, ordering on tuples $(\upsilon,m)$, i.e.  $(\upsilon,m)\succeq(\upsilon',m')$ if $\upsilon> \upsilon'$, or if $\upsilon=\upsilon'$ and $m\geq m'$. Let
\begin{align}
\mathscr{S}(\tau)=\prod_{(\upsilon,m)}^{\longleftarrow}e^{i\tau a_{(\upsilon,m)}H_{\pi_{\upsilon}(m)}},
\end{align}
be a product formula with $\prod_{(\upsilon,m)}^{\longleftarrow}$ dictating that the operations are applied from right to left in ascending order according to $\succeq$. Let $\prod_{(\upsilon,m)}^{\longrightarrow}$ refer to descending order. After differentiating the evolution operator, and some algebra, one obtains that $\mathscr{E}(\tau)$ can be expressed as
\begin{equation}
\mathscr{E}(\tau)=\sum_{(\upsilon,m)}\prod_{(\upsilon',m')\succ(\upsilon,m)}^{\longleftarrow}e^{i\tau a_{(\upsilon',m')}H_{\pi_{\upsilon'}(m')}}\big(a_{(\upsilon,m)}H_{\pi_{\upsilon}(m)}\big)\prod_{(\upsilon',m')\succ(\upsilon,m)}^{\longrightarrow}e^{-i\tau a_{(\upsilon',m')}H_{\pi_{\upsilon'}(m')}}-H.
\label{eq:Etau}
\end{equation}
We can bound the norm of this term through the following Taylor expansion of the term 
\begin{equation}\label{eq:ABA_expans}
\begin{aligned}
&   e^{i\tau A_s}\cdots e^{i\tau A_{2}}e^{i\tau A_{1}}Be^{-i\tau A_{1}}e^{-i\tau A_{2}}\cdots e^{-i\tau A_s} \\
    &=C_0+C_1\tau+\cdots+C_{p-1}\tau^{p-1}\\
&\quad + i^{p}\sum_{k=1}^{s}\sum_{\substack{q_1+\cdots+q_k=p\\q_k\neq 0}}
e^{i\tau A_s}\cdots e^{i\tau A_{k+1}}\\
&\qquad\qquad\qquad\qquad \cdot\int_{C}\mathrm{d}\tau_2\ e^{i\tau_2 A_k}\ad_{A_k}^{q_k}\cdots\ad_{A_1}^{q_1}(B)e^{-i\tau_2 A_k}\cdot\frac{(\tau-\tau_2)^{q_k-1}\tau^{q_1+\cdots+q_{k-1}}}{(q_k-1)!q_{k-1}!\cdots q_1!}\\
&\qquad\qquad\qquad\qquad \cdot e^{-i\tau A_{k+1}}\cdots e^{-i\tau A_s},
\end{aligned}
\end{equation}
where $\ad_{A}\left(B\right):=[A,B]$. We can bound the norm of that last term in \Cref{eq:ABA_expans} through triangle inequality for contour integrals, and for sake of simplicity, we fix contour $C$ going straight from $0$ to $\tau$. We also assume for the moment that matrices $A_j$ are positive definite. Thus, the upper bound becomes:
\begin{align}
& \sum_{k=1}^{s}\sum_{\substack{q_1+\cdots+q_k=p\\q_k\neq 0}} \prod^{s}_{j=k+1} \overbrace{\| e^{i\tau A_{j}} \|  \| e^{-i\tau A_{j}} \| }^{\leq \, e^{ 2|\text{Im}\left(\tau\right)|\, \| A_j \|}}
\frac{\abs{\tau}^{q_1+\cdots+q_{k-1}}}{(q_k-1)!q_{k-1}!\cdots q_1!} 
\norm{\ad_{A_k}^{q_k}\cdots\ad_{A_1}^{q_1}(B)} \cr
&\qquad\qquad \int_{C} \mathrm{d}\tau_2 |\tau-\tau_2|^{q_k-1} \overbrace{\|e^{i\tau_2 A_k}\| \|e^{-i\tau_2 A_k}\|}^{\leq \, e^{2|\text{Im}\left(\tau_2\right)|\, \| A_k \|}} \\
&\leq \sum_{k=1}^{s}\sum_{\substack{q_1+\cdots+q_k=p\\q_k\neq 0}}
\binom{p}{q_1\ \cdots\ q_k}\frac{\abs{\tau}^p}{p!}
\norm{\ad_{A_k}^{q_k}\cdots\ad_{A_1}^{q_1}(B)} e^{2|\tau| \, \sum^{s}_{j=k} \| A_j \|} \\
& \leq\sum_{q_1+\cdots+q_s=p}
\binom{p}{q_1\ \cdots\ q_s}\frac{\abs{\tau}^p}{p!}
\norm{\ad_{A_s}^{q_s}\cdots\ad_{A_1}^{q_1}(B)} e^{2|\tau| \, \sum^{s}_{j=1} \| A_j \|} \\
&= \acomm\big(A_s,\ldots,A_1,B\big)\frac{\abs{\tau}^p}{p!} e^{2|\tau| \, \sum^{s}_{j=1} \| A_j \|},
\end{align}
where
\begin{equation}
\acomm\big(A_s,\ldots,A_1,B\big):=\sum_{q_1+\cdots+q_s=p}\binom{p}{q_1\ \cdots\ q_s}\norm{\ad_{A_s}^{q_s}\cdots\ad_{A_1}^{q_1}(B)}.
\end{equation}
Thus, the upper bound on $\norm{\mathscr{E}(\tau)}$ is
\begin{align}
\norm{\mathscr{E}(\tau)} \leq \sum_{(\upsilon,m)} \acomm\big(a_{(\Upsilon,m)}H_{\pi_{\Upsilon}(m)},\ldots,a_{(\upsilon,m+1)}H_{\pi_{\upsilon}(m+1)},a_{(\upsilon,m)}H_{\pi_{\upsilon}(m)}\big) \cr
\cdot\frac{\abs{\tau}^p}{p!} e^{2|\tau| \, \sum_{(\upsilon',m')\succ(\upsilon,m)} \left\| a_{(\upsilon',m')}H_{\pi_{\upsilon'}(m')} \right\|},
\end{align}
where $(\upsilon,m+1)=(\upsilon+1,1)$.
With this, the upper bound in $\| \tilde{H}_z(t) - H \|$
\begin{align}
    \norm{H - \tilde{H}_z} \leq \frac{5}{2}\sum_{(\upsilon,m)} \acomm\big(a_{(\Upsilon,m)}H_{\pi_{\Upsilon}(m)},\ldots,a_{(\upsilon,m+1)}H_{\pi_{\upsilon}(m+1)},a_{(\upsilon,m)}H_{\pi_{\upsilon}(m)}\big) \cr
\cdot\frac{\abs{\tau}^p}{(p+1)!} e^{2|\tau| \, \sum_{(\upsilon',m')\succ(\upsilon,m)} \left\| a_{(\upsilon',m')}H_{\pi_{\upsilon'}(m')} \right\|} \cr
\leq \frac{5}{2}\sum_{(\upsilon,m)} \acomm\big(a_{(\Upsilon,m)}H_{\pi_{\Upsilon}(m)},\ldots,a_{(\upsilon,m+1)}H_{\pi_{\upsilon}(m+1)},a_{(\upsilon,m)}H_{\pi_{\upsilon}(m)}\big) \cr
\cdot\frac{\abs{\tau}^p}{(p+1)!} e^{2|\tau| \, \sum^{m}_{j} \left\| H_{j} \right\|}
\end{align}
\end{proof}

\subsection{Avoiding fractional queries}

For the Gaussian QPEA defined in~\Cref{subsubsec:gqpea}, we must define the $\tilde{U}_{s_k}$ operator. A first attempt looks the following way
\begin{align}\label{eq:Uk}
\tilde{U}_{s_k}=\left(S_p\left(t s_k\right)\right)^{e_k}=\left(\exp\left\{-i t s_k \tilde H_{s_k}\right\}\right)^{e_k},
\end{align}
where $s_k$ are the Chebyshev nodes with $a=1$, and
\begin{align}
e_k =  \frac{s_1}{s_k},
\end{align}
such that the total simulation times 
\begin{align}
    T=s_1 t,
\end{align}
which is the same for each $\tilde{U}_{s_k}$.

This, however, will require fractional powers of (that is, a non-integer number of queries of) $S_p\left(t s_k\right)$. Although this can be efficiently achieved through quantum signal processing and one extra ancillary qubit, within this section we would like to stay as rudimentary as possible, minimizing quantum overhead. Hence, we will instead use the evolution operator
\begin{align}\label{eq:Ukprime}
\tilde{U}'_{s_k}=\left(S_p\left(t s_k\right)\right)^{e'_k}= \exp\left\{-i t s_k \mathrm{\,sgn}{s_k}\left\lceil \frac{s_1}{\left| s_k\right|} \right\rceil \tilde H_{s_k}\right\}.
\end{align}
Here,
\begin{align}\label{eq:ek}
e'_k = \sgn{s_k}\left\lceil \frac{s_1}{\left| s_k\right|} \right\rceil,
\end{align}
is an integer by construction. For this new operator $\tilde{U}'_{s_k}$, the Trotter step size remains
\begin{align}
    dt = t s_k,
\end{align}
but the total time evolution changes to
\begin{align}
    T'_k = e'_k\times dt =  t s_k \mathrm{\,sgn}{s_k}\left\lceil \frac{s_1}{\left| s_k\right|} \right\rceil = \Theta (t s_1).
\end{align}
We note that now the equivalent evolution time $T'_k$ varies for different $k$'s, but this has a negligible effect in state preparation (See \Cref{thm:gauss_state_prep}) or in eigen energy estimation (See \Cref{thm:gauss_phase_est}) with respect to the effective Hamiltonian $\tilde{H}_{s_k}$.

When it comes to query cost, this has the effect of increasing the resolution and the cost of QPEA by a factor of
\begin{align}
\frac{|e'_k|}{|e_k|}\leq 2.
\end{align}

As a last step before jumping into the proof of \Cref{thm:gauss_state_prep} we will first discuss some estimates for the coefficients of the terms in the extrapolation.  
The first such result is given in In the framework we propose here, most algorithms' cost will scale proportionally to $\| e' \|_1\leq 2\|e\|_1$.

\subsection{Proof of~\Cref{thm:gauss_state_prep}}
\label{app:proof_state_prep}

\begin{proof}[Proof of~\Cref{thm:gauss_state_prep}]
In the framework we propose here, most algorithms' cost will scale from~\Cref{lem:Trotter_step_bound}
\begin{align}\label{eq:dk_bound}
\norm{e'}_1\leq 2 \| e \|_1 = 2s_1\sum^{n}_{k=1} \frac1{\left\vert s_k \right\vert}= O\left( n \log{n} \right).
\end{align}
For example, the ground state preparation methods detailed in~\cite{rendon2022effects} would scale proportionally to $\norm{e}_1$. More precisely, the cost of state preparation scales with the minimum eigenvalue gap as
\begin{align}
    C_{\rm prep} = O\left( \frac{\|e'\|_1}{t \min_k\left( {E}_1(s_k) - {E}_0(s_k)\right)}\right),
\end{align}
where, from \Cref{lem:BauerFike} and the fact that the spectral gap is bounded below by $\gamma_0$, we have that
%
% \begin{align}
% \gamma_0_{\rm min} \geq \gamma_0_{t,0}\min_k\left( {E}_1(s_k) - {E}_0(s_k)\right) \geq {\gamma_0} \gamma_0_{t,0} - \max_{x\in I_{-1,1} }{\| H - H_{\rm eff}(x)  \|} \cr
% \geq \gamma_0 \gamma_0_{t,0} - O\left(\gamma_0_{t,0}^{p+1}\right).
% \end{align}
\begin{align}
\min_k\left( {E}_1(s_k) - {E}_0(s_k)\right) \geq {\gamma_0}  - \max_{t_1 \leq t,s\leq 1 }{\| H - \tilde H_s (t_1)  \|} \cr
\geq \gamma_0 - O\left(t^{p}\right).
\end{align}
Thus, for $t=O(\gamma_0^{1/p})$,
%
% \begin{align}
% \gamma_0_{\rm min} = \Omega({\gamma_0} \gamma_0_{t,0})
% \end{align}
\begin{align}
\min_k\left( {E}_1(s_k) - {E}_0(s_k)\right) = \Omega(\gamma_0).
\end{align}
From~\Cref{lem:radius_of_analiticity}, we also have that
\begin{align}
r = \Omega \left( \left(\frac{(p+1)!}{2^{p+1}}\right)^{1/p}\frac{\gamma_0^{1/p}}{t} \right)
\end{align}
Recalling \Cref{lem:Berns} that the interpolation error
\begin{align}
    \epsilon = O\left(\frac{1}{\rho^{M+1}}\right)
\end{align}
Thus, solving for $n$ we find:
\begin{align}
    n = O\left(\frac{\log 1/\epsilon}{\log\rho}\right)
\end{align}
We also have that the radius of the Bernstein elipse needed satisfies $\rho_\mathrm{max} = \Theta (r_\mathrm{max})$ from~\Cref{lem:radius_of_analiticity}; hence, 
\begin{align}
    n = O\left(\frac{p\log(1/\epsilon)}{\log\left(\gamma_0/t^p\right) + \log((p+1)!/2^{p+1}})\right).
\end{align}
It is clear that a constraint is
\begin{align}
t = \Omega\left(\gamma_0^{1/p}\right),
\end{align}
since we need the radius to be a quantity greater than $1$ and there is no advantage in going to a much smaller $t$.
Finally, the cost for state preparation is
\begin{align}
    C_{\rm prep} = O\left(\frac{\log(1/\epsilon)\log\log(1/\epsilon)}{\gamma_0^{1+1/p}} \right).
\end{align}
\end{proof}

\subsection{Proof of~\Cref{thm:gauss_phase_est}}

\begin{proof}[Proof of~\Cref{thm:gauss_phase_est}]
Now, the cost of estimating the energy through GPE would, naively speaking, also scale like $O\left( n \log{n} \right)$. However, to get a tighter cost bound, we will allow the variances of observables to vary across nodes and then minimize the cost function:
\begin{align}\label{eq:L0}
L_0=\sum^{n/2}_{k=1} \frac1{\sigma_k \cos{\frac{2k-1}{2n}\pi}}
\end{align}
with the constraint
\begin{align}
2\sum^{n/2}_{k=1} c^2_k\sigma^2_k-\sigma_{P}^2=0.
\end{align}
Thus, we the cost function with Lagrange multiplier is:
\begin{align}
L=L_0+\lambda \left( 2\sum^{n/2}_{k=1} c^2_k\sigma^2_k-\sigma_{P}^2 \right).
\end{align}
Note that the terms in the cost function $L_0$ go like $1/\sigma_k$ due to Fourier duality and \Cref{thm:spectral_gaussian_error} and the other factor is from the cost of a single $U_k$ implementation. After doing the minimization we are left with:
\begin{align}
L_0&=\frac{2^{1/2}}{ \sigma_{P}}\left(\sum^{n/2}_{k=1} \frac{\lvert c_k\rvert^{2/3}}{\cos^{2/3}\left(\frac{2k-1}{2n}\pi\right)}\right)^{3/2} \cr
\sigma_k&=\frac{\sigma_{P}}{2^{1/2}\lvert c_k\rvert^{2/3} \cos^{1/3}\left(\frac{2k-1}{2n}\pi\right)\left(\sum^{n/2}_{k=1} \frac{\lvert c_k\rvert^{2/3}}{\cos^{2/3}\left(\frac{2k-1}{2n}\pi\right)}\right)^{1/2}}.
\end{align}
The optimal $L_0$ goes as
\begin{align}
L_0 \in O\left(\frac{n}{\sigma_{P}}\right).
\end{align}
Now, the cost for estimating the interpolant with a variance of $\sigma_{P}$ and a bias $\epsilon$ using Gaussian phase estimation is
\begin{align}
    C_{\rm est} \in O\left(\frac{\log 1/\epsilon}{\sigma_P}\right)
\end{align}
On the other hand, using~\Cref{lem:c_norm}, if one wishes to estimate the energy using a single ancillar qubit approach, the cost of estimating the energy is
\begin{align}
    C_{\rm est,1-qubit} \in O\left(\frac{\log 1/\epsilon \log \log 1/\epsilon}{\epsilon}\right),
\end{align}
where $\epsilon$ is the semi-deterministic error coming from the Heisenberg-limited estimation algorithms like IQAE~\cite{IQAE} or the one in Ref.~\cite{HeisenbergL_Higgins}, or non-Heisenberg-limited alternatives like the semi-classical QPE~\cite{SemiclassicalQFT_1995,kitaev1995quantum}, the single-qubit version of the textbook QPE~\cite{Abrams_QPE}.

\end{proof}

\section{Cost asymptotics for expectation values
\label{app:expvals_proof}}
Here we provide proofs for Lemma~\ref{lem:interp_error_evals} and Theorem~\ref{thm:expvals_cost} as stated in Section~\ref{sec:expvals} of the main paper.

\begin{proof}[Proof of Lemma~\ref{lem:interp_error_evals}]
For scalar functions $f(s)$, derivatives of $\exp(f(s))$ can be expressed through the complete Bell polynomials via Faà di Bruno's formula.
\begin{align}
    \partial_{s}^n e^{f(s)} = Y_n(f'(s), f''(s), \dots, f^{(n)}(s)) e^{f(s)}
\end{align}
For operator exponentials such as $\exp(-i \tilde{H}_s t)$, derivatives can be expressed via repeated application of Duhamel's formula. Yet these expressions are always upper bounded by the commuting (scalar) case~\cite{watkins2022time}, so that
\begin{align}
    \norm{\partial_s^n e^{-i \tilde{H}_s t}} \leq Y_n\left(t \norm{\partial_s \tilde{H}_s}, t \norm{\partial_s^2 \tilde{H}_s}, \dots, t \norm{\partial_s^n \tilde{H}_s}\right).
\end{align}
Note that the exponential disappeared in the bound since it has norm one. Applying Lemma~\ref{lem:Heffderiv} and invoking the fact that $Y_n$ is monotonic in each argument, this is upper bounded by
\begin{align}
    Y_n \left((2 j^j (e^2c)^{j+1})_{j=1}^n\right).
\end{align}
An explicit formula for this is given by 
\begin{align}
     Y_n \left((2 j^j (e^2c)^{j+1})_{j=1}^n\right) = \sum_{D} \frac{n!}{d_1! \dots d_n!} \prod_{j=1}^n \left(\frac{2 j^j (e^2c)^{j+1}}{j!}\right)^{d_j}
\end{align}
where $D$ is a sum over all indices $(d_j)_{j=1}^n$ such that $d_j \geq 0$ and
\begin{align} \label{eq:Bell_index_sum}
    \sum_{j=1}^n d_j j = n.
\end{align}
Using a Stirling-type bound
\begin{align}
    \frac{1}{j!} \leq \left(\frac{e}{j}\right)^j \frac{1}{\sqrt{2\pi}}
\end{align}
allows us to write
\begin{align}
\begin{aligned}
    Y_n \left((2 j^j (e^2c)^{j+1})_{j=1}^n\right) &\leq \sum_{D} \frac{n!}{d_1! \dots d_n!} \prod_{j=1}^n \left(\sqrt{\frac{2}{\pi}} e^j (e^2c)^{j+1}\right)^{d_j} \\
    &= (e^3c)^n \sum_{D} \frac{n!}{d_1! \dots d_n!} \prod_{j=1}^n \left(\sqrt{\frac{2}{\pi}} e^2 c\right)^{d_j} \\
    &=(e^3c)^n Y_n(\sqrt{2/\pi} e^2 c, \sqrt{2/\pi} e^2 c, \dots, \sqrt{2/\pi} e^2 c) \\
    &= (e^3c)^n B_n (\sqrt{2/\pi} e^2 c).
\end{aligned}
\end{align}
In the second line we brought out $n$ factors of $ec$ using the condition on the indices $D$, and we identified $Y_n$ evaluated the same at every argument to be the single-variable Bell (or Touchard) polynomial $B_n$. We can bound the size of $B_n (\sqrt{2/\pi} e^2 c)$~\cite{ahle2022sharp} by
\begin{align}
    B_n(\sqrt{2/\pi} e^2 c) \leq \left(\frac{n}{\log(1 + \sqrt{\frac{\pi}{2}} n/(e^2c))}\right)^n
\end{align}
for all $n > 0$, with $n = 0$ defined by the limit (which is $1$). With this,  
\begin{align}
    \norm{\partial_s^n e^{-i \tilde{H}_s t}} \leq \left(\frac{e^3 cn}{\log(1 + \sqrt{\frac{\pi}{2}} n/(e^2c))}\right)^n \leq \left(\frac{e^3 cn}{2}\right)^n \left(1+\sqrt{\frac{8}{\pi} }\frac{e^2c}{n}\right)^n
\end{align}
where we've used the bound $1/\log(1+x) \leq 1/2 + 1/x$. Again, this inequality is valid for $n = 0$ via the limit, which is always one.

With this bound on the ST formula derivatives, we now turn to bounding $\partial_s^n O_s (t)$. Applying the binomial theorem and triangle inequality to~\eqref{eq:setup_expvals},
\begin{align} \label{eq:1st_pass_expvals}
\begin{aligned}
    \frac{\norm{\partial_s^n O_s (t)}}{\norm{O}} &\leq \sum_{p=0}^n \binom{n}{p} \norm{\partial_s^p e^{i t \tilde{H}_s}} \norm{\partial_s^{n-p} e^{-i t \tilde{H}_s}} \\
    &\leq \left(\frac{e^3c}{2}\right)^n \sum_{p=0}^n \binom{n}{p} p^p (n-p)^{n-p} \left(1+\sqrt{\frac{8}{\pi}}\frac{e^2c}{p}\right)^p \left(1+\sqrt{\frac{8}{\pi}}\frac{e^2c}{n-p}\right)^{n-p}.
\end{aligned}
\end{align}
At this point, it will be fruitful to consider two regimes. Recall that $c$ encodes information about the simulation time.

In the case where {$c > n$}, we have
\begin{align}
\begin{aligned}
    \frac{\norm{\partial_s^n O_s}}{\norm{O}} &\leq \left(\frac{e^3 c}{2}\right)^n \sum_{p=0}^n \binom{n}{p} (c + \sqrt{8/\pi} e^2 c)^p (c + \sqrt{8/\pi}e^2 c)^{n-p} \\
    &\leq \left(\frac{e^3c}{2}\right)^n c^n \left(1+\sqrt{8/\pi}e^2\right)^n \sum_{p=0}^n \binom{n}{p}  \\ 
    &= \left(c\sqrt{e^3(1+\sqrt{8/\pi}e^2)} \right)^{2n}.
\end{aligned}
\end{align}
This implies a relative error in the polynomial fit bounded by
\begin{align}
    \abs{E_{n-1}(0)} < \left(129\frac{c^2 a}{n}\right)^n.
\end{align}

In the case where $c\le n$, the approximation
\begin{align}
    \left(1+e^2\sqrt{\frac{8}{\pi}} \frac{c}{p}\right)^p < e^{ce^2\sqrt{8/\pi}}
\end{align}  
holds and is not so crude. Applying this to~\eqref{eq:1st_pass_expvals}, 
\begin{align}
    \frac{\norm{\partial_s^n O_s}}{\norm{O}} &\leq \left(\frac{e^3 c}{2}\right)^n n! \sum_{p=0}^n \frac{p^p}{p!} \frac{(n-p)^{n-p}}{(n-p)!} e^{4ce^2\sqrt{2/\pi}}.
\end{align}
Regrouping and employing a Stirling bound where appropriate,
\begin{align}
\begin{aligned}
    \frac{\norm{\partial_s^n O_s}}{\norm{O}} &\leq e^{4ce^2\sqrt{2/\pi}} \left(\frac{e^3c}{2}\right)^n n! \left(2 \frac{e^n}{\sqrt{2\pi n}}+\frac{e^n}{2\pi} \sum_{p=1}^{n-1} \frac{1}{\sqrt{p(n-p)}}\right) \\
    &\leq  e^{4ce^2\sqrt{2/\pi}} \left(\frac{e^4 c}{2}\right)^n n! \left(\frac{2}{\sqrt{2\pi n}}+\frac{\sqrt{n-1}}{2\pi}\right) \\
    &\leq \frac{1}{2\pi}(\sqrt{8\pi} +  \sqrt{n-1}) \left(\frac{e^4 c}{2}\right)^n n! e^{4ce^2\sqrt{2/\pi}} \\
    &\leq \sqrt{\frac{2n}{\pi}}\left(\frac{e^4 c}{2}\right)^n n! e^{4ce^2\sqrt{2/\pi}}
\end{aligned}
\end{align}
After another Stirling bound, this gives a corresponding interpolation error of 
\begin{align}
    E_{n-1}(0)< 2 \sqrt2 n \left(6 ca\right)^n  e^{24 c}.
\end{align}

\end{proof}

\begin{proof}[Proof of Theorem~\ref{thm:expvals_cost}]
Let $f(s) = \langle O_s (t)\rangle/\norm{O}$ be the normalized expectation value under Trotter evolution. Our interpolation algorithm produces an estimate $\Bar{f}$ of $f(0)$ which we require to be accurate within $\epsilon$.
\begin{align}
    \abs{f(0) - \Bar{f}} \leq \epsilon
\end{align}

There is the interpolation error from the polynomial $P_{n-1} f$ fitting $f$ assuming perfect interpolation points $(s_i, f(s_i))$. But $f(s_i)$ can only be estimated; let's call $\tilde{y}_i$ this estimate. The error in $\tilde{y}_i$ in our analysis comes from the statistical error inherent in the estimation protocol as well as the error in the fractional query procedure for a $1/s$ evolution. We can independently consider the interpolation error and the data error via the triangle inequality.
\begin{align}
\begin{aligned}
    \abs{f(0) - \Bar{f}} &\leq \abs{f(0) - P_{n-1} f (0)} + \abs{P_{n-1}f(0) - \tilde{P}_{n-1} f(0)} \\
    &\leq \epsilon_{\rm int} + L_n \epsilon_{\rm data}
\end{aligned}
\end{align}
Here $L_n$ is the Lebesgue constant of the interpolation, essentially a condition number, and $\epsilon_{\rm data}$ is an upper bound on the error in the data. $\tilde{P}_{n-1}f$ is the fit to the imperfect data and $P_{n-1}f$ the fit to the perfect data $(s_i, f(s_i))$. For generic interpolation nodes, $L_n$ can grow rapidly; however, for the set of Chebyshev nodes we obtain a near-optimal value~\cite{rivlin2020chebyshev}.
\begin{align}
    L_n \leq \frac{2}{\pi} \log(n+1) + 1
\end{align}
Since we want the total error to be within a threshold $\epsilon$, we can require
\begin{align}
    \epsilon_{\rm data} \leq \frac{\epsilon}{2 L_n}, \quad \epsilon_{\rm int} \leq \frac{\epsilon}{2}.
\end{align}

Given these error bounds, we can now turn to the cost of acquiring the data points. Because $O/\gamma$ can be block encoded, the expectation value calculation can be encoded as an amplitude estimation problem. Specifically, a Hadamard test circuit gives the amplitude
\begin{align}
    \frac{1+\langle O_{s_i} (t)\rangle/\gamma}{2}.
\end{align}
If we estimate this amplitude to within accuracy $\epsilon_\mathrm{data} \norm{O}/2\gamma$, we can estimate $f(s_i)$ within $\epsilon_\mathrm{data}$. Using Iterative Quantum Amplitude Estimation~\cite{IQAE}, we can obtain this estimate using a Grover iterate $G$ constructed from two Hadamard test oracles. The number of Grover oracles $N_G$ required is given by
\begin{align} \label{eq:IQAE_bound}
    N_G \leq \frac{200 \gamma L_n}{\norm{O}\epsilon_\mathrm{data}} \log\left(\frac{2n}{\delta} \log_2 \left(\frac{\gamma L_n \pi}{\norm{O}\epsilon_\mathrm{data}}\right)\right)
\end{align}
to ensure probability $1-\delta$ of all data being within $\epsilon_\mathrm{data}$ of the true value. Each $G$ requires two Hadamard tests, and each Hadamard oracle calls a (controlled) ST evolution once. The number of controlled exponentials needed for a single data point at value $s_i$ is in
\begin{align}
    O\left(\frac{N_k}{\abs{s_i}} \log 1/\epsilon_\mathrm{data}\right)
\end{align}
where $N_k = 2m 5^{k-1}$, and where the logarithm comes from the need for fractional queries with QSVT. There is also a $O(1)$ overhead associated with the fractional queries and IQAE. Altogether, the number of exponentials for a single data point is in 
\begin{align}
    O\left(N_G \times 2 \times \frac{N_k}{\abs{s_i}} \log 1/\epsilon_\mathrm{data}\right).
\end{align}
Therefore, the total number $N_\mathrm{exp}$ of generating all $n/2$ data points (we only need half due to symmetry) is in
\begin{align} \label{eq:Cdata_def}
    N_\mathrm{exp} \in O\left(N_G N_k \sum_{i=1}^{n/2} \frac{1}{s_i} \log(1/\epsilon_\mathrm{data})\right).
\end{align}
Plugging in~\eqref{eq:IQAE_bound} for $N_G$ above and summing over $1/s_i$ using Lemma \ref{lem:Trotter_step_bound},
\begin{align} \label{eq:Cdata_pre_a}
\begin{aligned}
    N_\mathrm{exp} &\in O\left(\frac{ \gamma N_k L_n n }{ \norm{O} \epsilon a} (\log n)\log\left(\frac{2n}{\delta}\log_2\left(\frac{\gamma L_n \pi}{\norm{O}\epsilon}\right)\right)\log(1/\epsilon_\mathrm{data})\right) \\
    &\subset \tilde{O}\left(\frac{n}{a\epsilon}\log1/\delta\right)
\end{aligned}
\end{align}
where $\tilde{O}$ suppresses factors logarithmic in $n$ and $\epsilon$. We also employed our assumption that $\gamma/\norm{O} \in O(1)$. The number of nodes $n$ and the interpolation interval $[-a,a]$ will be determined by $\epsilon_{\rm int}$, the interpolation error assuming perfect data. To apply our error bounds from the previous subsection, choose $a$ to satisfy Lemma~\ref{lem:Heffderiv}, i.e. $ca < \pi/20$, while also taking $1/a \in O(c)$.

Choose $n \geq \lceil c \rceil$. Then the second bound of Lemma~\ref{lem:interp_error_evals} holds. From the interpolation error, we must satisfy
\begin{align}
    2\sqrt2 n (6ca)^n e^{24c} < \epsilon/2
\end{align}
which in turn can be satisfied provided that 
\begin{align}
    4\sqrt2 n \left(6e^{24} ca\right)^n < \epsilon
\end{align}
since $n \geq c$. Choose $a$ such that $6e^{24}ca = 1/2$, which is consistent with our previous conditions on $a$. Then, to satisfy the error bound, $n$ can satisfy
\begin{align}
    4\sqrt2 n 2^{-n} < \epsilon.
\end{align}
This can be solved using the $-1$ branch of the $\mathrm{LambertW}$ function. 
\begin{align}
    n > - \frac{\mathrm{LambertW}_{-1}\left(-\epsilon \log 2/(4\sqrt2))\right)}{\log 2}.
\end{align}
The appropriate asymptotics is $n \in O(\log(1/\epsilon)$. By taking $n = n_*$ where
\begin{align}
    n_* &= \max\left\{\lceil c \rceil, \left\lceil- \frac{\mathrm{LambertW}_{-1}(-\epsilon \log 2/(4\sqrt2)}{\log 2}\right\rceil\right\} \\
    &\in O(\max\{c, \log(1/\epsilon)\})
\end{align}
we satisfy all required constraints and arrive at our final asymptotic scaling.
\begin{align}
    N_\mathrm{exp} \in \tilde{O}\left(\max\{c, \log(1/\epsilon)\} c \epsilon^{-1} \log(1/\delta)\right)
\end{align}
\end{proof}

\begin{remark}
    In the proof above, we set $n > c$ from the beginning, in order to use the second of the two bounds from Lemma~\ref{lem:interp_error_evals}. Together with $1/a \in O(c)$ this condemns us to a suboptimal $c^2$ scaling in the large $c$ limit. However, using the first bound instead of the second would not help us, since the $nc^2/a$ term in that bound must be order one. 
\end{remark}

\end{document}